\documentclass[11pt]{article}
\usepackage[margin=1in]{geometry}
\usepackage{fullpage}
\usepackage{tgtermes}
\usepackage[T1]{fontenc}
\usepackage[hyphens]{url}
\usepackage[colorlinks,citecolor=blue,linkcolor=blue,urlcolor=red,pagebackref]{hyperref}
\usepackage{graphicx}
\usepackage{amsfonts,amsmath,amsthm,amssymb,dsfont,mathtools}
\usepackage{xfrac,nicefrac}
\usepackage{mathdots}
\usepackage{bm}
\usepackage{paralist}
\usepackage{enumerate}
\usepackage[normalem]{ulem}
\usepackage{xspace}
\xspaceaddexceptions{]\}}
\usepackage[capitalise]{cleveref}
\usepackage{comment}
\usepackage{tabu}
\usepackage{enumerate}
\usepackage{framed}
\usepackage{float,wrapfig}
\usepackage{subfigure}
\usepackage{tikz}
\usepackage{algpseudocode} 
\usepackage[shortlabels]{enumitem}
\usepackage[usenames,dvipsnames]{pstricks}
\usepackage[linesnumbered,boxed,ruled,vlined]{algorithm2e}
\usepackage{thm-restate}

\theoremstyle{plain}

\newtheorem{theorem}{Theorem}[section]
\newtheorem{lemma}[theorem]{Lemma}

\newtheorem{hypo}[theorem]{Hypothesis}
\newtheorem{claim}[theorem]{Claim}
\newtheorem{obs}[theorem]{Observation}

\theoremstyle{definition}
\newtheorem{definition}[theorem]{Definition}

\newtheorem{property}[theorem]{Property}

\renewcommand{\epsilon}{\varepsilon}
\newcommand{\eps}{\varepsilon}
\renewcommand{\Pr}{\operatorname*{\mathbf{Pr}}}
\newcommand{\Ex}{\operatorname*{\mathbf{E}}}

\newcommand{\poly}{\operatorname{\mathrm{poly}}}
\newcommand{\tr}{\operatorname{\mathrm{tr}}}
\newcommand{\polylog}{\poly\log}
\newcommand{\F}{\mathbb{F}}
\newcommand{\R}{\mathbb{R}}

\newcommand{\C}{\mathbb{C}}

\newcommand{\Z}{\mathbb{Z}}
\renewcommand{\tilde}{\widetilde}
\newcommand{\tO}{\tilde{O}}

\newcommand{\caA}{\mathcal{A}}

\newcommand{\caC}{\mathcal{C}}

\newcommand{\caH}{\mathcal{H}}

\DeclarePairedDelimiterX\braket[2]{\langle}{\rangle}{#1\,\delimsize\vert\,\mathopen{}#2}

\DeclareMathOperator*{\argmax}{arg\,max}

\title{Removing Additive Structure in $3$SUM-Based Reductions}
\author{
Ce Jin\thanks{cejin@mit.edu}\\MIT  \and Yinzhan Xu\thanks{xyzhan@mit.edu}\\MIT
}

\begin{document}

	\setcounter{page}{0} \clearpage
	\maketitle
	\thispagestyle{empty}

	\begin{abstract}

Our work explores the hardness of $3$SUM instances without certain additive structures, and its applications. As our main technical result, we show that solving $3$SUM on a size-$n$ integer set that avoids solutions to $a+b=c+d$ for $\{a, b\} \ne \{c, d\}$ still requires $n^{2-o(1)}$ time, under the $3$SUM hypothesis. Such sets are called Sidon sets and are well-studied in the field of additive combinatorics.
	
\begin{itemize}
\item 
Combined with previous reductions, this implies that the All-Edges Sparse Triangle problem on $n$-vertex graphs with maximum degree $\sqrt{n}$ and at most $n^{k/2}$ $k$-cycles for every $k \ge 3$ requires $n^{2-o(1)}$ time, under the $3$SUM hypothesis.  This can be used to strengthen the previous conditional lower bounds by Abboud, Bringmann, Khoury, and Zamir [STOC'22] of $4$-Cycle Enumeration, Offline Approximate Distance Oracle and Approximate Dynamic Shortest Path. In particular, we show that no algorithm for the $4$-Cycle Enumeration problem  on $n$-vertex $m$-edge graphs with $n^{o(1)}$ delays has $O(n^{2-\varepsilon})$ or $O(m^{4/3-\varepsilon})$ pre-processing time for $\varepsilon >0$. We also present a matching upper bound via simple modifications of the known algorithms for $4$-Cycle Detection. 

\item 
A slight generalization of the main result also extends the result of Dudek, Gawrychowski, and Starikovskaya [STOC'20] on the $3$SUM hardness of nontrivial 3-Variate Linear Degeneracy Testing (3-LDTs): we show $3$SUM hardness for all nontrivial 4-LDTs.
\end{itemize}
	
The proof of our main technical result combines a wide range of tools: Balog-Szemer{\'e}di-Gowers theorem, sparse convolution algorithm, and a new almost-linear hash function with almost $3$-universal guarantee for integers that do not have small-coefficient linear relations.
	
\end{abstract}
\newpage

	\setcounter{page}{0} \clearpage
\tableofcontents{}
	\thispagestyle{empty}
\newpage

\section{Introduction}
Fine-grained complexity theory provides conditional lower bounds for a wide range of problems, by designing fine-grained reductions from a few central problems that are hypothesized to be hard (see e.g.~\cite{virgisurvey});
specifically, a fine-grained reduction from a central problem $A$ to some problem $B$ of interest would establish a conditional lower bound for $B$ based on the hardness of $A$.
Sometimes, certain structured classes of inputs already capture the full hardness of problem $A$, formally shown by a fine-grained reduction from arbitrary instances of $A$ to structured instances of $A$.
Such kind of results would be extremely productive for proving conditional lower bounds: having structures in $A$ makes it much easier to design reductions from $A$ to $B$. 

A famous example is the equivalence between $3$SUM and the (seemingly easier) $3$SUM Convolution problem \cite{Patrascu10, ChanH20}.
In the $3$SUM problem, we need to determine if a set of $n$ integers contains three integers that sum up to $0$.\footnote{Another equivalent variant of $3$SUM is its tripartite version, in which we are given three sets and need to determine if there are three numbers, one from each set, that sum up to $0$. In this work, we use $3$SUM to refer to the one set version by default.} 
The $3$SUM Convolution problem  essentially can be thought of as $3$SUM with the additional property that the $n$ input integers have distinct remainders modulo $n$.\footnote{In a more popular definition of $3$SUM Convolution, we are given an integer array $A$ indexed by, say, $\{0 ,\ldots, n-1\}$, and the goal is to decide whether there exist $i, j$ such that $A[i]+A[j]+A[i+j]=0$. Clearly, it is equivalent to $3$SUM on the set $\{2n \cdot A[i] + i\}_{i=0}^{n-1}$, and this set has the aforementioned property.} 
The results in \cite{Patrascu10, ChanH20} 
established that $3$SUM Convolution requires essentially quadratic time, under the hypothesis that $3$SUM requires essentially quadratic time (which is a central hypothesis in fine-grained complexity called the $3$SUM hypothesis).
It turned out that this extra structure makes it easier to design fine-grained reductions from $3$SUM Convolution, leading to 
tight conditional lower bounds for problems such as Triangle Listing~\cite{Patrascu10} and Exact Triangle~\cite{williams2013finding}, under the $3$SUM hypothesis.

Dudek, Gawrychowski and  Starikovskaya \cite{ldt} showed that $3$SUM is subquadratically equivalent to all nontrivial 3-Variate Linear Degeneracy Testing (3-LDT). In particular, they showed that $3$SUM is equivalent to the AVERAGE problem, in which one needs to determine whether a given set of $n$ integers contains a $3$-term arithmetic progression involving distinct numbers, i.e., three distinct numbers $a, b, c$ where $a-2b+c = 0$.
In more details, the reduction from AVERAGE to $3$SUM was known earlier \cite{Erickson99}. The reduction from $3$SUM to AVERAGE first goes through a structured version of tripartite AVERAGE, in which there are three given arrays $A, B, C$, and we must have $a \in A, b \in B, c \in C$; their structure is that all of $A, B, C$ are $3$-AP free (i.e., do not contain $3$-term arithmetic progressions involving distinct numbers). It is simple to show that this problem is equivalent to the tripartite version of $3$SUM where each array is $3$-AP free, by scaling appropriately. This is thus another example of (tripartite) $3$SUM that is hard on structured inputs, and it (implicitly) helps showing the equivalence between $3$SUM and AVERAGE.

In light of generalizing this result, it is natural to ask whether $3$SUM is still hard on inputs without a certain equation involving $4$ numbers. One particularly interesting equation involving $4$ numbers is $a+b-c-d = 0$, and a set without nontrivial solutions to $a+b-c-d = 0$ is called a Sidon set (also known as Golomb ruler). Here, a solution is nontrivial if $\{a, b\} \ne \{c, d\}$. Sidon sets are extensively studied in the field of additive combinatorics (e.g. see the survey \cite{o2012complete}) and are also mentioned explicitly in the conference talk \cite{ldt-talk} as a barrier for generalizing \cite{ldt}'s results.

As mentioned, \cite{ldt}'s reduction from $3$SUM to AVERAGE  goes through a version of tripartite AVERAGE in which all three arrays are $3$-AP free. They achieve this by partitioning each input array of an unstructured  tripartite AVERAGE to a subpolynomial number of $3$-AP free sub-arrays, and then  solve AVERAGE on each triple of these $3$-AP free sub-arrays. Such a partitioning is only possible because of the existence of $3$-AP free subsets of $[U]$ of sizes $U^{1-o(1)}$ \cite{behrend}. On the contrary, all Sidon subsets of $[U]$ have sizes at most $\sqrt{U} + O(U^{1/4})$ \cite{ErdosT41}, which is too small to apply \cite{ldt}'s technique. 

Thus, previously there was no answer for the following natural question: 

\begin{quote}
\centering
\em Question 1: Does $3$SUM on Sidon sets require $n^{2-o(1)}$ time under the $3$SUM hypothesis?
\end{quote} 

Recently, a work by Abboud, Bringmann, Khoury,  and Zamir \cite{AbboudBKZ22} shows another example of fine-grained hardness on structured problems. Using their ``short cycle removal technique'', they were able to show hardness of certain Triangle Detection problems in graphs with few $k$-cycles. In particular, they showed that detecting whether an $m$-edge $4$-cycle free graph has a triangle requires $\Omega(m^{1.1194})$ time, assuming Triangle Detection on $n$-vertex graphs with  maximum degree at most $\sqrt{n}$ requires $n^{2-o(1)}$ time. Here, the structure of the input is $4$-cycle freeness. Their technique is also able to provide conditional lower bounds under more standard hypotheses. In particular, they showed that detecting whether each edge is in a triangle (All-Edges Sparse Triangle) on a graph with maximum degree $\sqrt{n}$ and $O(n^{2.344})$ $4$-cycles (or more precisely, $O(n^{\frac{\omega+7}{4} + \eps})$ $4$-cycles for any $\eps > 0$, where $\omega < 2.37286$ \cite{alman2021refined} is the square matrix multiplication exponent)  requires $n^{2-o(1)}$ time, assuming that All-Edges Sparse Triangle on a graph with maximum degree $\sqrt{n}$ requires $n^{2-o(1)}$ time. As the assumption is known to hold under either the $3$SUM hypothesis or the APSP hypothesis~\cite{Patrascu10, williams2020monochromatic},  the lower bound holds under these two central hypotheses in fine-grained complexity as well. This lower bound (and its more general version for $k$-cycle) has a variety of applications, including the hardness for Approximate Offline Distance Oracles,  Approximate Dynamic Shortest Path, and $k$-Cycle Enumeration. 

However, it is hard to imagine that $n^{2.344}$ $4$-cycles are indeed the smallest amount of $4$-cycles on a graph with maximum degree $\sqrt{n}$, so that All-Edges Sparse Triangle still requires $n^{2-o(1)}$ time. For instance, consider random graphs with maximum degree $\sqrt{n}$. It is unclear how the current best $O(m^{2\omega/(\omega+1)})$ time algorithm  for All-Edges Sparse Triangle \cite{AlonYZ97}, or the brute-force $O(n^2)$ time algorithm that enumerates all pairs of neighbors of each vertex, 
can exploit the randomness of the graph. To the best of our knowledge, $O(n^2)$ time is still the best running time
for such random graphs, even if perfect matrix multiplication exists (i.e. $\omega = 2$). However, in such random graphs, the expected number of $4$-cycles is only $O(n^2) \ll n^{2.344}$. 
It is thus natural to ask whether All-Edges Sparse Triangle on graphs with maximum degree $\sqrt{n}$ and fewer than $n^{2.344}$ $4$-cycles is hard. 

\begin{quote}
\centering
\em Question 2: Does All-Edges Sparse Triangle on a graph with maximum degree $\sqrt{n}$ and fewer $4$-cycles still require $n^{2-o(1)}$ time?
\end{quote} 

We affirmatively answer both Question 1 and Question 2.
Though not obvious, Question 1 and Question 2 are actually strongly related. As we will show later, an affirmative answer to Question 1 actually implies an affirmative answer to Question 2. Our work thus also connects the previous two seemingly unrelated directions of research \cite{ldt, AbboudBKZ22}.

\subsection{Our Results}
\paragraph{$3$SUM on Sidon Set.} 
As our main result, we show that $3$SUM on Sidon sets is indeed hard, resolving Question 1. 

\begin{restatable}{theorem}{ThreeSUMSidon}
\label{thm:main}
Under the $3$SUM hypothesis, for all constants $\delta>0$, $3$SUM on size-$n$ Sidon sets of integers bounded by $[-n^{3+\delta},  n^{3+\delta}]$ requires $n^{2-o(1)}$ time. 
\end{restatable}

Our techniques differ from \cite{ldt}'s techniques in significant ways. As a high level overview, our reduction combines the celebrated Balog-Szemer{\'e}di-Gowers Theorem \cite{balog1994statistical,gowers2001new} and efficient sparse convolution algorithm~\cite{ColeH02,ArnoldR15,ChanL15,sparse4,sparse3,Nakos20,nfold,sparse2,BringmannFN22} to solve $3$SUM instances on sets with very high additive energy (many tuples $(a, b, c, d)$ with $a + b = c + d$) in truly subquadratic time. On the other hand, for $3$SUM instances on sets with moderately low additive energy, we modify known self-reductions of $3$SUM \cite{BaranDP08} by designing hash functions with better universality guarantee, to self-reduce such $3$SUM instances to $3$SUM instances on Sidon sets. See Section~\ref{sec:overview} for a more detailed overview. 

Given Theorem~\ref{thm:main}, it is not difficult to obtain the following corollary using techniques in \cite{ldt}. 
\begin{restatable}{cor}{}
Under the $3$SUM hypothesis, for all $\delta>0$, determining whether a given set of $n$ integers bounded by $[-n^{3+\delta},  n^{3+\delta}]$  is a Sidon set requires $n^{2-o(1)}$ time. 
\end{restatable}

More generally, we are able to show that all nontrivial $4$-LDTs are $3$SUM-hard, generalizing \cite{ldt}'s result on $3$-LDTs. A $4$-LDT is parameterized by integers $\beta_1, \beta_2, \beta_3, \beta_4, t$, and asks to determine whether a given integer set $A$ contains a  solution to $\sum_{i=1}^4 \beta_i a_i = 0$ for distinct $a_i \in A$. Following \cite{ldt}'s notation, a $4$-LDT parameterized by $\beta_1, \beta_2, \beta_3, \beta_4, t$ is called trivial if either 
\begin{enumerate}
    \item Any of $\beta_1, \beta_2, \beta_3, \beta_4$ is $0$, or
    \item $t \ne 0$ and $\gcd(\beta_1, \beta_2, \beta_3, \beta_4) \nmid t$. 
\end{enumerate}
For the first case, the $4$-LDT is degenerated to a $3$-LDT; for the second case, the answer is always NO. All other $4$-LDTs are called nontrivial.
Using their techniques, it is easy to show that a nontrivial $4$-LDT parameterized by $\beta_1, \beta_2, \beta_3, \beta_4, t$ requires $n^{2-o(1)}$ time under the $4$SUM hypothesis (which is a weaker hypothesis than the $3$SUM hypothesis) if $t \ne 0$ or $\beta_1 + \beta_2 + \beta_3 + \beta_4 \ne 0$. Therefore, we will focus on the remaining cases, i.e., $t = 0$ and $\beta_1 + \beta_2 + \beta_3 + \beta_4 = 0$.

\begin{restatable}{theorem}{FourLDT}
\label{thm:4ldt}
Fix any non-zero  integers $\beta_1, \beta_2, \beta_3, \beta_4$ where $\sum_{i=1}^4\beta_i=0$ and any real number $\delta > 0$. Determining whether a size-$n$  set of integers bounded by  $[-n^{3+\delta},  n^{3+\delta}]$ avoids  solutions $\sum_{i=1}^4 \beta_i a_i = 0$ for distinct $a_i$ in the set requires $n^{2-o(1)}$ time, assuming the $3$SUM hypothesis. 
\end{restatable}

We remark that this definition of $4$-LDT does not perfectly fit the definition of Sidon sets. When $\beta_1 = \beta_2 = 1$ and $\beta_3 = \beta_4 = -1$, a solution $2a = b+c$ is not allowed in Sidon sets as $\{a, a\} \ne \{b, c\}$, but is allowed in the definition in \cref{thm:4ldt}, as $a, a, b, c$ are not all distinct. However, it would still make sense to count $2a = b+c$ as a solution to the $4$-LDT, as such solutions do not trivially exist in all sets. 
Thus, we define a slight variant of $4$-LDT, in which we need to determine whether a size-$n$  set $A$ of integers avoids \textit{nontrivial} solutions to $\sum_{i=1}^4 \beta_i a_i = 0$ for  $a_i \in A$. Here, a solution is trivial if for every $a\in \{a_1,\dots,a_4\}$ it holds that $\sum_{i: a_i=a}\beta_i=0$. We will show that \cref{thm:4ldt} still works under this alternative definition (and all nontrivial $4$-LDTs with $t \ne 0$ or $\beta_1 + \beta_2 + \beta_3 + \beta_4 \ne 0$ are still $4$SUM hard under this definition, following previous techniques).

\paragraph{Quasirandom graph.} 
Question 2 concerns graphs with few, say $n^2$, $4$-cycles. As mentioned, in a random $n$-vertex graph with maximum degree at most $\sqrt{n}$, we expect to see $\Theta(n^2)$ $4$-cycles. Graphs in which the numbers of $4$-cycles are close to those of random graphs with the same edge density are actually well-studied in additive combinatorics, and such graphs are called pseudorandom graphs (see e.g. \cite{zhao2019graph}). In additive combinatorics, a sequence of graphs $(G_n)$ with $G_n$ having $n$ vertices and $(p+o(1))\binom{n}{2}$ edges are called (sparse) pseudorandom graphs if the number of labeled $4$-cycles\footnote{The number of labeled subgraph $H$ in a graph $G$ is the number of injective graph homomorphisms from $H$ to $G$.} in $G_n$ is at most $(1+o(1))p^4 n^4$. 
We adapt this terminology as follows:

\begin{restatable}[Quasirandom Graph]{definition}{Quasirandom}
\label{def:quasirandom}
An undirected unweighted $n$-vertex graph  is called a quasirandom graph if it has maximum degree at most $\sqrt{n}$ and has at most $n^{2}$ $4$-cycles.  
\end{restatable}

As an application of \cref{thm:main}, we show that All-Edges Sparse Triangle is still hard even on quasirandom graphs, answering Question 2 affirmatively. 

\begin{restatable}{theorem}{AllEdgesLowerBound}
\label{thm:ae-lb}
Under the $3$SUM hypothesis, All-Edges Sparse Triangle on $n$-vertex quasirandom graphs requires $n^{2-o(1)}$ time.
\end{restatable}

\cref{thm:ae-lb} actually implies hardness of All-Edges Sparse Triangle on certain graphs with few  $k$-cycles for any $k\ge 3$. 

\begin{restatable}{cor}{AllEdgesLowerBoundCor}
\label{cor:ae-lb}
Under the $3$SUM hypothesis, All-Edges Sparse Triangle on $n$-vertex graphs which has maximum degree at most $\sqrt{n}$ and has at most $n^{k/2}$ $k$-cycles for every $k \ge 3$ requires $n^{2-o(1)}$ time.
\end{restatable}

\paragraph{$4$-Cycle Enumeration.} 
As \cref{thm:ae-lb} improves a result of \cite{AbboudBKZ22}, we naturally obtain improved conditional lower bounds for several problems they consider. In particular, we achieve tight conditional lower bound for the $4$-Cycle Enumeration problem.

In the $4$-Cycle Enumeration problem, we need to first pre-process a given simple graph, and then enumerate all the $4$-cycles in this graph with subpolynomial time delay for every $4$-cycle enumerated. This problem was first studied by \cite{AbboudBKZ22}, inspired by both  the classic $4$-Cycle Detection problem \cite{yuster1997finding,AlonYZ97} and the recent trend of enumeration algorithms \cite{segoufin2015constant, florenzano2018constant, carmeli2020enumeration, carmeli2021enumeration}. 

 \cite{AbboudBKZ22} showed an $m^{1+\frac{3-\omega}{2(4-\omega)} - o(1)}$ pre-processing time lower bound for $4$-Cycle Enumeration on $m$-edge graphs, under either the $3$SUM hypothesis or the APSP hypothesis. This lower bound is only $m^{5/4-o(1)}$ even if $\omega = 2$. Using \cref{thm:ae-lb}, we show the following improved lower bound. 

\begin{restatable}[$4$-Cycle Enumeration, lower bound]{theorem}{FourCycleLower}
\label{thm:4cycle_lower}
Assuming the $3$SUM hypothesis,  there is no algorithm with $O(n^{2-\eps})$ pre-processing time and $n^{o(1)}$ delay that solves $4$-Cycle Enumeration on $n$-node graphs with $m = \lfloor 0.49n^{1.5} \rfloor$ edges, for any constant $\eps>0$.
\end{restatable}

In terms of $m$, this lower bound is $m^{4/3-o(1)}$. By the same reasoning as \cite{AbboudBKZ22}, \cref{thm:4cycle_lower} and the known $3$-SUM hardness of Triangle Listing \cite{Patrascu10}
imply an $m^{4/3-o(1)}$ pre-processing lower bound for $k$-Cycle Enumeration for any $k \ge 3$, under the $3$SUM hypothesis.

Note that $m=n^{1.5\pm o(1)}$ is the hardest density for $4$-Cycle Enumeration. See \cref{sec:4-cycle-enum-ub} for more details.

It is known how to solve classic $4$-Cycle Detection in $O(\min\{n^2, m^{4/3}\})$ time \cite{yuster1997finding,AlonYZ97}. We show, by simple modifications of the existing $4$-Cycle Detection algorithms, that $4$-Cycle Enumeration can also be solved in $O(\min\{n^2, m^{4/3}\})$ pre-processing time. Thus, the conditional lower bound in \cref{thm:4cycle_lower} is indeed tight. 

\begin{restatable}[$4$-Cycle Enumeration, upper bound]{theorem}{FourCycleUpper}
\label{thm:4cycle_upper}
Given an $n$-vertex $m$-edge undirected graph,  we can enumerate $4$-cycles in $O(1)$ delay after an $O(\min\{n^2, m^{4/3}\})$ time pre-processing. The algorithm is deterministic.
\end{restatable}

\paragraph{Offline Approximate Distance Oracle and Dynamic Approximate Shortest Paths.} Another result obtained by \cite{AbboudBKZ22} is the hardness of Offline Approximate Distance Oracle. A Distance Oracle needs to pre-process a given graph, and then answer (approximate) distance between two query vertices. There exist Distance Oracles that can pre-process an  $n$-vertex $m$-edge undirected weighted graph in $\tO(mn^{1/k})$\footnote{We use $\tilde O(\cdot )$  hide poly-logarithmic factors in input size.} time and then answer $(2k-1)$-approximate distance queries  in $O(1)$ time, for any constant integer $k \ge 1$ \cite{thorup2005approximate, roditty2005deterministic, chechik2014approximate}, and it remains the current best trade-off between pre-processing time and approximate factor. It is thus a natural question to ask whether near-linear  pre-processing time is possible for some constant approximation factors.

P{\u{a}}tra{\c{s}}cu, Roditty, and Thorup \cite{patrascu2012new} showed that all Distance Oracles with $(3-\eps)$-approximation factor and constant query time must use $\Omega(m^{1.5})$ space, under a set intersection conjecture. This  implies that near-linear pre-processing time is impossible for $(3-\eps)$ approximation. \cite{AbboudBKZ22} ruled out this for all $k \ge 4$ as well. More specifically, they showed that, for any $k \ge 4, \delta > 0$, any algorithm for returning a $(k-\delta)$-approximation of the  distances between $m$ pairs of vertices (given at once) in an $m$-edge undirected unweighted graph requires $m^{1+\frac{3-\omega}{4k-(2\omega-2)}-o(1)}$ time, under either the $3$SUM hypothesis or the APSP hypothesis. When $\omega = 2$, the lower bound becomes 
$m^{1+\frac{1}{4k-2}-o(1)}$. 
This essentially establishes that $m^{1+\frac{1}{\Theta(k)}}$ pre-processing time is the correct answer, for achieving $k$-approximation and near-constant query time. However, the constant factor hidden in $\Theta(k)$ still does not match.

Combining Theorem~\ref{thm:ae-lb} and \cite{AbboudBKZ22}'s technique, we make progress towards closing this gap. 
\begin{restatable}[Offline Distance Oracles, I]{theorem}{DOLower}
\label{thm:DO}
Assuming the $3$SUM hypothesis, 
    for any constant integer $k \ge 2$ and $\eps, \delta>0$, there is no $O(m^{1+\frac{1}{2k-1} - \eps})$ time algorithm that can $(k - \delta)$-approximate the distances between $m$ given pairs of vertices in a given $n$-vertex $m$-edge undirected unweighted graph, where $m = \Theta(n^{1+\frac{1}{2k-2}})$. 
\end{restatable}
\begin{restatable}[Offline Distance Oracles, II]{theorem}{thmdo}
\label{thm:do2}
Assuming the $3$SUM hypothesis, 
    for any constant integer $k \ge 3$ and $\eps, \delta>0$, there is no
    $(k - \delta)$-approximate distance oracle with 
    $O(n^{1+\frac{1}{k - 1} - \eps})$ pre-processing time  and $n^{o(1)}$ query time for an $n$-vertex $O(n)$-edge undirected unweighted graph.
\end{restatable}
Note that \cref{thm:do2} has a higher lower bound (in terms of the number of edges) for pre-processing time, while \cref{thm:DO} applies to distance oracles with possibly slower query time $m^{1/(2k-1)-\eps}$. Compared to the Thorup-Zwick distance oracle \cite{thorup2005approximate} with $(2k-1)$ approximation,  constant query time and $O(n^{1+1/k})$ pre-processing time on $O(n)$-edge graphs, our lower bound in \cref{thm:do2} loses a factor of $2$ on the exponent for large constant $k$, while the previous bound by \cite{AbboudBKZ22} loses a factor of 8 (when $\omega = 2$).

Dalirrooyfard, Jin, Vassilevska Williams and Wein \cite{ansc} studied a similar question called $n$-Pair Shortest Paths, which is the Offline Distance Oracle problem with $n$ queries, and obtained close-to-optimal combinatorial lower bounds for algorithms achieving $(1+1/k)$-approximation. 

Another problem similar in nature to Distance Oracle is Dynamic Shortest Paths. Here, the difference is that we also need to support updates that can insert or delete an edge in the given graph. For its decremental version where only edge deletions are allowed, data structures with $\tO(n^{1/k})$ amortized update time and $\tO(1)$ query time for $O(k)$ approximation are known for weighted undirected graphs \cite{chechik2018near}. For the fully dynamic version, the best known data structure with $\tO(n^{1/k})$ amortized update time and $\tO(1)$ query time provides $O(\log n)^{O(k)}$-approximation \cite{forster2021dynamic}. 

For $(k - \delta)$-approximations, \cite{AbboudBKZ22} shows   that (when $\omega = 2$) no algorithms for Decremental Dynamic Approximate Shortest Path on undirected unweighted graphs can have both $O(m^{1+\frac{1}{4k-2}-\eps})$ total update time and $O(m^{\frac{1}{4k-2}-\eps})$ query time for positive $\eps$, under either the $3$SUM hypothesis or the APSP hypothesis. As mentioned in \cite{AbboudBKZ22}, this bound follows immediately from their lower bound for Offline Distance Oracle. Thus, our \cref{thm:DO} implies improved lower bounds under the $3$SUM hypothesis: no algorithms can have both $O(m^{1+\frac{1}{2k}-\eps})$ total update time and $O(m^{\frac{1}{2k}-\eps})$ query time. 

They also provided a conditional lower bound for Fully Dynamic Approximate Shortest Path. More specifically, they showed that (when $\omega = 2$) no algorithm can pre-process  an $n$-vertex undirected unweighted graphs in $O(n^3)$ time and supports fully dynamic updates and queries in $O(m^{\frac{1}{4k-2}-\eps})$ time, where the queries need to be approximated  within $(k - \delta)$ factor, for $\delta,\eps>0$, under either the $3$SUM hypothesis or the APSP hypothesis. Combining their approach with \cref{thm:ae-lb}, we also obtain improved lower bounds. 

\begin{restatable}[Dynamic Approximate Shortest Path]{theorem}{DynamicSP}
\label{thm:DynamicSP}
Assuming the $3$SUM hypothesis,  for any constant integer $k \ge 3$ and $\eps, \delta>0$, no algorithm can support insertion and deletion of edges and support querying $(k - \delta)$-approximate distance between two vertices in $O(m^{\frac{1}{2k-1} - \eps})$ time per update and query, after an $O(n^3)$ time pre-processing, in $n$-vertex $m$-edge undirected unweighted graphs, where $m = \Theta(n^{1+\frac{1}{2k-2}})$.  
\end{restatable}

\paragraph{All-Nodes Shortest Cycles.}

We also explore the conditional lower bound of the All-Nodes Shortest Cycles problem, which was not considered by~\cite{AbboudBKZ22}. In this problem, we are given a graph and are required to compute the length of the shortest cycle through every vertex. This problem was first considered by Yuster~\cite{YusterANSC}, who gave an $\tO(n^{(3+\omega)/2})$ time algorithm for unweighted undirected graph. It was later improved by Agarwal and Ramachandran \cite{agarwal2018fine} to $\tO(n^\omega)$ time. 
Sankowski and W{\k{e}}grzycki \cite{sankowski2019improved} showed the same $\tO(n^\omega)$ time bound for unweighted directed graphs. 

The study of the All-Nodes Shortest Cycles problem in the approximate setting was initiated by Dalirrooyfard, Jin, Vassilevska Williams and Wein \cite{ansc}, who gave various algorithms and conditional lower bounds for approximate All-Nodes Shortest Cycles. In particular, they showed an $\tO(mn^{1/k})$ time algorithm for $(k+\eps)$-approximate All-Nodes Shortest Cycles in undirected unweighted graphs, for arbitrary $\eps > 0$.

Using \cref{thm:ae-lb}, we show the following conditional lower bound, suggesting that $m^{1+\frac{1}{\Theta(k)}}$ is likely the correct running time for $k$-approximate All-Nodes Shortest Cycles.

\begin{restatable}[All-Nodes Shortest Cycles]{theorem}{ANSC}
\label{thm:intro:ansc}
Fix any integer $k\ge 4$. Assuming the $3$SUM hypothesis, 
    for any $\eps, \delta>0$, there is no $O(m^{1+1/k-\eps})$ time algorithm  that can solve the All-Nodes Shortest Cycles problem within $(k/3 - \delta)$ approximation factors on $n$-vertex $m$-edge graphs with $m=
    \Theta(n^{1+\frac{1}{k-1}})$.
\end{restatable}

\paragraph{Triangle Detection. }

As mentioned, \cite{AbboudBKZ22}  showed a conditional lower bound for Triangle Detection on a $4$-cycle free graph, assuming Triangle Detection on $n$-vertex graphs with  maximum degree at most $\sqrt{n}$ requires $n^{2-o(1)}$ time. 

Even though All-Edges Sparse Triangle on $n$-vertex graphs with  maximum degree at most $\sqrt{n}$ requires $n^{2-o(1)}$ time under either the $3$SUM hypothesis or the APSP hypothesis \cite{Patrascu10, williams2020monochromatic}, the same is not known for Triangle Detection. In fact, it is an open problem to base the hardness of Triangle Detection on some central hypotheses in fine-grained complexity, explicitly asked by \cite{AbboudBKZ22}. 

Towards resolving this open problem, \cite{AbboudBKZ22} proposes a new hypothesis, which they call the Strong Zero-Triangle conjecture. It states that detecting a zero-weight triangle in a edge-weighted tripartite graph with vertex parts of sizes $A, B, C$ and with integer weights in $\{-W, \ldots, W\}$ requires $(\min\{W(AB+BC+CA), ABC\})^{1-o(1)}$ time. Under this hypothesis, they showed that Triangle Detection on $n$-vertex graphs with  maximum degree at most $\sqrt{n}$ requires $n^{2-o(1)}$ time. 

As a side result, we make another progress towards this open problem. Our hardness result is based on a more well-known hypothesis called the Strong $3$SUM hypothesis, which was first proposed by Amir, Chan, Lewenstein and Lewenstein \cite{chanstrong3sum}, and later used by, e.g., \cite{Abboudstrong3sum,hsu2017multidimensional,mucha2019subquadratic,abboud2020impossibility,BringmannW21}. We remark that there is no known direct relations between the Strong Zero-Triangle conjecture and the Strong $3$SUM hypothesis. As far as we know, neither, either, or both of them could be true.

\begin{hypo}[Strong $3$SUM hypothesis]
    In the Word-RAM model with $O(\log n)$-bit words, $3$SUM on size-$n$ set of integers from $[-n^2, n^2]$ cannot be solved in $O(n^{2-\eps})$ time, for any positive constant $\eps>0$.
\end{hypo}

We show the following:
\begin{restatable}{theorem}{TriangleDet}
\label{thm:triangledet}
Under the Strong $3$SUM hypothesis, Triangle Detection on $n$-node graphs with maximum degree $O(n^{1/4})$ requires $n^{1.5-o(1)}$ time. 
\end{restatable}

Our lower bound is arguably lower than that in \cite{AbboudBKZ22}: in terms of $m$, their lower bound is $m^{4/3-o(1)}$, while ours is $m^{6/5-o(1)}$. 
Nevertheless, basing the hardness of Triangle Detection on a more popular hypothesis gives  more confidence that it requires super-linear time. 

Combining the techniques of \cite{AbboudBKZ22} with \cref{thm:triangledet}, one can obtain an $m^{1+\Omega(1)}$-time lower bound for $4$-cycle detection in $m$-edge graphs, assuming the Strong $3$SUM hypothesis. Of course, the exponent here can only be lower than what \cite{AbboudBKZ22} obtained from their Triangle Detection hypothesis.

\subsection{Technical Overview}
\label{sec:overview}

In this section, we will describe the high-level ideas of our reductions from $3$SUM to $3$SUM on Sidon sets, and subsequently to All-Edges Sparse Triangle on quasirandom graphs. 

The additive energy of a set $A \subset \Z$, which we call $E(A)$, is defined as the number of tuples $(a, b, c, d) \in A$ such that $a + b = c + d$. The first component of our reduction is an efficient algorithm for $3$SUM on sets with very large ($\ge n^{3} / K$ for some small $K$) additive energy. For sets $A$with  $E(A) < n^{3} / K$ (moderate energy), we use self-reduction of $3$SUM to reduce the instance to a number of smaller instances, so that the total additive energy of the smaller instances is small. This means that there are very few tuples $(a, b, c, d)$ from the same instance such that $a + b = c + d$ and $\{a, b\} \ne \{c, d\}$, so it suffices to remove these numbers from the instances that contain them. In the following, we describe each of these steps in more details. 

\paragraph{Efficient algorithm for sets with high additive energy. }

First, suppose we need to solve $3$SUM on a size-$n$ set $A$ where $E(A) \ge n^3 / K$ for some small $K$. By the Balog-Szemer{\'e}di-Gowers Theorem, such a set $A$ contains a large subset $A' \subseteq A$ of size $|A'| \ge K^{-O(1)} n$, and with small doubling, $|A' + A'| \le K^{O(1)} |A'|$. Furthermore, such a subset can be found in subquadratic time, by an adaptation of an algorithmic version of the Balog-Szemer{\'e}di-Gowers Theorem given by Chan and Lewenstein \cite{ChanL15}.  If we are able to solve the tripartite version of $3$SUM on sets $A', A, A$, then we can remove $A'$ from $A$ afterwards. We repeat this procedure until either the size of $A$ becomes truly sublinear, when we can use brute-force to solve $3$SUM on $A$, or the energy of $A$ becomes small, when we can apply the reduction for the moderate energy case. Either way, as the size of $A'$ is large, we only need to repeat $K^{O(1)}$ times. 

Therefore, it suffices to give an efficient algorithm for solving tripartite $3$SUM on size-$n$ sets $A, B, C$, where $A$ has small doubling, i.e. $|A+A|$ is small. We will show an algorithm that runs in $\tO(\sqrt{n} |A+A|)$ time, which is truly subquadratic when $|A+A|$ is sufficiently small. The algorithm roughly works as follows. First, we show an efficient algorithm that partitions $B$ into subsets $B_1, B_2, \ldots, B_m$, so that each $B_i$ is a subset of a shift of $A$, i.e., there exists $s_i$ such that $B_i \subseteq A + s_i$. Furthermore,  our algorithm also finds $C_i$, which is a superset of $C \cap -(A + B_i)$, so that $\sum_i |C_i|$ is bounded by $\tO(|A + A|)$ (we cannot afford to simply set $C_i$ to be $C \cap -(A + B_i)$, as finding $C \cap -(A + B_i)$ is not simpler than solving $3$SUM on $A, B_i, C$). Clearly, for some $i$ and $b \in B_i$, the only possible numbers in $C$ that can form a $3$SUM solution with $b$ are from $C_i$. If the size of $C_i$ is smaller than some parameter $t$, we enumerate all pairs $b \in B_i$ and $c \in C_i$, and check whether they are in a $3$SUM solution in $\tO(1)$ time; over all such $C_i$, it takes $\tO(nt)$ time. Otherwise, we use sparse convolution \cite{BringmannFN22} to compute $A + B_i$ and test whether $C_i \cap -(A + B_i)$ is empty in $\tO(|C_i| + |A + B_i|)$ time. Note that sparse convolution runs in $\tO(|A+B_i|) \le \tO(|A + (A + s_i)|) = \tO(|A + A|)$ time, so over all such $C_i$, it takes $\tO(\sum_i |C_i| + \frac{\sum_i |C_i|}{t} \cdot |A + A|) \le \tO(|A+A| + \frac{|A+A|^2}{t})$. Setting $t = \frac{|A+A|}{\sqrt{n}}$ gives the desired $\tO(\sqrt{n} |A+A|)$ time. 

To give some intuition why it is possible to find $B_1, \ldots, B_m$ and $C_1, \ldots, C_m$, we describe the following (inefficient) algorithm that is analogous to our efficient algorithm. Suppose for each $z \in \Z$, we add it to a set $S$ with probability $O(\frac{\log n}{n})$. For every $b \in B$, there exists $s \in S$ such that $b \in A + s$ if and only if $S \cap (b - A) \ne \emptyset$, which happens with high probability. Thus, for every $b$, we can arbitrarily assign it one of the $s \in S$ such that $b \in A + s$. Grouping $b \in B$ assigned with the same $s$ to the same group forms the partition $B_1, \ldots, B_m$. Then let $C_i = C \cap -(s_i + A + A)$. It is a superset of $C \cap -(A + B_i)$ since $B_i \subseteq s_i + A$. Also, each $c \in C$ is in $C_i$ if and only if $s_i \in -(A + A + c)$. As each number in $-(A + A + c)$ is added to $S$ with probability $O(\frac{\log n}{n})$, the expected number of such $s_i$ is $\tO(\frac{|A+A|}{n})$, i.e., $c$ appears in $\tO(\frac{|A+A|}{n})$ many $C_i$ in expectation. Summing over all $c \in C$ gives the desired $\sum_i |C_i| \le \tO(|A+A|)$ bound.

\paragraph*{Self-reduction for sets with moderate additive energy. }

It is well-known that the $3$SUM problem has an efficient self-reduction \cite{BaranDP08} through almost-linear hash functions, such as modulo a random prime, or Dietzfelbinger's hash function (see e.g., \cite{Dietzfelbinger96,Dietzfelbinger18,ChanH20}). Say the input range of a $3$SUM instance is $[-U, U]$ and consider an almost-linear hash  family $\caH$ mapping from $[-U, U]$ to $[m]$. The almost-linear property states that, for any $H \in \caH$ and every $a, b \in [-U, U]$, $H(a) + H(b) + H(-a - b)$ can only have $U^{o(1)}$ possible values. Let us first review the high level ideas of the self-reduction of $3$SUM. 
The self-reduction first samples $H \sim \caH$, and creates a bucket $G_x$ for every $x \in [m]$ that contains every number $a \in A$ where $H(a) = x$. Then we enumerate triples of buckets, and solve a tripartite $3$SUM on numbers from these three buckets. By the almost-linear property, we only need to enumerate $m^2 U^{o(1)}$ triples of buckets. Suppose each bucket has $O(n/m)$ numbers, we get $m^2 U^{o(1)}$ small instances of $3$SUM on sets of sizes $O(n/m)$. 

Suppose the input set $A$ has moderate additive energy $E(A) < n^{3} / K$, i.e., there are only $O(n^{3} / K)$ tuples $(a, b, c, d) \in A$ with $a+b=c+d$ and $\{a, b\} \ne \{c, d\}$. For simplicity, let us focus on the tuples with distinct $a, b, c, d$. Suppose that we can bound the probability of $a, b, c, d$ being in the same bucket by $p$. By the almost-linear property, this bucket appears in $m U^{o(1)}$ small instances. Over all small instances, the expected number of tuples $(a, b, c, d)$ in an instance  with $a+b=c+d$ and $\{a, b\} \ne \{c, d\}$ contributed this way is roughly $\frac{n^3}{K} \cdot p \cdot m U^{o(1)}$. Such tuples can appear in a small instance via other possibilities, such as the case where $a, c$ are in a bucket while $b, d$ are in another bucket, and these two buckets belong to some small instance. However, we ignore such cases in this overview for simplicity. For each such tuple, we remove all its $4$ numbers from the small instance they belong to, but we need to pay $O(n/m)$ time for each number in order to check $3$SUM solutions involving them in a brute-force way. Therefore, the overall running time becomes $\frac{n^4}{K} \cdot p \cdot U^{o(1)}$ and the small instances are on Sidon sets (after we convert the tripartite instance to one-set version in a standard way). In order for this running time to be truly subquadratic, we need $p$ to be close to $1/n^2$. This is possible when $m$ is close to $n$ and the hash family has almost $3$-wise independence guarantees. 

Unfortunately, given the almost-linearity requirement, it seems difficult to achieve full $3$-wise independence: for three integers $x,y,z$ with $x+y+z=0$, the hash value of $z$ is almost determined (up to $U^{o(1)}$ possibilities) by the hash values of $x$ and $y$. We mitigate this by only requiring $3$-universality on triples $(x, y, z)$ with certain properties. More specifically, our hash functions satisfy that
$\Pr[H(x) = H(y) = H(z)]$ is roughly $\frac{1}{m^2}$ on integers $x, y, z$ such that there does not exist integers $\alpha, \beta, \gamma$ with small absolute values such that $\alpha + \beta + \gamma = 0$, $\alpha,\beta,\gamma$ are not all zeros and $\alpha x + \beta y + \gamma z = 0$. Furthermore, we borrow the proof idea from \cite{ldt} that uses Behrend's set \cite{behrend} to split each bucket to multiple sub buckets so that triples of integers with such relations do no appear in the same sub bucket.

\paragraph{Reduction to pseudorandom graphs}

Next, we show hardness of All-Edges Sparse Triangle on pseudorandom graphs by reducing from $3$SUM on Sidon sets. The reduction follows a previous line of reduction from $3$SUM to All-Edges Sparse Triangle via $3$SUM Convolution and Exact Triangle \cite{Patrascu10, williams2013finding, williams2020monochromatic}. As we will show in our reductions, the number of tuples $(a, b, c, d)$ with $a + b = c + d$ in the $3$SUM instance relates to the number of $4$-cycles in the All-Edges Sparse Triangle instance, so starting from a $3$SUM instance on Sidon set helps reducing the number of $4$-cycles in the All-Edges Sparse Triangle instance. Along the way of the reduction, we also achieve a hardness result for Exact Triangle on graphs with certain properties (see \cref{property:weighted-graphs}).  

\paragraph{Comparison with \cite{AbboudBKZ22}.}

The short cycle removal technique of \cite{AbboudBKZ22} can be seen as removing short cycles directly in the input graph of a Triangle Detection or All-Edges Sparse Triangle instance, which incurs some overhead in time complexity. In comparison, our approach removes $4$-cycles in a more indirect way: we trace the hardness of All-Edges Sparse Triangle back to $3$SUM, and remove tuples $(a, b, c, d)$ with $a + b = c + d$ and $\{a, b\} \ne \{c, d\}$ (or, ``arithmetic $4$-cycles'') in the $3$SUM instance, which translates to removing $4$-cycles in the All-Edges Sparse Triangle instance. 
The benefit of our approach is that we can exploit the additive structure of $3$SUM and apply various tools from additive combinatorics and additive algorithms, so that our reduction removes cycles more efficiently.
However, one advantage of their results is that their lower bounds hold under APSP hypothesis as well as $3$SUM hypothesis.

\subsection{Independent works}
Concurrently and independently, Abboud, Bringmann, and Fischer \cite{abf} proved similar fine-grained lower bounds for the 4-Cycle Enumeration problem (or, 4-Cycle Listing) (\cref{thm:4cycle_lower}), and for the Approximate Distance Oracle problem (\cref{thm:4-cycle-enum-sparse}). 
Additionally, they obtained fine-grained lower bounds for Approximate Distance Oracle with stretch $2\le \alpha < 3$ \cite[Theorem 1.3]{abf}.
Similar to our proof, \cite{abf} also uses an energy reduction framework, with some technical parts implemented differently. One noticeable difference is that they reduce the additive energy of 3SUM instance $A$ to $O(|A|^{2+\delta})$ (for arbitrary constant $\delta >0$), while we perform more steps to reduce it all the way down to $2|A|^2 - |A|$ (i.e., $A$ is a Sidon set).
The former bound  already suffices for the application to 4-Cycle Listing and Approximate Distance Oracles in \cite{abf}, whereas the latter bound allows us to show further results such as the 3SUM-hardness of 4-LDT (\cref{thm:4ldt}).

Concurrently and independently, Abboud, Khoury, Leibowitz, and Safier \cite{akls} also gave an algorithm for listing all $t$ 4-cycles in an $n$-node $m$-edge undirected graph in $\tilde O(\min \{n^{2},m^{4/3}\} + t)$ time. 

\subsection{Further Related Works}

Chan and Lewenstein \cite{ChanL15} designed subquadratic-time algorithms for a certain clustered version of $3$SUM. One of their key components is an algorithmic version of the 
Balog-Szemer{\'e}di-Gowers Theorem in a very generalized scenario that requires one to keep track of the uncovered edges in a dense bipartite graph, which needs quadratic time in total (see \cite[Theorem 2.3]{ChanL15}).\footnote{To get subquadratic overall time, they need to apply this lemma on compressed instances.} 
Our proof borrows one of their subroutines, but does not require the full generality of their algorithm. This then allows us to use Ruzsa triangle inequality to avoid spending quadratic time, which is crucial to our proof.

Many problems studied in fine-grained complexity or parameterized complexity have a monochromatic version and a colorful version. The monochromatic version cannot be harder than the colorful version, by a simple reduction using color coding \cite{alon1995color}. 
For several important problems, such as OV and $3$SUM, the reverse direction is also true, proved by simple gadget reductions. However, for some other problems, the reduction is extremely nontrivial. Examples include Euclidean Closest Pairs \cite{bcp1,bcp2,bcp3,SM20}, Bichromatic Graph Diameter \cite{diameter0,diameter1,diameter2,diameter3}, $4$-Cycle Detection \cite{lincoln2018tight,AbboudBKZ22}, $3$-LDTs \cite{ldt}.
Our work provides another example of this phenomenon: $4$-partite version of $4$-LDTs can be easily shown to be $4$SUM hard, but for $1$-partite version of some $4$-LDTs it takes a lot of effort to just prove $3$SUM-hardness.
On the other hand, for some problems, the monochromatic property can be nontrivially used in designing algorithms that run faster than the colorful version, e.g., Element Distinctness in the low-space setting (\cite{BCM13,BansalGN018,ChenJWW22,xinlyu}).

\begin{figure}[ht]
    \centering
    \scalebox{0.8}{
    \begin{tikzpicture}
    
        \pgfmathsetmacro{\R}{4}
        \node at(0, 0)  [anchor=center, text width=3.7cm,align=center] (aetriangle){AE Sparse Tri.\ on quasirandom graphs};
        \node at({-\R + \R * cos(288)}, {\R * sin(288)})  [anchor=center, align=center] (3sum){$3$SUM};
        \node at({-\R + \R * cos(216)}, {\R * sin(216)})  [anchor=center, text width=2.5cm,align=center] (moderate3sum){Moderate-Energy $3$SUM};
        \node at({-\R + \R * cos(144)}, {\R * sin(144)})  [anchor=center, text width=2cm,align=center] (sidon3sum){$3$SUM on Sidon sets};        
        \node at({-\R + \R * cos(72)}, {\R * sin(72)})  [anchor=center, text width=3.5cm, align=center] (exacttri){Exact Tri.\ on graphs with Property~\ref{property:weighted-graphs}};        
		
		\draw[opacity=0.4, dashed, rounded corners=3] (3sum.north east) -- (3sum.north west) -- (3sum.south west) -- (3sum.south east) -- cycle;
		\draw[opacity=0.4, dashed, rounded corners=3] (moderate3sum.north east) -- (moderate3sum.north west) -- (moderate3sum.south west) -- (moderate3sum.south east) -- cycle;
		\draw[opacity=0.4, dashed, rounded corners=3] (sidon3sum.north east) -- (sidon3sum.north west) -- (sidon3sum.south west) -- (sidon3sum.south east) -- cycle;
		\draw[opacity=0.4, dashed, rounded corners=3] (exacttri.north east) -- (exacttri.north west) -- (exacttri.south west) -- (exacttri.south east) -- cycle;
		\draw[opacity=0.4, dashed, rounded corners=3] (aetriangle.north east) -- (aetriangle.north west) -- (aetriangle.south west) -- (aetriangle.south east) -- cycle;

		\draw[->,line width=1pt, bend left=29] (3sum) to[]  node[below] {Thm.~\ref{thm:moderate-energy-reduction}} (moderate3sum);
		\draw[->,line width=1pt, bend left=30] (moderate3sum) to[]  node[left] {Thm.~\ref{thm:main-largeinput}} (sidon3sum);
		\draw[->,line width=1pt, bend left=19] (sidon3sum.50) to[]  node[left] {Lem.~\ref{lem:3sumsidon23sumconv},~\ref{lem:3sumconv2zwt}} (exacttri);
		\draw[->,line width=1pt, bend left=27] (exacttri) to[]  node[right] {Lem.~\ref{lem:zwt2ae}} (aetriangle);

		\node at(8, 3)  [anchor=center, text width=3.7cm,align=center] (4cycle){$4$-Cycle Enumeration};
		\node at(8, 1)  [anchor=center, text width=3.7cm,align=center] (DO){Offline Approximate Distance Oracle};
		\node at(8, -1)  [anchor=center, text width=3.7cm,align=center] (approxSP){Dynamic Approximate Shortest Paths};
		\node at(8, -3)  [anchor=center, text width=3.7cm,align=center] (ANSC){All-Nodes Shortest Cycles};
		
		\draw[opacity=0.4, dashed, rounded corners=3] (4cycle.north east) -- (4cycle.north west) -- (4cycle.south west) -- (4cycle.south east) -- cycle;
		\draw[opacity=0.4, dashed, rounded corners=3] (DO.north east) -- (DO.north west) -- (DO.south west) -- (DO.south east) -- cycle;
		\draw[opacity=0.4, dashed, rounded corners=3] (approxSP.north east) -- (approxSP.north west) -- (approxSP.south west) -- (approxSP.south east) -- cycle;
		\draw[opacity=0.4, dashed, rounded corners=3] (ANSC.north east) -- (ANSC.north west) -- (ANSC.south west) -- (ANSC.south east) -- cycle;

		\draw[->,line width=1pt] (aetriangle.15) to[]  node[above, yshift=-10] {\rotatebox{31}{Thm.~\ref{thm:4cycle_lower}}} (4cycle.180);
		\draw[->,line width=1pt] (aetriangle.5) to[]  node[above,yshift=-9] {\rotatebox{11}{Thm.~\ref{thm:DO},\ref{thm:do2} }} (DO.180);
		\draw[->,line width=1pt] (aetriangle.-5) to[]  node[above,yshift=-5] {\rotatebox{-10}{Thm.~\ref{thm:DynamicSP}}} (approxSP.180);
		\draw[->,line width=1pt] (aetriangle.-15) to[]  node[above, yshift=-12] {\rotatebox{-31}{Thm.~\ref{thm:intro:ansc}}} (ANSC.180);
    \end{tikzpicture}}
    \caption{This figure depicts the main reductions in this paper. The reductions from All-Edges Sparse Triangle on quasirandom graphs to $4$-Cycle Enumeration, Offline Approximate Distance Oracle, Dynamic Approximate Shortest Paths and All-Nodes Shortest Cycles also use \cref{cor:ae-lb} and \cref{lem:application-lemma} as intermediate steps. Specially, the reduction to Offline Approximate Distance Oracle in \cref{thm:do2} uses a variant of \cref{lem:application-lemma} (\cref{lem:application-lemma-unbalan}).}
\end{figure}
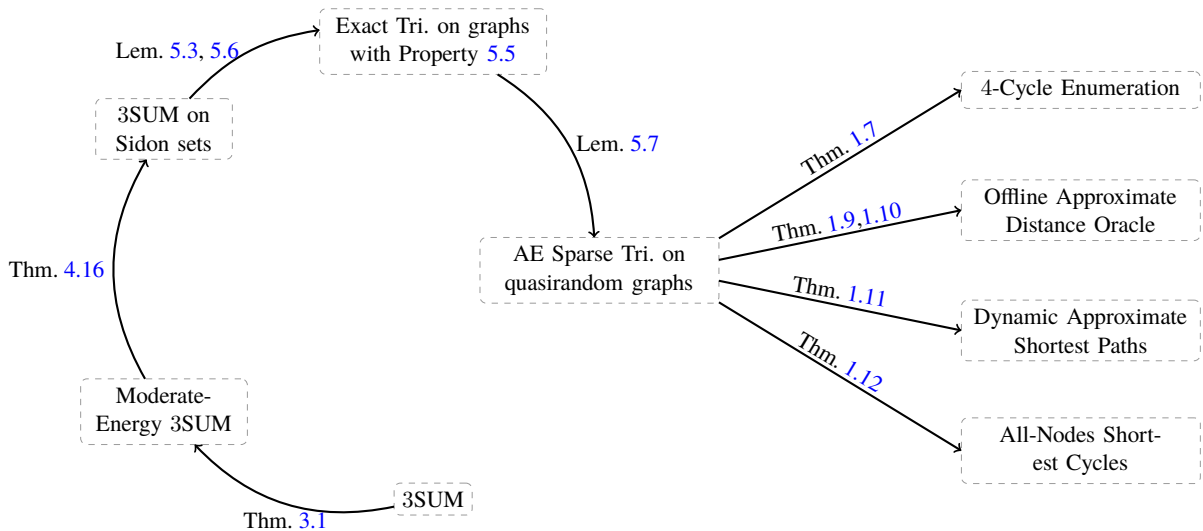

\subsection{Paper Organization}
We give necessary definitions and backgrounds in \cref{sec:prelim}. Next, we reduce $3$SUM to $3$SUM on sets with moderate additive energy in \cref{sec:moderate} and then further reduce to $3$SUM on Sidon sets in \cref{sec:sidon}. This result is used in \cref{sec:quasirandom} to show the $3$SUM-hardness of All-Edges Sparse Triangle on quasirandom graphs, which is then applied to show conditional lower bounds of $4$-Cycle Enumeration, Distance Oracles, Dynamic Shortest Paths and All-Nodes Shortest Cycles in \cref{sec:graph}. In \cref{sec:4-cycle-enum-ub}, we show algorithms that match the above-mentioned conditional lower bounds of $4$-Cycle Enumeration. In \cref{sec:strong-3sum}, we show a lower bound on Triangle Detection under the Strong $3$SUM hypothesis. Finally, we conclude with several open problems in \cref{sec:open}. 
\section{Preliminaries}
\label{sec:prelim}
Denote $[n] = \{1,2,\dots,n\}$.

We use the convention that $(a\bmod p) \in \{0,1,\dots,p-1\}\subset \Z$ regardless of the sign of $a \in \Z$. For a set $A$, denote $A\bmod p = \{a\bmod p: a\in A\}$.

\subsection{Problem Definitions}

The $3$SUM problem is defined as follows.
\begin{definition}[$3$SUM]
 Given an integer set $A \subseteq \Z \cap [-U,U]$ of size $|A|=n$, decide if there exist $a,b,c\in A$ such that $a+b+c=0$.
\end{definition}

\begin{hypo}[$3$SUM hypothesis]
\label{hypo3sum}
    In the Word-RAM model with $O(\log n)$-bit words, $3$SUM with input range $U=n^3$ cannot be solved in $O(n^{2-\eps})$ time, for any positive constant $\eps>0$.
\end{hypo}

There are several variants of the $3$SUM problem studied in the literature. To name a few, one could require integers $a,b,c$ to be distinct, or could ask for $a+b+c=t$ for a given target $t$ not necessarily zero. 
Another variant is the $3$-partite version (sometimes called colorful $3$SUM), where the input contains three sets $A,B,C$ instead of a single set, and $a\in A,b \in B,c \in C$ is required.
It is a standard exercise to show the equivalence of these variants of $3$SUM (see also \cite{ldt}).

Sidon sets are well-studied objects in additive combinatorics. We consider the computational problem of deciding whether a set is a Sidon set.
\begin{definition}[Sidon Set Verification]
  Set $A\subset \Z$ is called a \emph{Sidon set} if $A$ contains no solutions to $a+b=c+d$ except when $\{a,b\}=\{c,d\}$.

 In the Sidon Set Verification problem, we are given an integer set $A \subseteq \Z \cap [-U,U]$ of size $|A|=n$, and need to decide whether $A$ is a Sidon set.
\end{definition}

The $3$SUM problem and the Sidon Set Verification problem are special cases of the more general  \emph{(homogeneous) $k$-Variate Linear Degeneracy Testing ($k$-LDT)} problem \cite{ldt}: for a fixed homogeneous linear equation $\sum_{i=1}^k \beta_i a_i=0$ with non-zero integer coefficients $\beta_i$, given an input integer set $A$, find a \emph{good solution} $(a_1,\dots,a_k)\in A^k$ that satisfies the equation. We consider two variants for the  definition of good solutions:
\begin{enumerate}
    \item The solution contains distinct $a_i$. This is the definition used in \cite{ldt}.
    \item The solution is \emph{nontrivial}, as defined below.
    \end{enumerate}
\begin{definition}[Nontrivial solutions to a linear equation]
\label{defn:nontri}
 A solution $(a_1,\dots,a_k)$ to the equation $\sum_{i=1}^k\beta_i a_i=0$ is called \emph{trivial}, if for every $a\in \{a_1,\dots,a_k\}$ it holds that $\sum_{i: a_i=a}\beta_i=0$. All other solutions are called nontrivial.
\end{definition}
The distinct $a_i$ definition is more restrictive than the nontrivial definition. Note that nontrivial solutions can only exist when $\sum_{i=1}^k \beta_i = 0$.
  For example: 
  \begin{itemize}
      \item In the AVERAGE problem (equation $a_1-2a_2+a_3=0$),  the trivial solutions are $(a_1,a_2,a_3)=(a,a,a)$ for all $a\in A$. In this case, the definition coincides with the distinct $a_i$ definition.
      \item  In the Sidon Set Verification problem (equation $a_1+a_2-a_3-a_4=0$), the trivial solutions are $(a_1,a_2,a_3,a_4)=(a,b,a,b)$ or $(a,b,b,a)$  for all $a,b\in A$. For $w\neq 0$, $(a,a+2w,a+w,a+w)$ is a nontrivial solution, but does not have distinct $a_i$.
  \end{itemize} 
In this paper, we only consider (homogeneous) $3$-LDTs and $4$-LDTs. Moreover, we only study the hardness of (homogeneous) $4$-LDTs with zero coefficient sum $\sum_{i=1}^4 \beta_i=0$, because $k$-LDTs with non-homogeneous equations or nonzero coefficient sum are known to be either trivial or $k$SUM-hard \cite{ldt}.\footnote{Dudek, Gawrychowski, and Starikovskaya~\cite{ldt} only proved  this fact for $k=3$, but their argument easily generalizes to larger $k$.}

\subsection{Additive Combinatorics}
Denote $a + B := \{a + b: b\in B\}$, and $a \cdot  B := \{ab: b\in B\}$.
\begin{definition}[Sumset]
   For sets $A,B\subseteq \Z$, define their \emph{sumset} as
   \[ A+ B := \{ a+b: a\in A, b\in B\}.\]
\end{definition}
For finite $A\subset \Z$, $|A+A|/|A|$ is called the doubling constant of $A$.

We use the sparse convolution algorithm to compute sumsets \cite{ColeH02,ArnoldR15,ChanL15,sparse4,sparse3,Nakos20,nfold,sparse2,BringmannFN22}.
\begin{theorem}[Sparse convolution, \cite{BringmannFN22}]
	\label{thm:sparseconv}
    Given two integer sets $A,B \subseteq [U]$, there is a deterministic algorithm that computes their sumset $A+B$ with output-sensitive time complexity $O(|A+B|\cdot \polylog U)$.
\end{theorem}

\begin{definition}[Additive energy]
	Let $A\subset \Z$ be a finite set. The \emph{additive energy} of $A$ is defined as
	\begin{equation}
        \label{eqn:addenergy}
         E(A):= |\{(a,b,c,d)\in A\times A\times A\times A:  a+b=c+d\}|.
    \end{equation}
\end{definition}
Let 
\begin{equation}
    r_A(x):= |\{(a,b)\in A\times A: a+b= x\}|.
\end{equation}
Then 
\begin{equation}
    E(A) = \sum_x r_A(x)^2.
\end{equation}

It holds that \[2|A|^2-|A|\le E(A) \le |A|^3,\] where the lower bound comes from the trivial solutions with $\{a,b\}=\{c,d\}$. This lower bound is achieved if and only if $A$ is a Sidon set.

We use the following standard tools in additive combinatorics (see e.g., \cite{tao_vu_2006,ruzsa2009sumsets,zhao2019graph}).
\begin{lemma}[Ruzsa sum triangle inequality]
	\label{lem:tria}
For finite integer sets $A,B,C$, 
\[   |A+B||C| \le |A+C||B+C| .\]
\end{lemma}

\begin{definition}[Fourier Transform]
   For $f\colon \Z \to \C$ with finite support, its Fourier transform $\widehat f: \R/ \Z \to \C$ is defined as
   \[ \widehat f(\theta) := \sum_{x\in \Z} f(x)e(-x\theta),\]
   where $e(t):= \exp(2\pi it)$.
\end{definition}
\begin{lemma}[Counting solutions to a linear equation]
    \label{lem:fourier}
   For finite set $A\subset \Z$ and coefficients $\beta_1,\dots,\beta_k\in \Z$, 
   \[ \big \lvert \big \{ (a_1,\dots,a_k)\in A^k : \beta_1 a_1+\dots +\beta_ka_k = 0\big \} \big \rvert  = \int_0^1 \widehat{1_A}(\beta_1 t)\widehat{1_A}(\beta_2 t)\cdots \widehat{1_A}(\beta_k t) dt,\]
   where $1_A(\cdot)$ is the indicator function of set $A$. This identity holds even when $A$ is a multiset, and $1_A(a)$ is the multiplicity of $a$ in $A$.
\end{lemma}
Using \cref{lem:fourier}, the additive energy $E(A)$ in \eqref{eqn:addenergy} can be written as
\begin{equation}
   E(A) =  \int_0^1 \big \lvert\widehat{1_A}(\theta)\big \rvert^4 d\theta.
\end{equation}

\begin{lemma}
	\label{lem:count4}
   For finite set $A\subset \Z$ and non-zero coefficients $\beta_1,\beta_2,\beta_3,\beta_4\in \Z \setminus \{0\}$, 
\[ \big \lvert \big \{ (a_1,a_2,a_3,a_4)\in A^4 : \beta_1a_1+\beta_2a_2+\beta_3a_3+\beta_4a_4= 0\big \} \big \rvert  \le E(A).\]
\end{lemma}

\begin{proof}
 By \cref{lem:fourier} and Cauchy-Schwarz inequality, the number of such 4-tuples is 
    \begin{align*}
    &\int_{0}^1 \widehat{1_A}(\beta_1\theta)\widehat{1_A}(\beta_2\theta)\widehat{1_A}(\beta_3\theta)\widehat{1_A}(\beta_4\theta) d\theta\\
    \le\,  &\Big (\int_{0}^1 \big \lvert \widehat{1_A}(\beta_1\theta)\widehat{1_A}(\beta_2\theta)\big \rvert^2 d\theta
    \cdot 
    \int_0^1
    \big \lvert\widehat{1_A}(\beta_3\theta)\widehat{1_A}(\beta_4\theta)\big \rvert^2 d\theta\Big )^{1/2}\\
     \le\, & \Bigg (\Big ( \int_{0}^1 \big \lvert \widehat{1_A}(\beta_1\theta)\big \rvert^4 d\theta\cdot \int_{0}^1 \big \lvert\widehat{1_A}(\beta_2\theta)\big \rvert^4 d\theta\Big )^{1/2}\Big (\int_{0}^1 \big \lvert\widehat{1_A}(\beta_3\theta)\big \rvert^4 d\theta\cdot \int_{0}^1 \big \lvert\widehat{1_A}(\beta_4\theta)\big \rvert^4 d\theta\Big )^{1/2}\Bigg)^{1/2}\\
    = \, &\Bigg(\Big(\int_0^1 \big \lvert\widehat{1_A}(\theta)\big \rvert^4 d\theta\Big)^4\Bigg )^{1/4} \tag{since $\beta_i\in \Z\setminus\{ 0\}$}\\
    = \, & E(A). \qedhere
    \end{align*}
\end{proof}

\section{Reduction to Moderate-Energy \texorpdfstring{$3$}{3}SUM}
\label{sec:moderate}
In general, a $3$SUM instance $A$ of size $|A|=n$ can have additive energy up to $n^3$ asymptotically. In this section, we provide a reduction from an arbitrary $3$SUM instance $A$ to another $3$SUM instance $\hat A$ with moderate additive energy, $E(\hat A) < O(|\hat A|^{3-\eps})$, for some positive constant $\eps>0$. 
The reduction is formally summarized in the following theorem.
\begin{theorem}[Reduction to moderate-energy $3$SUM]
    \label{thm:moderate-energy-reduction}
   There exist universal constants $\delta>0$ and $\eps>0$ such that the following holds.
   Given an integer set $A\subseteq \Z \cap [-U,U]$ of size $|A|=n$, there is a randomized algorithm in $O(n^{2-\delta}\polylog U)$ time that, with probability $1$, either
   \begin{enumerate}[label=(\alph*)]
    \item finds a $3$SUM solution $a,b,c\in A, a+b+c=0$, or \label{item:reduction-found}
    \item returns a subset $\hat A\subseteq A$, such that $A$ has a $3$SUM solution if and only if $\hat A$ has one. \label{item:case2}
   \end{enumerate}
   Moreover, the probability that Case \ref{item:case2} occurs and $E(\hat A)>|\hat{A}|^{3-\eps}$ is at most $1/3$.  
\end{theorem}

\subsection{Overview}
The proof of \cref{thm:moderate-energy-reduction} has two ingredients: an algorithmic version of the celebrated Balog-Szemer{\'e}di-Gowers Theorem \cite{balog1994statistical,gowers2001new}, and a subquadratic-time algorithm for 3-partite $3$SUM when one of the input sets has small doubling. 

The BSG theorem states that any set $A\subset \Z$ with high additive energy $E(A)\ge  |A|^3/K$ must have a subset $A'\subseteq A$ that has large size $|A'|\ge K^{-O(1)}|A|$, and small doubling, $|A'+A'|\le K^{O(1)}|A'|$.
 The following lemma gives a subquadratic-time randomized algorithm for finding such subset $A'$.
\begin{lemma}[BSG lemma]
	\label{lem:findAbsg}
   There exist universal constants $\alpha>0,\beta>0$ such that the following holds.
	 Given an integer set $A\subseteq \Z \cap [-U,U]$ of size $|A|=n$, and a parameter $1\le K\le n$, there is a randomized algorithm in $O(n K^\alpha \polylog U)$ time that, with probability $1$, either
	\begin{enumerate}[label=(\roman*)]
	\item returns a subset $A'\subseteq A$ such that $|A'| \ge n/K^\beta$ and $|A'+A'|\le nK^{\alpha}$, or \label{item:bsg-success}
	\item outputs ``failure''.\label{item:bsg-low}
	\end{enumerate}
Moreover, if $E(A)>n^3/K$, then the failure probability is at most $1/n^2$.
\end{lemma}

The second ingredient we need is a specialized algorithm for the 3-partite $3$SUM problem, where we are given three integer sets $A,B,C\subseteq \Z \cap  [-U,U]$, and need to find $a\in A, b\in B, c\in C$ such that $a+b+c = 0$. This algorithm has sub-quadratic running time provided $|A+A|$ is small.
\begin{lemma}[$3$SUM with small doubling]
	\label{lem:3sum-small-doub}
    Given input sets $A,B,C\subseteq \Z \cap  [-U,U]$ with $\max\{|A|,|B|,|C|\}\le n$, the 3-partite $3$SUM problem can be solved by a Las Vegas randomized algorithm with time complexity 
    \[  O\left ( \frac{n\,  |A+A|}{\sqrt{|A|}}\cdot \polylog U\right ).\]
\end{lemma}

Now we prove \cref{thm:moderate-energy-reduction} using \cref{lem:findAbsg} and \cref{lem:3sum-small-doub}.
\begin{proof}[Proof of \cref{thm:moderate-energy-reduction}]
Let $A\subseteq \Z\cap [-U,U]$ be the input $3$SUM instance of size $|A|=n$. 
Let $\eps>0$ be a small constant to be determined.

The reduction is described in Algorithm~\ref{alg:reduction}. 
It maintains a subset $\hat A \subseteq A$ initialized to $A$, and repeatedly uses the BSG lemma (\cref{lem:findAbsg}) to peel off a large subset $A'\subseteq \hat A$ with small doubling, and then uses \cref{lem:3sum-small-doub} to find $3$SUM solutions involving $A'$, i.e., $(a,b,c)\in A'\times \hat A\times \hat A$ with $a+b+c=0$. 
The algorithm is terminated whenever a $3$SUM solution is found (Case~\ref{item:reduction-found} of \cref{thm:moderate-energy-reduction}).
 We return this subset $\hat A$ (Case~\ref{item:case2} of \cref{thm:moderate-energy-reduction}) once the BSG lemma reports failure (Line~\ref{line:break}). 
 If $|\hat A|$ becomes smaller than $n^{1-\varepsilon}$ (Line~\ref{line:brute}), we can afford to solve $3$SUM on $|\hat A|$ by brute-force in $O(|\hat A|^2) \le O(n^{2-2\eps})$ time, and return the found $3$SUM solution, or return an empty set $\hat A:=\emptyset$ if no solution is found.
    \begin{algorithm}
\caption{Reduction to moderate-energy $3$SUM}
\label{alg:reduction}
\DontPrintSemicolon
Initialize $\hat A \gets A$\\
\While{$|\hat A| \ge n^{1-\eps}$\label{line:small}}{
    Apply BSG lemma (\cref{lem:findAbsg}) to $\hat A$ with $K:= |\hat A|^\eps$\\
    \If{BSG lemma successfully returned a subset $A'\subseteq \hat A$}{
     Solve 3-partite $3$SUM on $A', \hat A, \hat A$ using \cref{lem:3sum-small-doub}, and terminate if a solution is found\\
     $\hat A \gets \hat A \setminus A'$ \label{line:return}
    }\lElse{\textbf{return} $\hat A$ \label{line:break} } 
}
Solve $3$SUM on $\hat A$ by brute-force in $O(|\hat A|^2)$ time\label{line:brute}
\end{algorithm}

Algorithm~\ref{alg:reduction} maintains the invariant that $A$ has a $3$SUM solution if and only if $\hat A$ has one. Indeed, if 3-partite $3$SUM on $A',\hat A,\hat A$ has no solution, then any $3$SUM solution of $\hat A$ must be contained in $\hat A \setminus A'$, so we can remove $A'$ from consideration.
This shows that with probability $1$ either Case~\ref{item:reduction-found} or Case~\ref{item:case2} in the theorem statement holds.
   
Now we analyze the time complexity of Algorithm~\ref{alg:reduction}.
Starting from $\hat A= A$, each iteration of the \textbf{while} loop removes from $\hat A$ a subset $A'$ of size \[|A'| \ge |\hat A|/K^\beta = |\hat A|^{1-\eps \beta} \ge n^{(1-\eps)(1-\eps \beta)},\]
so the total number of iterations is at most \[|A|/n^{(1-\eps)(1-\eps \beta)} = n^{\eps(\beta +1-\eps \beta)}.\]
In each iteration, ignoring poly-logarithmic factors, \cref{lem:findAbsg} has time complexity
\[ |\hat A|\cdot K^{\alpha}  =|\hat A|\cdot |\hat A|^{\eps\alpha} \le n^{1+\eps\alpha},  \]
and \cref{lem:3sum-small-doub} has time complexity
\[ \frac{|\hat A| \cdot |A'+A'|}{\sqrt{|A'|}} \le \frac{|\hat A|\cdot |\hat A|K^{\alpha}}{\sqrt{|\hat A|/K^\beta}} = |\hat A|^{1.5}K^{\alpha+\beta/2} \le n^{1.5+\eps(\alpha+\beta/2)}.\]
Summing over all iterations, the total time complexity of Algorithm~\ref{alg:reduction} is at most
\begin{equation}
 n^{\eps(\beta + 1 - \eps \beta)} \cdot O(n^{1.5+\eps(\alpha+\beta/2)} \polylog U) \le O(n^{1.6}\polylog U), \label{eqn:time}
\end{equation}
where we set $\eps = 0.1/(\alpha+1.5\beta+1)$.

In each iteration, the probability that $E(\hat A)> |\hat A|^{3}/K = |\hat A|^{3-\eps}$ yet the BSG lemma outputs failure is at most $1/|\hat{A}|^2 \le 1/n^{2-2\eps}$. 
By a union bound over at most $n^{\eps(\beta + 1-\eps \beta)} < n$ iterations, with at most $1/n^{1-2\eps}$ probability we eventually return $\hat A$ with too large additive energy $E(\hat A) >  |\hat A|^{3-\eps}$.

A minor issue is that the time bound in \eqref{eqn:time} is in expectation rather than worst-case, since the 3-partite $3$SUM algorithm in \cref{lem:3sum-small-doub} is Las Vegas randomized.
To fix this, we terminate Algorithm~\ref{alg:reduction} (and default to return $\hat A\gets A$) after executing longer than 10 times the expected time bound, which additionally incurs $1/10$ failure probability by Markov's inequality.
\end{proof}

The next two sections prove \cref{lem:findAbsg} and \cref{lem:3sum-small-doub} respectively.

\subsection{Balog-Szemer{\'e}di-Gowers Theorem}
One crucial ingredient in the proof of the BSG theorem is a graph-theoretic lemma, which was implicit in Gowers' proof \cite{gowers2001new} and appeared in the presentation of Sudakov, Szemer{\'e}di, and Vu \cite{sudakov2005question}. 
Here we use an algorithmic version of this graph-theoretic lemma given by Chan and Lewenstein \cite{ChanL15}, which achieved better running time using random sampling.

\begin{lemma}[{\cite[Lemma 7.1 and 7.2]{ChanL15}}]
	\label{lem:chanbsg}
Given a bipartite graph $G = (A \cup B, E)$, with $|E| \ge  \alpha|A||B|$, there exist $A' \subseteq A$ and $B' \subseteq B$ such that
\begin{itemize}
	\item  for every $a' \in A' , b' \in B'$, there are $\Omega(\alpha^5 |A||B|)$ length-3 walks from $a'$ to $b'$, and
\item  $|A'| \ge \alpha |A|/8$, $|B'|\ge \alpha |B|/8$.
\end{itemize}
Moreover, $A'$ and $B'$ can be computed by a randomized Monte Carlo algorithm in $\tilde O\big ((|A|+|B|)\cdot (1/\alpha)^5\big )$ time that succeeds with high probability, given access to the adjacency matrix of $G$.
\end{lemma}

  The following lemma uses a simple sampling to find all the popular sums in $A+A$, i.e., elements $x$ with large $r_A(x)$ (recall $r_A(x) = |\{(a,b)\in A\times A: a+b= x\}|$).
\begin{lemma}[Find popular sums]
	\label{lem:findpopular}
	 Given a set $A\subset \Z \cap [-U, U]$ of size $|A|=n$ and a parameter $1\le K\le n$, there is a randomized algorithm that computes a set $S\subseteq A+A$ in $\tilde O(Kn)$ time, such that with at least $1-1/\poly(n)$ probability,
	\begin{itemize}
		\item for all $x\in A+A$ with $r_A(x)\ge 0.5n/K$, it holds that $x\in S$, and
		\item for all $x\in S$, $r_A(x) \ge 0.2n/K$.
	\end{itemize}
\end{lemma}
\begin{proof}
Uniformly and independently sample $m = O(Kn \log n)$ random pairs $(a_1,b_1),\dots,(a_m,b_m)\in A\times A$, and compute the sums $s_i=a_i+b_i$ ($1\le i\le m$). Include $x \in S$ if and only if $x$ appears at least $m\cdot \frac{0.3n/K}{|A\times A|} = \frac{0.3m}{Kn}$ times in $s_1,\dots,s_m$. This succeeds with high probability by Chernoff bound.
\end{proof}

Now we prove the algorithmic BSG lemma by directly plugging \cref{lem:chanbsg} and \cref{lem:findpopular} into the proof of BSG theorem \cite{gowers2001new,sudakov2005question}.
\begin{proof}[Proof of \cref{lem:findAbsg}]
First apply \cref{lem:findpopular} on $A$, and obtain set $S\subseteq A+A$  in $\tilde O(Kn)$ time. Assume $S$ is correct, which happens with high probability.

If $|S|\le 0.5n/K$, then
\begin{align*}
	 E(A) &= \sum_{x\in S}r_A(x)^2 + \sum_{x\notin S} r_A(x)^2 \\
	 & \le \sum_{x\in S}r_A(x)^2 + \sum_{x\notin S} r_A(x) \cdot (0.5n/K) \tag{by \cref{lem:findpopular}}\\
	 & \le \sum_{x\in S} |A|^2 + |A|^2 \cdot (0.5n/K) \\
	 & \le n^3/K, \tag{by $|S|\le 0.5n/K$}
\end{align*}
and we can return failure due to small additive energy. 

Otherwise, $|S|>0.5n/K$. Define a bipartite graph $G =(A_1 \cup A_2, E)$ where $A_1,A_2$ are copies of $A$ and $ E=\{(a,b)\in A\times A: a+b \in S\}$, which has size 
\[|E| = \sum_{x\in S}r_A(x) \ge \sum_{x\in S} 0.2n/K \ge 0.1 n^2/K^2.\]  Apply \cref{lem:chanbsg} on $G$ with  $\alpha = 0.1/K^2$, and in $\tilde O(n\cdot \poly(K))$ time obtain $A',B'\subseteq A$ with \[|A'|,|B'|\ge \alpha n/8 =   \Omega(n/K^2),\] such that for every $a'\in A',b'\in B'$ there are $\Omega(\alpha^5n^2)$ length-3 walks from $a'$ to $b'$ (if \cref{lem:chanbsg} is successful). 

 For every $s\in A'+B'$, arbitrarily pick $(a',b')\in A'\times B'$ with $a'+b'=s$. Then the $\Omega(\alpha^5 n^2)$ length-3 walks $(a', b'',a'',b')$ in $G$ yield \emph{distinct} representations of $s = x-y+z$ with $(x,y,z)\in S\times S\times S$, by setting $x = a'+b'',y=b''+a'',z=a''+b'$.
Hence,
\begin{equation}
|S\times S \times S| \ge |A'+B'|\cdot \Omega(\alpha^5n^2). \label{eqn:bsg1}
\end{equation}

 Note that 
 \begin{equation}
|S| \le  | \{ x\in A+A: r_A(x)\ge 0.2n/K\} | \le \frac{\sum_{x \in A+A}r_A(x) }{ 0.2n/K} = 5nK.\label{eqn:bsg2}
 \end{equation}
Combining \eqref{eqn:bsg1} and \eqref{eqn:bsg2} gives $|A'+B'| \le |S|^3/\Omega(\alpha^5 n^2) \le O(n K^{13})$.
Then, by Ruzsa sum triangle inequality (\cref{lem:tria}), 
 \[|A'+A'|\le \frac{|A'+B'|^2}{|B'|} \le O\left (\frac{n^2K^{26}} {n/K^2}\right ) \le  O(n K^{28}),\]
 so we can return $A'$ as the desired subset.

 Note that the desired properties $|A'|\ge \Omega(n/K^2)$ and $|A'+A'|\le O(nK^{28})$ can be deterministically verified efficiently using sparse convolution (\cref{thm:sparseconv}), and the algorithm can return failure if the verification fails. Thus, the algorithm errs only when it returns failure (Case~\ref{item:bsg-low} of \cref{lem:findAbsg}), which happens with at most $1/\poly(n)$ error probability as guaranteed by \cref{lem:chanbsg} and \cref{lem:findpopular}. 
\end{proof}

\subsection{Solving \texorpdfstring{$3$}{3}SUM on Sets with Small Doubling}
Recall that in the 3-partite $3$SUM problem, we are given three integer sets $A,B,C\subseteq \Z \cap  [-U,U]$ with size \[\max\{|A|,|B|,|C|\}\le n,\] and need to find $a\in A, b\in B, c\in C$ such that $a+b+c = 0$. 
We are interested in the case where the doubling constant  
\begin{equation}
\label{eqn:doubleA}
K:=\frac{|A+A|}{|A|}
\end{equation} is small.
Readers are encouraged to think of the special case where $A$ is contained in an interval of length $O(|A|)$, and hence has doubling constant $K= O(1)$. Designing an $\tilde O(n^{1.5})$ algorithm for this special case is a standard exercise hinted by Chan and Lewenstein \cite[Section 4]{ChanL15}. 
However, generalizing this to arbitrary $A$ of small doubling requires some amount of effort.

The following lemma exploits the small doubling of $A$ to decompose the problem $(A,B,C)$ into several subproblems $\{(A,B_i,C_i)\}_{i}$ of small total size. 
\begin{lemma}
	\label{lem:compute-decomp}
There is a Las Vegas randomized algorithm with time complexity $\tilde O(nK)$ that computes subsets $B_1,B_2,\dots, B_m \subseteq B$ and $C_1,C_2,\dots,C_m \subseteq C$, with the following properties:
\begin{enumerate}[label=(\roman*)]
	\item $B_1,B_2,\dots,B_m$ form a partition of $B$. \label{item:partition}
	\item For all $1\le i\le m$, $ C \cap -(A+B_i)\subseteq C_i$. \label{item:coverc}
	\item For all $1\le i\le m$, $|A + B_i| \le |A+A|.$ \label{item:absmall}
	\item $\sum_{i=1}^m |C_i| \le  O( nK\log n )$.\label{item:csmall}
\end{enumerate}
\end{lemma}

We first describe the algorithm for 3-partite $3$SUM assuming \cref{lem:compute-decomp}.
\begin{proof}[Proof of \cref{lem:3sum-small-doub} using \cref{lem:compute-decomp}]
Run the algorithm of \cref{lem:compute-decomp} in $\tilde O(nK)$ time, and obtain subsets
$B_1,B_2,\dots, B_m \subseteq B$ and $C_1,C_2,\dots,C_m \subseteq C$. By properties \ref{item:partition} and \ref{item:coverc}, any $3$SUM solution $(a,b,c)\in A\times B\times C$ ($a+b+c=0$) must be included in $A\times B_i \times C_i$ for some $1\le i\le m$, and vice versa.

Let $t\ge 1$ be some parameter to be determined.
For each $1\le i\le m$, to solve 3-partite $3$SUM on $(A,B_i,C_i)$, there are two cases:
\begin{enumerate}[label=(\alph*)]
	\item If $|C_i|\ge t$, we use sparse convolution (\cref{thm:sparseconv}) to compute $A + B_i$ in $O(|A+B_i|\polylog U) \le O(|A+A|\polylog U)$ time (due to property \ref{item:absmall}), and then compute $(A+B_i) \cap -C_i$ which corresponds to the $3$SUM solutions on $(A,B_i,C_i)$.  \label{item:caselarge}
	\item Otherwise, $|C_i|< t$. We iterate over all $(b,c) \in B_i\times C_i$ and check if $-b-c\in A$, in $\tilde O(|B_i||C_i|) $ time.\label{item:casesmall}
\end{enumerate}

The total time complexity for Case~\ref{item:caselarge} is at most
\begin{equation}
	\frac{\sum_{i=1}^m|C_i|}{t}\cdot O(|A+A|\polylog U)
\le O\left (\frac{n|A|K^2 }{t}\polylog U\right ) \label{eqn:time1}
\end{equation}
by \cref{eqn:doubleA} and property~\ref{item:csmall}.
The total time complexity for Case~\ref{item:casesmall} is at most
\begin{equation}
 \sum_{i: |C_i|<t} \tilde O(|B_i||C_i|) \le \sum_{i:|C_i|<t} \tilde O(|B_i|t)\le \tilde O(nt)\label{eqn:time2}
\end{equation}
by property~\ref{item:partition}.
Balance \eqref{eqn:time1} and \eqref{eqn:time2} by choosing $t = K\sqrt{|A|} = |A+A|/\sqrt{|A|}$,  and the total time complexity is \[O\left (\frac{n|A+A|}{\sqrt{|A|}} \cdot \polylog U\right ).\qedhere\]
\end{proof}

Now we describe the algorithm claimed in \cref{lem:compute-decomp} for computing the subsets $B_1,B_2,\dots, B_m \subseteq B$ and $C_1,C_2,\dots,C_m \subseteq C$. 
The overall strategy here is to use randomly shifted copies of set $A$ to cover all the elements in $B$. However, the family of random shifts needs to have certain structure that helps us find the covered elements efficiently.

\begin{proof}[Proof of \cref{lem:compute-decomp}]

We start with the following claim. Denote $p\Z+r := \{pz+r: z\in \Z\}$.
\begin{claim}
	\label{claim1}
	Given integer sets $A',B'\subset \Z$ and prime $p$, let $r\in \F_p$ be uniformly chosen at random. 
	For each $s\in p\Z+r$, define subset $B'_s := B'\cap (s + A')$. 
	Then the expected total size of these subsets is
	\[ \Ex_{r\in \F_p} \left[\sum_{s\in p\Z + r}|B'_s| \right] = \frac{|A'||B'|}{p}.\]
	Moreover, we can compute all non-empty subsets $B_s'$ (where $s\in p\Z+r$) in near-linear $\tilde O( |A'|+|B'|+ \sum_{s\in p\Z+r}|B_s'|)$ time.
\end{claim}
\begin{proof}[Proof of \cref{claim1}]
	To find all subsets $B_s' $, it suffices to iterate over every $b\in B'$ and find all $s\in p\Z+r$ such that $b\in s+A'$, or equivalently, $b= s+ a$ for some $a\in A'$, which necessarily implies $a \equiv b-r \pmod{p}$.
	Conversely, every $a\in A' \cap (p\Z+b-r)$ determines an $s=b-a\in p\Z+r$ such that $b\in B_s'$. This means 
	\begin{align*}
		 \sum_{s\in p\Z+r} |B'_s| &=\sum_{b\in B'} |s\in p\Z+r:  b\in B'_s|\\
		 &= \sum_{b\in B'} |A' \cap (p\Z+b-r)|.
	\end{align*}
 After grouping elements of $A'$ based on their remainders modulo $p$, one can then output all subsets $B'_s$  in near-linear time.

When $r\in \F_p$ is randomly chosen, by linearity of expectation, the expected total size of $B_s'$ is
\begin{align*}
	\Ex_{r\in \F_p}\left[\sum_{s\in p\Z+r} |B'_s|\right] &=
	\Ex_{r\in \F_p}\left[\sum_{b\in B'} |A' \cap (p\Z+b-r)|\right]\\
	&=\sum_{b\in B'} \Ex_{r\in \F_p}\left[|A' \cap (p\Z+b-r)|\right]\\
	& = \sum_{b\in B'} \frac{|A'|}{p}\\
	& = \frac{|A'||B'|}{p}. \qedhere
\end{align*}
\end{proof}

We randomly generate a set of shifts $S\subseteq \Z$ as follows.	Let $k = \Theta(\log U\log n)$ and $\hat p = \Theta(|A|\log U)$. For each $1\le j\le k$, uniformly independently sample a random prime $p_j \in [\hat p, 2\hat p]$ and a random field element $r_j\in \F_{p_j}$.
Let $T_j :=  p_j\Z+r_j$, and define the set of shifts as $S:= T_1 \cup \dots \cup T_k$.

For every shift $s\in S$, define $B_s := B\cap (s + A)$.  We can compute $B_s$ for all $s\in S$ by applying \cref{claim1} to $A':=A,B':=B$ for every $p_j\Z+r_j$ ($1\le j\le k$). Similarly, define $C_s := C\cap -(s + A + A)$, and compute $C_s$ for all $s\in S$ by applying \cref{claim1} to $A':=A+A,B':=-C$
(note that $A+A$ can be computed in $ O(|A+A|\polylog U)$ time by \cref{thm:sparseconv}). Summing over $1\le j\le k$, the bound in \cref{claim1} implies
\[ \Ex\left[\sum_{s\in S}|B_s|\right]\le k\cdot \frac{|A||B|}{\hat p} = O(n\log n), \]
and 
\begin{equation}
	\label{eqn:csmall}
	\Ex\left[\sum_{s\in S}|C_s|\right]\le k\cdot \frac{|A+A||C|}{\hat p} = O\left(\frac{|A+A| \, n\log n}{|A|}\right) =O(Kn\log n). 
\end{equation}
The sets $B_s,C_s$  can be computed  in near-linear time for all $s\in S$, and we can assume their total size does not exceed a large constant times the expectation, by Markov's inequality.

It remains to check that $\bigcup_{s\in S} B_s = B\cap \bigcup_{s\in S} (s+A)$ covers the entire $B$ with good probability.
For every $b\in B$, note that $b\in A + T_j$ if and only if there exists $a\in A$ such that $b\equiv r_j+a \pmod{p_j}$.  Hence, for fixed $p_j$, we have 
\begin{equation}
\Pr_{r_j\in \F_{p_j}}[b\in A+T_j] = \frac{| A \bmod p_j| }{p_j}, \label{eqn:amod}
\end{equation}
where $|A\bmod p_j|$ stands for the number of distinct remainders of elements in $A$ modulo $p_j$, and can be bounded as 
\begin{equation}
	\label{eqn:Aminus}
	|A \bmod p_j| \ge |A| \,- \sum_{a,a'\in A,a<a'} [a'-a\equiv 0 \pmod{p_j}].
\end{equation}
  Fixing $a<a'$, for a random prime $p\in [ \hat p, 2\hat p]$, the prime number theorem implies
\[\Pr_{p_j\in [\hat p,2\hat p]}[a'-a\equiv 0 \pmod{p_j}] \le O\left(\frac{\log_{\hat p} (a'-a)}{\hat p/\log \hat p}\right) \le O\left(\frac{\log U /\log \hat p}{ \hat p/\log \hat p}\right) \le  \frac{1}{|A|} ,\]
where the last step follows by setting the constant factor hidden in $\hat p=\Theta(|A|\log U)$ large enough.
 Combining with \eqref{eqn:Aminus}, this implies 
 $\Ex_{p_j\in [\hat p,2\hat p]} |A\bmod p_j| \ge |A|/2$.
 Then from \eqref{eqn:amod} we get \[\Pr_{p_j,r_j}[b\in A+T_j] \ge \Omega(|A|/\hat p) \ge  \Omega(1/\log U). \]
Recall $S = T_1\cup\dots \cup T_k$ for some $k=\Theta(\log U \log n)$, where $T_j$ are sampled independently from each other. So
with high probability $b$ is contained in $A+T_j$ for some $1\le j\le k$.
Then, by a union bound over all $b\in B$,  we have $B \subseteq \bigcup_{s\in S} (s+A)$ with high probability.

Finally, we return the subsets $\{B_s\}_{s\in S}$ and $\{C_s\}_{s\in S}$, except that we first remove the duplicates among the sets $B_s$ to ensure that they form a partition of $B$ (property~\ref{item:partition}).
  By definition, $B_s \subseteq s+A$, and $|B_s+A|\le |s+A+A|=|A+A|$, which proves property~\ref{item:absmall}. By definition of  $C_s$, we have
  $C_s = C \cap -(s+A+A) \supseteq C \cap -(B_s+A)$, which proves  property~\ref{item:coverc}. Finally, property~\ref{item:csmall} follows from \eqref{eqn:csmall} and Markov's inequality. 
  Note that removing duplicates from $B_s$ does not hurt properties \ref{item:coverc}, \ref{item:absmall}, \ref{item:csmall}.

  Since properties \ref{item:partition} and \ref{item:csmall} can be deterministically verified,
  and properties \ref{item:coverc} and \ref{item:absmall} are guaranteed to hold,
  the algorithm above can be made Las Vegas.
\end{proof}
\section{Reduction to \texorpdfstring{$3$}{3}SUM on Sidon Sets}
\label{sec:sidon}

In this section, we further reduce a moderate-energy $3$SUM instance to a $3$SUM instance on Sidon sets. 
In fact, the produced instance avoids not only Sidon $4$-tuples, but all small-coefficient $4$-term linear relations as well. 
To state our formal result, we make the following technical definition, which is also used crucially in our proof.

\begin{definition}[$k$-term $\ell$-relation]
We say $k$ integers $a_1,\dots,a_k$ have an \emph{$\ell$-relation}, if there exist integer coefficients $\beta_1,\dots,\beta_k\in [-\ell,\ell]$ that have sum  $\sum_{i=1}^k\beta_i = 0$ and are not all zero, such that $\sum_{i=1}^k \beta_i a_i = 0$. 
Moreover, we say $\sum_{i=1}^k \beta_i a_i = 0$ is a \emph{nontrivial} $\ell$-relation, if $(a_1,\dots,a_k)$ is a nontrivial solution to the equation $\sum_{i=1}^k \beta_i a_i=0$ (see \cref{defn:nontri}).
\end{definition}
We only consider $3$-term and $4$-term relations. 
For example, a nontrivial $3$-term arithmetic progression $(a,b,c)$ form a nontrivial $3$-term $2$-relation $a-2b+c=0$
(where the coefficients $1,(-2),1$ have zero sum and maximum magnitude $2$), and a Sidon $4$-tuple $(a,b,c,d)$  (where $\{a,b\}\neq \{c,d\}$) form a nontrivial $4$-term $1$-relation $a+b-c-d=0$.
Here are more examples: integers $101,103,109$ have a nontrivial $3$-term $4$-relation $3\cdot 101 - 4\cdot 103 + 1\cdot 109=0$, but do not have any $3$-term $3$-relations.
Integers $999,101,103,109$ have a nontrivial $4$-term $4$-relation $0\cdot 999+3\cdot 101-4\cdot 103 +1\cdot 109 = 0$. 
Integers $101,103,103,109$ have a nontrivial $4$-term $3$-relation $3\cdot 101 - 2\cdot 103 - 2\cdot 103 + 2\cdot 109=0$, and also a \emph{trivial} $4$-term $1$-relation $0\cdot 101 + 1\cdot 103  - 1\cdot 103+ 0\cdot 109=0$.

We prove the following theorem.
\begin{theorem}[Generalized version of \cref{thm:main}]
\label{thm:$3$SUM-sidon}
For any constants $\ell\ge 1$ and $\delta\in (0,1)$, solving $3$SUM on size-$n$ sets of integers bounded by $[-n^{3+\delta},n^{3+\delta}]$ avoiding nontrivial $4$-term $\ell$-relations requires $n^{2-o(1)}$ time, assuming the $3$SUM hypothesis. 
\end{theorem}
In particular, solving $3$SUM on Sidon sets is $3$SUM-hard, proving \cref{thm:main}. Note that if a set avoids nontrivial $4$-term $\ell$-relations, it does not contain four distinct numbers that have an $\ell$-relation either. Thus, avoiding nontrivial $4$-term $\ell$-relations is a stronger condition and \cref{thm:$3$SUM-sidon} also holds if we replace ``nontrivial $4$-term $\ell$-relations'' with ``$4$-term $\ell$-relations involving $4$ distinct numbers''.

In comparison, the main technique of \cite{ldt} can establish a special case of \cref{thm:$3$SUM-sidon}, the $3$SUM-hardness of $3$SUM on sets avoiding nontrivial $3$-term $\ell$-relations. Although a weaker aspect of our result is that our reduction is Las Vegas randomized, while their reduction is deterministic.

Before proving \cref{thm:$3$SUM-sidon}, we first show that it implies the $3$SUM-hardness of detecting  solutions to any nontrivial $4$-LDT, proving \cref{thm:4ldt}, which we recall here:
\FourLDT*

\begin{proof}
In the proof, we will also comment on what need to change if the condition ``avoids  solutions $\sum_{i=1}^4 \beta_i a_i = 0$ for distinct $a_i$'' is replaced with ``avoids nontrivial solutions $\sum_{i=1}^4 \beta_i a_i = 0$ for $a_i$'' (for short, the ``distinct'' condition is replaced with the ``nontrivial'' condition). 

Let $B$ be the input set and let $\ell:= (\max_i{\beta_i})^2$.
    Let $A$ be a $3$SUM instance without nontrivial 4-term $\ell$-relations. By Theorem~\ref{thm:$3$SUM-sidon}, solving $A$ requires $n^{2-o(1)}$ time under the $3$SUM hypothesis. 
    
    We first perform random color-coding \cite{alon1995color} to partition $A$ to $A_1, A_2, A_3$. If $A$ has a $3$SUM solution, then with constant probability, there exist $a_1 \in A_1, a_2 \in A_2, a_3 \in A_3$ such that $a_1+a_2+a_3 = 0$. Let $M := 20(\max_{i} |\beta_i|)^4$, and $U := \max_{a \in A} |a| + 1$. We then create the following sets of rational numbers
    \begin{align*}
        B_1 &= (A_1 + M U) / \beta_1,\\
        B_2 &= (A_2 + M^2 U) / \beta_2, \\
        B_3 &= (A_3  + M^3 U) / \beta_3,\\
        B_4 &= \left\{-((M + M^2 + M^3) U) / \beta_4\right\}.
    \end{align*}
    It suffices to show that there exist $a_1 \in A_1, a_2 \in A_2, a_3 \in A_3$ such that $a_1+a_2+a_3 = 0$ if and only if $B_1 \cup B_2 \cup B_3 \cup B_4$ has a solution to $\sum_{i=1}^4 \beta_i b_i = 0$ for distinct $b_i \in B_1 \cup B_2 \cup B_3 \cup B_4$ for $i \in [4]$ (A caveat is that $B_1, B_2, B_3, B_4$ are sets of rational numbers, but we can easily change them to integers by multiplying every number by $\beta_1 \beta_2 \beta_3 \beta_4$. All integers are within $[-O(n^{\delta}), O(n^\delta)]$). The forward direction is clear: if there exist $a_1 \in A_1, a_2 \in A_2, a_3 \in A_3$ such that $a_1+a_2+a_3 = 0$, then we can find $b_1 = (a_1 + MU) / \beta_1 \in B_1, b_2 = (a_2 + M^2 U) / \beta_2 \in B_2, b_3 = (a_3 + M^3 U) / \beta_3$ and $b_4 = (-(M + M^2 + M^3) U) / \beta_4$ so that clearly $\sum_{i=1}^4 \beta_i b_i = 0$. We then consider the backward direction. 
    
    Suppose there is a solution to $\sum_{i=1}^4 \beta_i b_i = 0$. For $i \in [4]$, let $t_i$ be such that $b_i \in B_{t_i}$. Also, let $b_i = (x_i + y_{t_i}^{(1)} MU + y_{t_i}^{(2)} M^2U + y_{t_i}^{(3)} M^3U) / \beta_{t_i}$, where $x_i \in A \cup \{0\}$ (we can WLOG assume $0 \not \in A$ since if $0 \in A$, we can test whether $0$ is in a $3$SUM solution in $\tO(n)$ time, and then remove $0$) and $y^{(j)}_{t_i}$ is the coefficient in front of $M^j U$ in the definition for $B_{t_i}$, i.e. $(y^{(1)}_1, y^{(1)}_2, y^{(1)}_3, y^{(1)}_4) = (1, 0, 0, -1), (y^{(2)}_1, y^{(2)}_2, y^{(2)}_3, y^{(2)}_4) = (0, 1, 0, -1)$ and $(y^{(3)}_1, y^{(3)}_2, y^{(3)}_3, y^{(3)}_4) = (0, 0, 1, -1)$. 
    
    \begin{claim}
    \label{cl:4ldt-claim1}
    $\sum_{i=1}^4 \frac{\beta_i}{\beta_{t_i}} x_i = 0$ and for every $j \in [3]$, $\sum_{i=1}^4 \frac{\beta_i}{\beta_{t_i}} y^{(j)}_{t_i} = 0$.
    \end{claim}
    \begin{proof}
    We first show $\sum_{i=1}^4 \frac{\beta_i}{\beta_{t_i}} y^{(3)}_{t_i} = 0$. We know that 
    $$\sum_{i=1}^4 \frac{\beta_i}{\beta_{t_i}} (x_i + y_{t_i}^{(1)} MU + y_{t_i}^{(2)} M^2U + y_{t_i}^{(3)} M^3U) = 0.$$
    Multiplying both sides by $\alpha := \beta_1 \beta_2 \beta_3 \beta_4$ gives us an integer equation
    $$\sum_{i=1}^4 \frac{\alpha \beta_i}{\beta_{t_i}} (x_i + y_{t_i}^{(1)} MU + y_{t_i}^{(2)} M^2U + y_{t_i}^{(3)} M^3U) = 0.$$
    We can bound the absolute value of all terms other than $y_{t_i}^{(3)} M^3U$ as follows:
    \begin{align*}
        \left| \sum_{i=1}^4 \frac{\alpha \beta_i}{\beta_{t_i}} (x_i + y_{t_i}^{(1)} MU + y_{t_i}^{(2)} M^2U)\right| & \le (\max_{i} |\beta_i|)^4 \cdot \sum_{i=1}^4 (U + MU + M^2U) \\
        & \le 12 (\max_{i} |\beta_i|)^4 M^2U < M^3U.
    \end{align*}
    Therefore, we must simultaneously have 
    $$
    \sum_{i=1}^4 \frac{\alpha \beta_i}{\beta_{t_i}} (x_i + y_{t_i}^{(1)} MU + y_{t_i}^{(2)} M^2U) = 0 \text{   and   } \sum_{i=1}^4 \frac{\alpha \beta_i}{\beta_{t_i}} (y_{t_i}^{(3)} M^3U) = 0.
    $$
    The second condition implies $\sum_{i=1}^4 \frac{\beta_i}{\beta_{t_i}} y^{(3)}_{t_i} = 0$ and we can use the first condition to show the remaining equations by the same method (and we omit the details). 
    \end{proof}
    
    \begin{claim}
    \label{cl:4ldt-not-all-equal}
    It is not possible $t_1 = t_2 = t_3 = t_4$.  
    \end{claim}
    \begin{proof}
    First, if $t_1, t_2, t_3, t_4$ are equal to $4$, then $b_1, b_2, b_3, b_4$ are all equal, contradicting to the condition that they are distinct (if the ``distinct'' condition is replaced with the ``nontrivial'' condition, it is also a contradiction as $b_1, b_2, b_3, b_4$ form a trivial solution). 
    
    Otherwise, by \cref{cl:4ldt-claim1}, we must have $\sum_{i=1}^4 \frac{\beta_i}{\beta_{t_i}} x_i = 0$, where $x_i \in A$ for every $i$. It implies that $\sum_{i=1}^t \beta_i x_i = 0$ for distinct $x_i$. As $A$ has no nontrivial $4$-term $\ell$-relations, this is impossible (same reasoning applies if the ``distinct'' condition is replaced with the ``nontrivial'' condition). 
    \end{proof}
    
    \begin{claim}
    $t_1, t_2, t_3, t_4$ must all be distinct.  
    \end{claim}
    \begin{proof}
    Suppose they are not all distinct. Consider the vectors $\vec{y}_i := \left(y_i^{(1)}, y_i^{(2)}, y_i^{(3)}\right)$ for $i \in [4]$. Observe that any three of the vectors are independent. \cref{cl:4ldt-claim1} implies that,  $\sum_{1 \le j \le 4} \frac{\beta_{j}}{\beta_{t_{j}}} \vec{y}_{t_{j}} = 0$. Since  $t_1, t_2, t_3, t_4$ are not all distinct, at most three $\vec{y}_i$ are involved in the previous equation. By independence, the coefficient in front of every $\vec{y}_i$ must be $0$. In other words, for every $i \in [4]$, $\sum_{1 \le j \le 4, t_{j}=i} \frac{\beta_{j}}{\beta_{i}} = 0$. Since $\beta_{j} \ne 0$ for every $j$, there are only two possibilities: 1) all $t_{j}$ are equal, which is ruled out by \cref{cl:4ldt-not-all-equal}; 2) $t_1, t_2, t_3, t_4$ are taken from two distinct values, each twice. 
    
    Now we show that the second case is also impossible. First of all, in this case, $\beta_1, \beta_2, \beta_3, \beta_4$ are paired with each other so that the sum of each pair is $0$ ($\beta_{j_1}$ is paired with $\beta_{j_2}$ if $t_{j_1} = t_{j_2}$). 
    By \cref{cl:4ldt-claim1},     $\sum_{j=1}^4 \frac{\beta_j}{\beta_{t_j}} x_j = 0$. If $t_j = 4$ for any $j$, then $(b_1, b_2, b_3, b_4)$ are not all distinct since two of them are from $B_4$, which only contains one element, a contradiction (it is also a contradiction if the ``distinct'' condition is replaced with the ``nontrivial'' condition, as the remaining two elements must be equal as well and the solution will be trivial). Therefore, we can assume $x_j \in A$ for every $j$. Multiplying $\sum_{j=1}^4 \frac{\beta_j}{\beta_{t_j}} x_j = 0$ by $\sqrt{|\beta_1 \beta_2 \beta_3 \beta_4|}$ gives us a linear relation between $(x_1, x_2, x_3, x_4)$ with coefficients from $[-(\max_i{\beta_i})^2, (\max_i{\beta_i})^2] = [-\ell, \ell]$. Since $(x_1, x_2, x_3, x_4)$ are distinct, this is a contradiction as $A$ has no nontrivial 4-term $\ell$-relations (if the ``distinct'' condition is replaced with the ``nontrivial'' condition, then by noticing that in this case if $(b_1, b_2, b_3, b_4)$ is a nontrivial solution, then $(x_1, x_2, x_3, x_4)$ must be distinct and we can apply the same reasoning). 
    \end{proof}
    
    Finally, it suffices to show the following. 
    \begin{claim}
    \label{cl:4ldt-perm}
    If $t_1, t_2, t_3, t_4$ are all  distinct, then there exist $a_1 \in A_1, a_2 \in A_2, a_3 \in A_3$ such that $a_1 + a_2 + a_3 = 0$. 
    \end{claim}
    \begin{proof}
    Let $(r_1, r_2, r_3, r_4)$ be such that $t_{r_i} = i$ for every $i$. For every $i \in [3]$, by \cref{cl:4ldt-claim1} and the definition of $y^{(i)}$, we must have $\frac{\beta_{r_i}}{\beta_i} - \frac{\beta_{r_4}}{\beta_4} = 0$. Thus, $\beta_i = \beta_{r_i} \cdot \frac{\beta_4}{\beta_{r_4}}$. Therefore, $\sum_{i=1}^4 \beta_i = \sum_{i=1}^4 \beta_{r_i} \cdot \frac{\beta_4}{\beta_{r_4}}$. Thus, $\beta_4 = \beta_{r_4}$. This further implies that $\beta_i = \beta_{r_i}$ for every $i \in [4]$. 
    
    By \cref{cl:4ldt-claim1}, $\sum_{i=1}^4 \frac{\beta_i}{\beta_{t_i}} x_i = 0$. Since $\sum_{i=1}^4 \frac{\beta_i}{\beta_{t_i}} x_i = \sum_{i=1}^4 \frac{\beta_i}{\beta_{r_{t_i}}} x_i = \sum_{i=1}^4 x_i$, and we know one of the $x_i$ is $0$ and the rest three are from $A_1, A_2, A_3$ respectively. Therefore, there exist $a_1 \in A_1, a_2 \in A_2, a_3 \in A_3$ such that $a_1 + a_2 + a_3 = 0$. 
    \end{proof}
\end{proof}

In the following, we prove \cref{thm:$3$SUM-sidon}, by applying a careful self-reduction on the moderate-energy $3$SUM instance obtained from \cref{thm:moderate-energy-reduction}.

\subsection{Self-Reduction for \texorpdfstring{$3$}{3}SUM}

It is well-known that the $3$SUM problem has an efficient self-reduction \cite{BaranDP08} through almost-linear hash functions, such as modulo a random prime, or Dietzfelbinger's hash function (see e.g., \cite{Dietzfelbinger96,Dietzfelbinger18,ChanH20}). 
We will use the same self-reduction with a few modifications.

In the following, for integer parameters $m\le U$, we always consider hash families $\caH$ consisting of hash functions of the form  
\[ H\colon \Z \cap [-U,U] \to [m].\]
 We always assume a hash function $H\in \caH$ can be described by a seed of length $U^{o(1)}$, and evaluating $H(x)$ can be done in $U^{o(1)}$ time given $x$ and the description of $H$. First we define the almost-linearity property of a hash family.
\begin{definition}[Almost linearity]
For an integer set $\Delta$,  we say a hash family $\caH$ is \emph{$\Delta$-almost-linear}, if for all hash functions $H\in \caH$ and all $x,y\in \Z \cap [-U,U] $,
	\[ H(x)+H(y) +H(-x-y)\in \Delta.\]
	The set $\Delta$ should be computable in $\poly(|\Delta|\log U)$ time. Sometimes we also say $\caH$ is $|\Delta|$-almost-linear.
\end{definition}

The standard $3$SUM self-reduction proceeds as follows: sample $H\in \caH$, and place input integer $x$ in the bucket numbered $H(x)$. By almost-linearity, it suffices to solve ($3$-partite) $3$SUM on the three buckets numbered $i,j,-i-j-d$ respectively, over all $i\in [m],j\in [m], d\in \Delta$. There are $m^2 |\Delta|$ small instances, and we need to set $|\Delta| = U^{o(1)}$ for time-efficiency. \footnote{The reduction would also work if the set $\Delta$ may depend on the sampled hash function $H\in \caH$ (a common example is to modulo a random prime). Here we do not need this relaxation in our constructions.}

Similar to previous works, in order to bound the size of the instances generated by the self-reduction, we require the hash family $\caH$ to be almost 2-universal: for $x\neq y$, $\Pr_{H\in \caH}[H(x)=H(y)] \approx 1/m$.
However, in our scenario of removing distinct numbers $a, b, c, d$ with $a+b=c+d$, we need stronger independence guarantees in order to bound the probability that $a, b, c, d$ all receive the same hash value.
Unfortunately, given the almost-linearity requirement, it seems difficult to achieve 3-wise independence: for three integers $x,y,z$ with $z=x+y$, the hash value of $z$ is almost determined (up to $|\Delta|$ possibilities) by the hash values of $x$ and $y$.
Nevertheless, we can achieve the desired independence guarantee for three integers that avoid $\ell$-relations (for some small $\ell$). 
We formally state the properties of our hash family in the following lemma, which will be proved in \cref{sec:hash}.
\begin{lemma}[Hash family]
	\label{lem:hash}
	Let $\ell = \lceil \exp((\log U)^{1/3}) \rceil$.\footnote{$\log U$ denotes natural logarithm. }
	Given an integer $m\in [1,U]$, there is a hash family $\caH \subseteq \{H\colon \Z \cap [-U,U] \to [m] \}$ 
	such that:
	\begin{itemize}
		\item $\caH$ is $U^{o(1)}$-almost-linear.
		\item For every $x,y\in \Z\cap [-U,U], x\neq y$, we have 
		\begin{equation}
			\label{eqn:2uni}
			\Pr_{H\in \caH}[H(x)=H(y)] \le \frac{U^{o(1)}}{m}.
		\end{equation}
		\item For every $x,y,z\in \Z\cap [-U,U]$ that do not have any  3-term $\ell$-relations, we have
\begin{equation}
	\label{eqn:3ni}
	\Pr_{H\in \caH}[H(x)=H(y)=H(z)] \le \frac{U^{o(1)}}{m^2}.
\end{equation}
	\end{itemize}
\end{lemma}

To deal with integers that do have 3-term $\ell$-relations, we borrow the proof idea from \cite{ldt} that uses Behrend's set \cite{behrend} to forbid these integers occurring simultaneously. 
The following adaptation of Behrend's construction will be proved in \cref{sec:behrend}.

\begin{lemma}[Behrend's construction]
	\label{lem:behrend}
	Let $\ell = \lceil \exp((\log U)^{1/3}) \rceil$.
Given  set $A = \{a_1,\dots,a_n\} \subset \Z \cap [-U,U]$, there is a deterministic $O(n\polylog U)$-time algorithm  that partitions $A$ into $b = \exp(O (\log U)^{2/3})$ disjoint subsets $B_1, B_2, \dots, B_b$,  such that  $B_i$ avoids nontrivial 3-term $\ell$-relations, for all $1\le i\le b$.
\end{lemma}

In our reduction, we start with a $3$SUM instance $A\subset \Z \cap [-U,U]$ (where $U=n^3$) of size $|A|\le n$ and moderate additive energy $E(A) \le |A|^{3-\eps}$  for some constant $\eps>0$.
Such an instance $A$ is generated (with $2/3$ success probability) by \cref{thm:moderate-energy-reduction} from an arbitrary $n$-size $3$SUM instance with input range $[-n^3,n^3]$ (see \cref{hypo3sum}).

In the following we assume $U =n^3$.
Let $\ell = \lceil \exp((\log U)^{1/3}) \rceil$ be the same parameter from \cref{lem:behrend} and \cref{lem:hash}.

The first step is to perform a self-reduction on $A$, which generates many small $3$SUM instances that have \emph{few} nontrivial 4-term $\ell$-relations.
\begin{definition}[Self-reduction]
	\label{defn:selfred}
  Given $A\subset \Z\cap [-U,U]$ (where $U=n^3$) of size $|A|\le n$,  and a small constant parameter $\gamma \in (0,\eps)$,
   generate smaller $3$SUM instances as follows. Let $m = \lceil n^{1-\eps/4}\rceil$.  
\begin{enumerate}[label={{(\arabic*)}}]
	\item 
	\label{item:stephash}
	Sample $H\colon \Z\cap [-U,U]\to [m]$ from the hash family $\caH$ in \cref{lem:hash},  and use $H$ to  partition $A$ into groups $G_i:= \{a\in A: H(a) = i\}$ for $i\in [m]$. 

	For every group $G_i$ of too large size $|G_i| > n^{1+\gamma}/{m}$, remove all its elements (i.e., redefine $G_i:= \emptyset$).  
	 Use brute-force to check for $3$SUM solutions in $A$ involving these removed elements.
	\item \label{item:stepgen}Use \cref{lem:behrend} to partition $A$ into $b= \exp(O(\log U)^{2/3})$ Behrend sets $B_1,\dots,B_b$.

For each $(x',y',z')\in [b]^3, (x,y,z)\in [m]^3$ with $x+y+z\in  \Delta$, generate the following $3$SUM instance:
	\begin{equation}
		\label{eqn:$3$SUMinstance}
		 \big (q_X U + (G_x\cap B_{x'})\big ) \, \cup\,  \big (q_Y U+(G_y\cap B_{y'})\big ) \,\cup\, \big (q_ZU + (G_{z}\cap B_{z'})\big ),
	\end{equation}
	where $q_X:=10\ell, q_Y:=100\ell^2, q_Z:= -q_X-q_Y$.
	\end{enumerate}
\end{definition}
We need the following property on the shifting coefficients $q_X,q_Y,q_Z$ defined in \cref{defn:selfred}.

\begin{claim}
	\label{claim:bound}
For all integers $\beta_X,\beta_Y,\beta_Z\in \Z \cap [-\ell,\ell]$, we have $|\beta_X q_X+\beta_Y q_Y+\beta_Z q_Z| >5\ell$ unless $\beta_X=\beta_Y=\beta_Z$.
\end{claim}
\begin{proof}
Suppose $|\beta_X q_X+\beta_Y q_Y+\beta_Z q_Z| \le  5\ell$, or equivalently, 
\[\left \lvert (\beta_X-\beta_Z)q_X+(\beta_Y-\beta_Z)q_Y \right \rvert \le  5\ell.\]
  If $\beta_X-\beta_Z=0$ and $\beta_Y-\beta_Z=0$, then $\beta_X=\beta_Y=\beta_Z$. If $\beta_Y-\beta_Z= 0$ and $\beta_X-\beta_Z\neq 0$, then
$LHS = |\beta_X-\beta_Z|\cdot q_X \ge q_X=10\ell>RHS$, a contradiction.  Finally, if $\beta_Y-\beta_Z\neq 0$, then $LHS \ge  |\beta_Y-\beta_Z|\cdot q_Y\ - \ |\beta_X-\beta_Z|\cdot q_X\  \ge q_Y - 2\ell q_X = 80\ell^2>RHS$, a contradiction.
\end{proof}

It is easy to see that the self-reduction from \cref{defn:selfred} preserves the $3$SUM solutions of $A$. Indeed, a $3$SUM solution that involves any integers from $A\setminus (G_1\cup\cdots\cup G_m)$ must be found in step~\ref{item:stephash}.
	Among the remaining integers, a $3$SUM solution $a_1+a_2+a_3=0$ with $a_1\in G_x,a_2\in G_y,a_3\in G_z$ must satisfy $x+y+z\in \Delta$, due to the almost linearity of $H$. Since $B_1,\dots,B_b$ partition $A$, we have $a_1\in G_x\cap B_{x'},a_2\in G_y\cap B_{y'},a_3\in G_z\cap B_{z'}$ for some $x',y',z'\in [b]$, so the shifted version of this solution, $(q_X U + a_1) + (q_Y U + a_2) +(q_ZU + a_3) = 0$, must be included in one of the $3$SUM instances (\cref{eqn:$3$SUMinstance}) generated in step~\ref{item:stepgen}. Conversely, any $3$SUM solution in a generated instance must use exactly one integer from each of the three parts in \cref{eqn:$3$SUMinstance} due to 
\cref{claim:bound},
		 and hence can be shifted back to a $3$SUM solution in $A$.

\begin{obs}
	Step~\ref{item:stephash} in \cref{defn:selfred} runs in expected $n^{2-\gamma+o(1)}$ time.
	\label{obs:subqua}
\end{obs}
\begin{proof}
For each integer $a\in A$, by 2-universality of $H$ (\cref{eqn:2uni}), the expected size of $G_{H(a)}$ is at most $1+(n-1)\cdot \frac{U^{o(1)}}{m}$. Note that we remove $a$ only if $|G_{H(a)}|>n^{1+\gamma}/m$, which happens with probability at most $1/n^{\gamma - o(1)}$ by Markov's inequality, so the total number of removed elements is at most $n^{1-\gamma+o(1)}$ in expectation.

Sort $A$ at the beginning. For each removed integer $a\in A$, it takes an $O(n)$-time scan to check for $3$SUM solutions involving $a$.  Hence, the expected total time to check removed elements is $n^{2-\gamma+o(1)}$.
\end{proof}

The instances generated by the self-reduction (\cref{defn:selfred}) may still contain a few nontrivial 4-term $\ell$-relations. The next step is to remove the elements involved in such relations, so that the remaining elements in each instance are completely free of nontrivial 4-term $\ell$-relations.
To do this, we first need to analyze the expected total number of nontrivial 4-term $\ell$-relations across all the generated instances.
To better understand the following technical parts, readers are encouraged to think of the representative case $\beta_1=\beta_3=1, \beta_2=\beta_4=-1$, i.e., Sidon 4-tuples.

\begin{lemma}[Types of nontrivial $4$-term $\ell$-relations]
	\label{lem:type-relation}
	Denote the $3$SUM instance defined in \cref{eqn:$3$SUMinstance} by $P_X\cup P_Y\cup P_Z$ for short. Then, every nontrivial 4-term $\ell$-relation $\sum_{i=1}^4 \beta_i a'_i=0$ on integers $a'_i\in P_X\cup P_Y\cup P_Z$ must have one of the following two types (up to permuting indices $\{1,2,3,4\}$ and/or $\{X,Y,Z\}$):
	\begin{enumerate}
		\item $a'_1,a'_2,a'_3,a'_4\in P_X$, and all $\beta_i$ are non-zero.
		
		We say this relation is induced by the nontrivial $\ell$-relation $\sum_{i=1}^4\beta_i a_i=0$ in $A$, where $a_i = a'_i - q_XU \in G_x\cap B_{x'}$ ($i\in [4]$).
		\item $a'_1,a'_2\in P_X$, $a'_3,a'_4\in P_Y$, $\beta_1+\beta_2=\beta_3+\beta_4=0$, and all $\beta_i$ are non-zero.
	
		We say this relation is induced by the $\ell$-relation $\sum_{i=1}^4\beta_i a_i=0$ in $A$, where $a_1= a'_1 - q_XU, a_2=a'_2 -q_XU \in G_x\cap B_{x'}$ and $a_3= a'_3 - q_YU, a_4=a'_4 -q_YU \in G_y\cap B_{y'}$.
	\end{enumerate}
\end{lemma}
\begin{proof}
	First note that $\beta_1,\beta_2,\beta_3,\beta_4$ cannot contain more than one zero. Otherwise, without loss of generality suppose $\beta_1=\beta_2=0$. Then $\beta_3=-\beta_4\neq 0$, and from $\beta_3a'_3+\beta_4a'_4 = -\beta_1a'_1-\beta_2a'_2 = 0$ we know $a'_3=a'_4$, meaning that $\sum_{i=1}^4 \beta_i a'_i=0$ is not a nontrivial relation, a contradiction.

Since $P_X = q_X U + (G_x\cap B_{x'})$ and $G_x\subset [-U,U]$, we have $P_X \subset [(q_X -1)U, (q_X +1)U]$. Similarly,
\begin{equation}
	\label{eqn:bound}
	 P_Y\subset [(q_Y -1)U, (q_Y +1)U], P_Z\subset [(q_Z-1)U, (q_Z +1)U].
\end{equation}
Now we use \cref{claim:bound} to rule out other ways of partitioning $a_1',a_2',a_3',a_4'$ into $P_X,P_Y,P_Z$.
\begin{itemize}
	\item Suppose $a'_1,a'_4 \in P_X, a'_2\in P_Y$, and $a'_3\in P_Z$.

	Then, $a'_1/U,a'_4/U \in [q_X-1,q_X+1], a'_2/U \in [q_Y-1,q_Y+1], a'_3/U \in [q_Z-1,q_Z+1]$, and hence 
\[ \Big \lvert (\beta_1+\beta_4)q_X + \beta_2q_Y + \beta_3 q_Z\,-\,\sum_{i=1}^4\beta_i a'_i/U  \Big \rvert \le \sum_{i=1}^4|\beta_i|. \] 	
Combining with $\sum_{i=1}^4\beta_i a'_i=0$ and $|\beta_i|\le \ell$, we get
\[  \lvert (\beta_1+\beta_4)q_X + \beta_2q_Y + \beta_3 q_Z\rvert \le 4\ell, \] 	
which implies $\beta_1+\beta_4=\beta_2=\beta_3$ by \cref{claim:bound}. Since $\sum_{i=1}^4\beta_i=0$, we must have $\beta_2=\beta_3=0$, but we already showed that $\beta_1,\beta_2,\beta_3,\beta_4$ contain at most one zero, a contradiction. 
	\item Suppose $a'_1,a'_2,a'_3\in P_X, a'_4 \in P_Y$.
	
	Using a similar reasoning to the previous case, we obtain
\[  \lvert (\beta_1+\beta_2+\beta_3)q_X + \beta_4q_Y \rvert \le 4\ell, \] 	
which implies $\beta_1+\beta_2+\beta_3=\beta_4=0$ by \cref{claim:bound}.
For $i\in \{1,2,3\}$ let $a_i:= a'_i-q_X U \in P_X - q_XU \subseteq B_{x'}$. Then $\sum_{i=1}^3 \beta_i a_i = \sum_{i=1}^3 \beta_i a'_i -\sum_{i=1}^3\beta_i q_XU = 0-0=0$ is a nontrivial 3-term $\ell$-relation in $B_{x'}$, contradicting the fact that $B_{x'}$ avoids nontrivial 3-term $\ell$-relations.
\end{itemize}
Up to permuting indices $\{1,2,3,4\}$ and/or $\{X,Y,Z\}$, we are left with the two cases claimed in the statement.

In the second case (where $a'_1,a'_2\in P_X$, $a'_3,a'_4\in P_Y$), we can use a similar reasoning to obtain $|(\beta_1+\beta_2)q_X + (\beta_3+\beta_4)q_Y|\le 4\ell$, which implies $\beta_1+\beta_2=\beta_3+\beta_4=0$ by \cref{claim:bound}. Note that none of $\beta_1,\beta_2,\beta_3,\beta_4$ can be zero, since otherwise we would have two zeros $\beta_1=\beta_2=0$ (or $\beta_3=\beta_4=0$).

In the first case (where $a'_1,a'_2,a'_3,a'_4\in P_X$), if $\beta_4=0$, then we would have the same contradiction as in the case of $a'_1,a'_2,a'_3\in P_X, a'_4\in P_Y$. So $\beta_4\neq 0$ (and the same holds for $\beta_1,\beta_2,\beta_3$).
\end{proof}

\cref{lem:type-relation} shows that the nontrivial $4$-term $\ell$-relations $\sum_{i=1}^4 \beta_ia_i'=0$ in the generated instances (\cref{eqn:$3$SUMinstance}) are always induced by $4$-term $\ell$-relations $\sum_{i=1}^4 \beta_ia_i=0$ (where $\beta_i$ are non-zero) from the original input set $A$.
  For each of them, we can bound the expected number of nontrivial relations it induces in the generated instances, and by linearity of expectation this allows us to bound the total number of such relations in the generated instances.
 This is the key property of our reduction.
\begin{lemma}
	\label{lem:boundnumber}
	The expected total number of nontrivial $4$-term $\ell$-relations in all instances (\cref{eqn:$3$SUMinstance}) generated by the self-reduction (\cref{defn:selfred}) is at most $E(A) \cdot n^{o(1)}/m$.
\end{lemma}
\begin{proof}
	For each $4$-term $\ell$-relation $\sum_{i=1}^4\beta_i a_i = 0$ where $a_i\in A, \beta_i\neq 0$ for all $i\in [4]$, we separately analyze the expected number of nontrivial $\ell$-relations $\sum_{i=1}^4\beta_i a_i' = 0$ it induces for each of the two types defined in \cref{lem:type-relation}.
	\begin{itemize}
		\item Type 1: $a'_1,a'_2,a'_3,a'_4\in P_X$, and $a_i = a'_i - q_XU$. Note that $\sum_{i=1}^4\beta_i a_i = 0$ is also a nontrivial relation.
		
		This can happen only if $a_1,a_2,a_3,a_4\in G_x\cap B_{x'}$ for some $x\in [m]$ and $x'\in [b]$. We show that $a_1,a_2,a_3,a_4$ contain at least $3$ distinct integers. Otherwise, $\{a_1,a_2,a_3,a_4\} \subseteq \{u,v\}$ for some $u\neq v$, and combining $\sum_{i=1}^4\beta_i = 0$ and $u\sum_{i: a_i= u}\beta_i  + v\sum_{i: a_i= v}\beta_i = 0$ would imply $\sum_{i: a_i= u}\beta_i  = \sum_{i: a_i= v}\beta_i = 0$, contradicting to the fact that $\sum_{i=1}^4\beta_i a_i = 0$ is a nontrivial relation.

		Without loss of generality assume $a_1,a_2,a_3$ are distinct. Since $a_1,a_2,a_3\in B_{x'}$, they do not have any $3$-term $\ell$-relation. Then, by the 3-universal property (\cref{eqn:3ni}) of the hash family (\cref{lem:hash}), we have 
\[ \Pr_{H\in \caH}[a_1,a_2,a_3,a_4 \in G_x \text{ for some $x\in [m]$}] \le \Pr_{H\in \caH}[H(a_1)=H(a_2)=H(a_3)] \le \frac{n^{o(1)}}{m^2}.\]

If $a_1,a_2,a_3,a_4\in G_x\cap B_{x'}$ happens for some $x\in [m]$ and $x'\in [b]$, then it may induce a nontrivial $4$-term $\ell$-relation in every instance that involve $G_x\cap B_{x'}$. Such instances (indexed by $(x,y,z,x',y',z')\in [m]^3\times [b]^3$ in \cref{defn:selfred}) should satisfy $y+z\in \Delta-x$, so there are only $m\cdot |\Delta|\cdot b^2 $ such instances.
So the expected total number of nontrivial $4$-term $\ell$-relations induced by
$\sum_{i=1}^4\beta_i a_i = 0$
is  at most $\frac{n^{o(1)}}{m^2} \cdot (m\cdot |\Delta|\cdot b^2) \le n^{o(1)}/m$ in expectation.
\item Type 2: $a'_1,a'_2\in P_X$, $a'_3,a'_4\in P_Y$, $\beta_1+\beta_2=\beta_3+\beta_4=0$,
and $a'_1-a_1 = a'_2-a_2 = q_XU, a'_3-a_3=a'_4-a_4=q_YU$.

This can happen only if $a_1,a_2\in G_x\cap B_{x'},a_3,a_4\in G_y\cap B_{y'}$ for some $x,y\in [m]$ and $x',y'\in [b]$. 
Observe that $a_1 \neq a_2$, since otherwise we must have $a_3=a_4$ as well, which would imply $a'_1=a'_2$ and $a'_3=a'_4$, contradicting the assumption that $\sum_{i=1}^4\beta_i a_i' = 0$ is a nontrivial relation.

Then, by the 2-universal property (\cref{eqn:2uni}) of the hash family (\cref{lem:hash}), we have 
\[ \Pr_{H\in \caH}[a_1,a_2 \in G_x \text{ for some $x\in [m]$}] \le \Pr_{H\in \caH}[H(a_1)=H(a_2)] \le \frac{n^{o(1)}}{m}.\]
If $a_1,a_2\in G_x\cap B_{x'},a_3,a_4\in G_y\cap B_{y'}$ happen for some $x,y\in [m]$ and $x',y'\in [b]$, then it may induce a nontrivial $4$-term $\ell$-relation in every instance that involve both $G_x\cap B_{x'}$ and $G_y\cap B_{y'}$. There are only $|\Delta|\cdot b$ such instances, so the expected total number of nontrivial $4$-term $\ell$-relations induced by
$\sum_{i=1}^4\beta_i a_i = 0$
is  at most $\frac{n^{o(1)}}{m} \cdot (|\Delta|\cdot b) \le n^{o(1)}/m$ in expectation.
	\end{itemize}
There are at most $O(\ell^3)$ equations $\sum_{i=1}^4\beta_i a_i = 0$ with integer coefficients $\beta_i \in [-\ell,\ell] \setminus\{0\}$ and $\sum_{i=1}^4\beta_i=0$, and each of them has at most $E(A)$ solutions in $A$ by \cref{lem:count4}, so there are at most $O(\ell^3 E(A))$ such  $4$-term $\ell$-relations in $A$.  
Summing over all of them (and accounting for all possible ways of permuting $\{X,Y,Z\}$ and/or $\{1,2,3,4\}$ in the types), by linearity of expectation, the expected total number of induced nontrivial $4$-term $\ell$-relations over all generated instances is 
$ O(\ell^3 E(A)) \cdot n^{o(1)}/m \le E(A)\cdot n^{o(1)}/m$.
\end{proof}

Now, we describe how to efficiently \emph{report} all the nontrivial $4$-term $\ell$-relations in the generated instances.
\begin{lemma}
	\label{lem:report}
	We can report all the nontrivial $4$-term $\ell$-relations in all instances (\cref{eqn:$3$SUMinstance}) generated by \cref{defn:selfred} in time linear in their number, plus $n^{2+2\gamma + o(1) }/m $ additional time.
\end{lemma}
\begin{proof}
	We first do the following pre-processing step.
For every  $x\in [m], x'\in [b]$, and every $\beta_1,\beta_2 \in \Z \cap [-\ell,\ell] \setminus \{0\}$, compute the set of tuples 
	\begin{equation}
		D_{x,x'}^{\beta_1,\beta_2} = \{(\beta_1 a_1 +\beta_2 a_2,\,  a_1, a_2): a_1,a_2\in G_x\cap B_{x'}, a_1\neq a_2\}.
	\end{equation}
	Then for every $t\in \Z$ and $\beta \in \Z \cap [-\ell,\ell] \setminus \{0\}$, compute a bucket $L_t^\beta$ that contains all tuples $(t,a_1,a_2)$ appearing in any set $D_{x,x'}^{\beta,-\beta}$. Technically, every tuple in the bucket $L_t$ also records which set $D_{x,x'}^{\beta,-\beta}$ it comes from.
This pre-processing step can be implemented in time 
\[  \tilde O\Big (\ell^2 \cdot \sum_{x\in [m],x'\in [b]} |G_x\cap B_{x'}|^2 \Big ) \le \tilde O(\ell^2 mb\cdot (n^{1+\gamma}/m)^2 ) \le n^{2+2\gamma + o(1) }/m. \]

To report all the nontrivial $4$-term $\ell$-relations $\sum_{i=1}^4\beta_i a_i' = 0$ in the generated instances, again we separately consider type 1 and type 2 as defined in \cref{lem:type-relation}.
	\begin{itemize}
		\item Type 1: induced by nontrivial $\ell$-relation $\sum_{i=1}^4\beta_i a_i=0$, where $a_i \in G_x\cap B_{x'}$ ($i\in [4]$).

		For every $x\in [m], x'\in [b]$, to find all nontrivial $4$-term $\ell$-relations in $G_x\cap B_{x'}$, simply enumerate $\beta_1,\beta_2,\beta_3,\beta_4$  and  compare $D_{x,x'}^{\beta_1,\beta_2}, D_{x,x'}^{-\beta_3,-\beta_4}$ to find common sums $\beta_1 a_1 +\beta_2 a_2=-\beta_3 a_3 -\beta_4 a_4$.\footnote{By definition of $D_{x,x'}^{\beta_1,\beta_2}$, here we only find the relations with $a_1\neq a_2$ and $a_3\neq a_4$. This already covers all the possibilities, since we argued in the proof of \cref{lem:boundnumber} that $a_1,a_2,a_3,a_4$ must contain at least three distinct integers.}
		Then we immediately find the induced relations $\sum_{i=1}^4\beta_i a_i' = 0$ in all instances that involve $G_x\cap B_{x'}$.
		\item Type 2: induced by $\ell$-relation $\sum_{i=1}^4\beta_i a_i=0$, where $\beta_1+\beta_2=\beta_3+\beta_4=0$ and $a_1,a_2\in G_x \cap B_{x'}, a_3,a_4\in G_y\cap B_{y'}$.
		
		For every $t\in \Z$ and $\beta_1,\beta_3\in \Z\cap [-\ell,\ell]\setminus \{0\}$, enumerate every pair of $(t,a_1,a_2)\in L_{t}^{\beta_1}, (t,a_3,a_4)\in L_{t}^{-\beta_3}$, which gives the an $\ell$-relation $\beta_1a_1-\beta_1a_2 =t =-\beta_3 a_3 + \beta_3 a_4$ in $A$. If $(t,a_1,a_2)\in D_{x,x'}^{\beta_1,-\beta_1}$ and $(t,a_3,a_4)\in D_{y,y'}^{-\beta_3,\beta_3}$, then it induces a nontrivial $\ell$-relation $\sum_{i=1}^4 \beta a'_i=0$ in every instance involving both $(G_x\cap B_{x'})$ and $(G_y\cap B_{y'})$, where $a'_1, a'_2 \in q_XU + (G_x\cap B_{x'}), a'_3, a'_4 \in q_YU + (G_y\cap B_{y'})$.
	\end{itemize}

	 In both cases, after the pre-processing is finished, reporting the $\ell$-relations does not incur any extra overhead in time complexity. 
\end{proof}

After finding all the nontrivial $4$-term $\ell$-relations in the generated instances, the final step is to remove these involved integers, so that the remaining integers in each instance completely avoid all nontrivial $4$-term $\ell$-relations. The reduction is summarized in the following theorem.
\begin{theorem}
	\label{thm:main-largeinput}
Suppose for some constant $\delta>0$ there is an $O(n_0^{2-\delta})$-time algorithm $\caA$ that solves $3$SUM on input set $A_0\subset \Z \cap [-n_0^{12000},n_0^{12000}]$ of size $|A_0|\le n_0$ that avoids nontrivial $4$-term $\exp((\log n_0)^{1/3})$-relations.

Then, there is an algorithm that solves $3$SUM on size-$n$ input set $A\subset \Z \cap [-n^3,n^3]$ in $O(n^{2-\delta'})$ time for some constant $\delta'>0$ depending on $\delta$. Moreover, this reduction is Las Vegas randomized.
\end{theorem}
\begin{proof}
Given input set $A\subset \Z \cap [-n^3,n^3]$ of size $|A|=n$, first run the sub-quadratic time reduction in \cref{thm:moderate-energy-reduction} to obtain an equivalent input set $\hat A \subseteq A$, which has moderate additive energy $E(\hat A) \le |\hat A|^{3-\eps}$ with at least $2/3$ probability, for some constant $\eps>0$.

Set $\gamma =\delta \eps/200$. Then, apply the self-reduction from \cref{defn:selfred} on $\hat A$, and obtain 
$b^3\cdot m^2\cdot |\Delta| \le m^2\cdot n^{o(1)}$
small $3$SUM instances (\cref{eqn:$3$SUMinstance}) each of size at most $n_0:= 3n^{1+\gamma}/m \ge n^{\gamma+\eps/4}$. This reduction takes sub-quadratic
$n^{2-\gamma+o(1)}$
time by \cref{obs:subqua}.  Recall $m = \lceil n^{1-\eps/4}\rceil $, where $\eps >1/1000$ is the universal constant from \cref{thm:moderate-energy-reduction}.

By \cref{lem:boundnumber}, the total number of nontrivial $4$-term $\ell$-relations $\sum_{i=1}^4\beta_i a_i' = 0$ in these generated instances is at most $E(\hat A)\cdot n^{o(1)}/m$ in expectation.
Use \cref{lem:report} to report all of them in sub-quadratic 
$n^{2+2\gamma + o(1) }/m $ time. Then, remove all the involved integers from the instances, so that the remaining integers in each instance avoid nontrivial $4$-term $\ell$-relations.
For each removed integer $a$ in an instance, use brute-force to check for $3$SUM solutions involving $a$ in that instance, with linear time complexity in the instance size. In total this takes 
\begin{align}
(E(\hat A)\cdot n^{o(1)}/m) \cdot n_0 & \le E(\hat A)\cdot n^{1+\gamma+o(1)}/m^2 \label{eqn:oldbound}\\
&\le n^{(3-\eps) + (1+\gamma)-2(1-\eps/4)+o(1)}\nonumber \\ 
&\le n^{2-0.4\eps+o(1)}\nonumber
\end{align}
 expected time.

Finally, use algorithm $\caA$ to solve these instances.  Each instance has size at most $n_0$, and contains integers from the range $[-n^3,n^3] \subset [-n_0^{12000},n_0^{12000}]$, which avoids nontrivial $4$-term $\ell$-relations for $\ell = \lceil \exp((\log n^3)^{1/3}) \rceil \ge \exp((\log n_0)^{1/3})$, so the input conditions of $\caA$ are satisfied.  Summing over all $m^2\cdot n^{o(1)}$ instances, the total running time of $\caA$ is 
\begin{equation}
\label{eqn:abound}
 m^2\cdot n^{o(1)} \cdot n_0^{2-\delta} \le  n^{2(1-\eps/4)+(\gamma + \eps / 4)(2-\delta) + o(1)}\le n^{2-6\eps\delta/25 +o(1)}.
\end{equation}

Hence, the overall expected time complexity is sub-quadratic. 
\end{proof}

\cref{thm:main-largeinput} is almost as good as the main theorem we claimed, except that the algorithm $\caA$ is assumed to work over very large input range $\Z\cap [-n_0^{12000},n_0^{12000}]$. 
This assumption can be weakened using a few additional standard techniques. Then we obtain the following theorem, which clearly implies \cref{thm:$3$SUM-sidon}.
\begin{theorem}
	\label{thm:main-smallinput}
Suppose for some constant $\delta>0$ there is an $O(n_1^{2-\delta})$-time algorithm $\caA$ that solves $3$SUM on input set $A_1\subset \Z \cap [-n_1^{3+\delta},n_1^{3+\delta}]$ of size $|A_1|= n_1$ that avoids nontrivial $4$-term $\exp((\log n_1)^{1/3})$-relations.

Then, there is an algorithm that solves $3$SUM on size-$n$ input set $A\subset \Z \cap [-n^3,n^3]$ in $O(n^{2-\delta'})$ time for some constant $\delta'>0$ depending on $\delta$. Moreover, this reduction is Las Vegas randomized.
\end{theorem}

\begin{proof}[Proof Sketch]
We assume the fast $3$SUM algorithm $\caA$ as stated in \cref{thm:main-largeinput}, except that now $\caA$ is only required to work on input set $A_1 \subset \Z \cap [-n_1^{3+\delta}, n_1^{3+\delta}]$ where $n_1=|A_1|$. In the following we describe the modifications we make in our reduction.

Recall that in the self-reduction defined in \cref{defn:selfred}, each generated small $3$SUM instance has size at most $n_0 = 3n^{1+\gamma}/m$ (where $n_0\ge n^{\gamma+\eps/4}$).   We first modify this self-reduction, by using a hash function to compress the input range of these instances down to barely super-cubic in their sizes.  
Specifically, let $V = n_0^{3+\delta/2}$, and sample a random prime $p\in [V/2,V]$. Then,  change the definition of the small instances from \cref{eqn:$3$SUMinstance} to the following three $3$SUM instances:
	\begin{align*}
		 \big (q_X V + (G_x\cap B_{x'})\bmod p\big ) \, \cup\,  \big (q_Y V+(G_y\cap B_{y'})\bmod p\big ) \,\cup\, \big (q_ZV + (G_{z}\cap B_{z'})\bmod p\big ),\\
		 \big (q_X V + (G_x\cap B_{x'})\bmod p - p\big ) \, \cup\,  \big (q_Y V+(G_y\cap B_{y'})\bmod p\big ) \,\cup\, \big (q_ZV + (G_{z}\cap B_{z'})\bmod p\big ),\\
		 \big (q_X V + (G_x\cap B_{x'})\bmod p - p\big ) \, \cup\,  \big (q_Y V+(G_y\cap B_{y'})\bmod p - p\big ) \,\cup\, \big (q_ZV + (G_{z}\cap B_{z'})\bmod p\big ),
	\end{align*}
	which consist of integers of magnitude at most \begin{equation}
	    \label{eqn:inputrange}
	    (\max\{|q_X|,|q_Y|,|q_Z|\} + 1)\cdot V\le n_0^{3+\delta/2+o(1)}.
	\end{equation}
	 This compression does not lose any original $3$SUM solutions, since for every original $3$SUM solution $a+b+c=0$ ($a\in G_x\cap B_{x'}, b\in G_y\cap B_{y'},c\in G_z\cap B_{z'}$), it holds that $(a\bmod p)+(b\bmod p)+(c\bmod p)\in \{0,p,2p\}$, and hence one of the three instances defined above captures this solution. 
	However, there are two potential issues introduced by this $\mathrm{mod }$ $p$ compression:
	\begin{itemize}
	    \item Original non-solutions $(a,b,c)\in A^3$ ($a+b+c\neq 0$) may correspond to $3$SUM solutions in these instances, if $(a+b+c)\bmod p = 0$.
	   This would lead to false positives if the instance fed to algorithm $\caA$ contains such a fake $3$SUM solution. 
	   
	   To fix this issue, we first use a simple binary search to have algorithm $\caA$ report a solution $a'+b'+c'=0$ rather than just outputting ``YES'', with only logarithmic overhead in time complexity. Then, we look up the original integers in $A$ that got mapped to $a',b',c'$, and check if they form an actual $3$SUM solution in $A$.\footnote{Multiple original integers from $A$ may be mapped to the same $a'$ if they have the same remainder modulo $p$. In such case we check all of them. }
	   If so, we return ``YES''; if not, we check for $3$SUM solutions involving $a',b'$, or $c'$ in time linear in the size of the instance, and then remove $a',b',c'$ from the instance, and run $\caA$ on the remaining numbers, and so on. The number of iterations here is bounded by the number of fake solutions in this instance.
	    
	    Over random prime $p\in [V/2,V]$, by the prime number theorem, the expected number of such fake $3$SUM solutions in an instance is at most $O(n_0^3 (\log U)/V) \le O(1/n_0^{\delta/2 - o(1)}) = O(1)$. So the total expected time complexity only increases by a constant factor.
	    
	    \item Similarly, this $\mathrm{mod }$ $p$ compression may introduce additional nontrivial 4-term $\ell$-relations in the instances.
	   Note that the proof of \cref{lem:type-relation} with $V$ in place of $U$ still applies to the new definition of instances here,  so the nontrivial 4-term $\ell$-relations in these instances can be divided into two parts: (1) those that would appear as well per original definition \cref{eqn:$3$SUMinstance}, and (2) the additional ones induced by $a_i\in A$ with $\sum_{i=1}^4\beta_ia_i$ being an \emph{non-zero} integer multiple of $p$.  Part (1) satisfies the same bound in \cref{lem:boundnumber}. For part (2),  similarly by the prime number theorem, the expected number of additional 4-term $\ell$-relations in each instance is at most $O(\ell^3 \cdot n_0^4 (\log U)/V) \le n_0^{1-\delta/2 + o(1)}$. Summing over all $m^2\cdot n^{o(1)}$ instances, the total expected count of additional 4-term $\ell$-relations is at most $m^2 \cdot n_0^{1-\delta/2 + o(1)}$.
	    To report these nontrivial 4-term $\ell$-relations (from both part (1) and part (2)),  the strategy of \cref{lem:report} still works with almost no changes.
	    
	    Accounting for the additional part (2) $\ell$-relations causes the time bound \cref{eqn:oldbound} of running brute-force on these involved elements to increase to
	     $ (E(\hat A)\cdot n^{o(1)}/m + 
	     m^2 \cdot n_0^{1-\delta/2 + o(1)} ) \cdot n_0 $, where the extra term is bounded by
	     \begin{align*}
	         m^2 \cdot n_0^{1-\delta/2 + o(1)} \cdot n_0 &\le m^2 \cdot (n^{1+\gamma}/m)^{2-\delta/2 +o(1)}\\
	         &\le n^{2(1-\eps/4) + (\gamma+\eps/4)(2-\delta/2)+o(1)}\\
	         & \le  n^{2 -23\eps\delta/200+o(1)},
	     \end{align*} 
	     which is still subquadratic.
	\end{itemize}
	Note that the two fixes to these two issues are compatible, and the overall reduction is still Las Vegas.
	
	Finally we note that if algorithm $\caA$ is fed with an input instance of size smaller than $(n_0)^{1-\delta/10}$, then we can instead directly run brute-force algorithm on it in $O(n_0^{2-\delta/5})$ time. This only worsens the time bound of \cref{eqn:abound} to 
	\[ 
 m^2\cdot n^{o(1)} \cdot n_0^{2-\delta/5} \le  n^{2(1-\eps/4)+(\gamma + \eps / 4)(2-\delta/5) + o(1)}\le n^{2-6\eps\delta/25 +o(1)}.\]
 If $\caA$ is fed with an input instance of size $n_1\ge (n_0)^{1-\delta/10}$, then the input range \cref{eqn:inputrange} is at most $n_0^{3+\delta/2+o(1)} \le n_1^{3+\delta}$, which satisfies the input condition of  $\caA$.
\end{proof}

\subsection{Construction of the Hash Family}
\label{sec:hash}

In this section we construct the almost linear hash family
$\caH \subseteq \{H\colon \Z \cap [-U,U] \to [m] \}$
claimed in \cref{lem:hash}. 
The building block is the following base case hash family.

\begin{definition}[Base case hash family]
	\label{defn:hashbase}
Given integer parameters $\ell\le U$, sample hash function \[h\colon \Z \cap [-U,U] \to  \{0,1,\dots,\ell-1\}\] as follows: pick prime $p\in [2U,4U]$ uniformly at random, and then pick $r\in \F_p^*$ uniformly at random. Let \[h(x) = \left \lfloor \frac{(r\cdot x) \bmod p}{p} \cdot \ell \right \rfloor .\]
\end{definition}
\begin{lemma}[Almost linearity]
	\label{lem:linear}
For $x,y,z\in \Z \cap [-U,U]$ with $x+y+z=0$, we always have $h(x)+h(y)+h(z) \in \{0, \ell,\ell-1,\ell-2,2\ell,2\ell-1,2\ell-2\}$.
\end{lemma}
\begin{proof}
Denote $x' = (r x) \bmod p,y' = (r y) \bmod p,z' = (r z) \bmod p$. From $x+y+z =0$, we know $x'+y'+z'= i p$ for some $i\in \{0,1,2\}$. 
By definition of $h(\cdot)$, 
	\begin{align*}
		 0 \le \Big( \frac{x'\ell}{p}-h(x) \Big)+\Big(\frac{y'\ell}{p}-h(y)\Big) +\Big (\frac{z'\ell}{p}-h(z)\Big) < 3,	\end{align*}
	or equivalently, $i\ell -3<h(x)+h(y)+h(z) \le i\ell$. The proof follows since $h(x)+h(y)+h(z)$ is a non-negative integer.
\end{proof}
For $a\in \Z$, let \[\|a\|_p:= \min \{ a\bmod p, (-a)\bmod p\}\] be the distance from $a$ to the closest multiple of $p$. 
\begin{lemma}[Almost 2-universality]
	\label{lem:2uni}
	For $x,y\in \Z\cap  [-U,U],x\neq y$, we have
	\[ \Pr_h[h(x)=h(y)] \le \frac{2}{\ell}.\]
\end{lemma}
\begin{proof}
By definition of $h(\cdot)$, observe that $h(x)=h(y)$ implies 
\[ \Bigg\lvert \, \frac{(r x) \bmod p}{p} \cdot \ell \, - \, \frac{(r y) \bmod p}{p} \cdot \ell\,\Bigg \rvert < 1,\]
which then implies
\begin{equation}
 \| rx - ry\|_p < \frac{p}{\ell}.\label{eqn:rxry}
\end{equation}
Since $p>2U\ge |x-y|$, $x-y\not \equiv 0\pmod{p}$. Then, for $r\in \F_p^*$ uniformly chosen at random, $r(x-y)\bmod p$ is uniformly distributed over $\F_p^*$. Hence, 
\[\Pr\Big [\,\| r(x - y)\|_p < \frac{p}{\ell} \,\Big ] = 2\Big \lfloor \frac{p}{\ell} \Big \rfloor\cdot \frac{1}{p-1}  \le \frac{2}{\ell}, \]
which finishes the proof.
\end{proof}
\begin{lemma}[``Almost'' almost 3-universality]
	\label{lem:3uni}
	Suppose $x,y,z\in \Z \cap [-U,U]$ do not have any 3-term $\ell$-relation (in particular, $x,y,z$ are distinct). 
	Then, 
	\[ \Pr_h[h(x)=h(y)=h(z)] \le O\left (\frac{\log U}{\ell^2}\right ).\]
\end{lemma}
\begin{proof}
Since $p>2U$, $x-z\not \equiv 0\pmod{p}$ and $y-z\not \equiv 0\pmod{p}$.
If $h(x)=h(z)$, then from \cref{eqn:rxry} we have
	\begin{equation}
 \|rx-rz\|_p < p/\ell < 4U/\ell.
	 \label{eqn:temphxy}
	\end{equation}
We can assume $\ell\ge 4$, and define integer set
\begin{equation}
A:= [-4U/\ell,4U/\ell] \cap \Z \setminus \{0\}, \label{eqn:defnA}
\end{equation}
whose elements are non-zero and distinct modulo $p$. Then, \cref{eqn:temphxy} implies
\[ r \in \big ((x-z)^{-1}\cdot A\big )\bmod p.\]
Similarly, $h(y)=h(z)$ implies $ r \in \big ((y-z)^{-1}\cdot A\big )\bmod p$. Hence, for fixed $p$, we have
\begin{align}
	 & \Pr_{r\in \F_p^*}[h(x)=h(y)=h(z)]\nonumber\\
	\le \ &    \Pr_{r\in \F_p^*}\big [ r \in \big ((x-z)^{-1}\cdot A\big )\bmod p \text{ and }  r \in \big ((y-z)^{-1}\cdot A\big )\bmod p \big ]\nonumber\\
	= \ & \frac{1}{p-1}   \Big \lvert \big ((x-z)^{-1}\cdot A\big )\bmod p \ \; \cap \  \big ((y-z)^{-1}\cdot A\big )\bmod p \  \Big \rvert\nonumber\\
	= \ & \frac{1}{p-1}   \Big \lvert \big ((y-z)\cdot A\big )\bmod p \ \; \cap \  \big ((x-z)\cdot A\big )\bmod p \  \Big \rvert. \label{eqn:temp233}
\end{align}
Note that for any two integer sets $X,Y \in \Z \cap [- U^{O(1)},+U^{O(1)}]$ we have 
\begin{align*}
	&\Ex_{\text{prime } p \in [2U,4U]} \big [\, \big \lvert  (X\bmod p) \cap (Y\bmod p)\big \rvert\, \big ] \\
	\le \ &  \sum_{x\in X,y\in Y} \Pr_{\text{prime } p \in [2U,4U]} [x-y\equiv 0 \pmod{p}]\\
	  = \ & |X\cap Y|  + \sum_{x\in X,y\in Y,x\neq y} \Pr_{\text{prime } p \in [2U,4U]} [x-y\equiv 0 \pmod{p}]\\
	  \le \ & |X\cap Y| + \frac{|X|\cdot |Y| \cdot O(\log U)}{ U },
\end{align*}
where the last step follows from the prime number theorem.  Applying to \cref{eqn:temp233} gives
	\begin{align}
	 & \Pr_{\text{prime }p\in [2U,4U], r\in \F_p^*}[h(x)=h(y)=h(z)]\nonumber \\
 \le \ & \frac{1}{2U}\Ex_{\text{prime }p\in [2U,4U]} \Big [\    \Big \lvert \big ((y-z)\cdot A\big )\bmod p \ \; \cap \  \big ((x-z)\cdot A\big )\bmod p \  \Big \rvert\ \Big  ] \nonumber \\
 \le \ &  \frac{1}{2U} \Big ( \big \lvert (y-z)\cdot A \cap (x-z)\cdot A \big \rvert  + \frac{|A|^2 \cdot O(\log U)}{U}\Big ). \label{eqn:1}
	\end{align}

	By definition of $A$ in \cref{eqn:defnA}, 
	for any positive integers $a,b$,
	\[ (a \cdot A) \cap (b \cdot A) \subseteq \big ([-4aU/\ell,4aU/\ell]\setminus\{0\}\big )  \cap (a\cdot \Z) \cap (b\cdot \Z),\]
	and hence
	\begin{align}
	 |(a \cdot A) \cap (b \cdot A)|  & \le     \frac{ 8a U/\ell}{\mathrm{lcm}(a,b)}\nonumber \\
	 & =  \frac{\gcd(a,b) }{b}\cdot \frac{8U}{\ell}. \label{temp:23333}
	\end{align}

	Letting $d = \gcd(y-z,x-z)$,  the following linear relation on $x,y,z$,  
	\begin{equation}
		\frac{y-z}{d} \cdot x + \frac{z-x}{d}\cdot y - \frac{(y-z) + (z-x)}{d}\cdot z = 0 \label{eqn:relation}
	\end{equation}
	has integer coefficients summing to zero. If $\max \{ |y-z|/d, |z-x|/d\} \le \ell/2$, then \cref{eqn:relation} would be an $\ell$-relation, contradicting the assumption on $x,y,z$. Hence, \begin{equation}
\label{eqn:temp234}
		\max \{ |y-z|/d, |z-x|/d\} > \ell/2.
	\end{equation} 
	
	Now the first term in \cref{eqn:1} can be bounded as
	\begin{align*}
		\Big \lvert ( (y-z) \cdot A) \cap ( (x-z) \cdot A)\Big \rvert  &= \Big \lvert ( |y-z| \cdot A) \cap ( |x-z| \cdot A)\Big \rvert  \\
		& \le \frac{d}{\max\{|y-z|,|x-z|\} } \cdot \frac{8U}{\ell} \tag{by \cref{temp:23333}}\\
		& < \frac{16U}{\ell^2}. \tag{by \cref{eqn:temp234}}
	\end{align*}
	Plugging into \eqref{eqn:1}, we obtain
	\begin{align*}
	  \Pr_{\text{prime }p\in [2U,4U], r\in \F_p^*}[h(x)=h(y)=h(z)] & \le  
 \frac{1}{2U} \Big ( \frac{16U}{\ell^2}  + \frac{(8U/\ell)^2 \cdot O(\log U)}{U}\Big )\\
 & \le O\left (\frac{\log U}{\ell^2}\right ). \qedhere
	\end{align*}
\end{proof}

The final hash family is constructed by composing the base case hash family.
\begin{proof}[Proof of \cref{lem:hash}]
Recall that $1\le m\le U$, and 
$\ell = \lceil \exp((\log U)^{1/3}) \rceil$.
Let $d$ be the maximum integer $d$ such that $\ell^d \le m$.
 Then $d \le (\log U)^{2/3}$, and $\ell^d > m/\ell$. 

To sample a hash function $H\colon \Z \cap [-U,U] \to [m]$, independently sample hash functions $h_0,\dots,h_{d-1}$ from the base case hash family $h\colon \Z \cap [-U,U] \to \{0,1,\dots,\ell-1\}$ defined in \cref{defn:hashbase}.
 Then the hash function $H$ is defined as 
\[ H(x) := 1 + \sum_{i=0}^{d-1} h_i(x) \cdot \ell^i.\]
It is clear that $1\le H(x)\le \ell^d \le m$, and $H(x)= H(y)$ if and only if $h_i(x)= h_i(y)$ for all $i\in [d]$.
For distinct integers $x,y\in [-U,U]$, by \cref{lem:2uni},
	\[ \Pr_H[H(x)=H(y)] \le \left (\frac{2}{\ell}\right )^d \le \frac{2^d \ell}{m}\le \frac{U^{o(1)}}{m}.\]
For three integers $x,y,z\in [-U,U]$ without any $\ell$-relation, by \cref{lem:3uni},
	\[ \Pr_H[H(x)=H(y)=H(z)] \le \frac{O(\log U)^d}{\ell^{2d}} \le \frac{O(\log U)^d \ell^2}{m^2}\le \frac{U^{o(1)}}{m^2} .\]

For three integers $x,y,z\in [-U,U]$ with $x+y+z=0$,  by \cref{lem:linear},
\[ H(x)+H(y)+H(z) = 3 + \sum_{i=0}^{d-1} c_i \cdot \ell^i\]
where $c_i \in \{0, \ell,\ell-1,\ell-2,2\ell,2\ell-1,2\ell-2\}$. The number of possibilities is $7^d \le U^{o(1)}$, so the hash family is $U^{o(1)}$-almost-linear.
\end{proof}

\subsection{Behrend's Construction}
\label{sec:behrend}
We include a slight adaptation of Behrend's proof \cite{behrend} (see also \cite{ruzsaequation}) here for completeness, and observe that it proves \cref{lem:behrend}.
This lemma can probably simplify the derandomization steps in \cite{ldt}.
\begin{proof}[Proof of \cref{lem:behrend}]
We can separately deal with non-negative integers and negative integers in the input set $A$. Now we can assume $A  \subseteq \Z \cap [0,U]$.

Recall $\ell = \lceil \exp((\log U)^{1/3}) \rceil$. Let $\ell' := 2\ell+1$, and $q := \lceil \exp((\log U)^{2/3}) \rceil$.
For each input integer $x\in A$, express $x$ in $q$-ary:
\[ x = \sum_{j=0}^{d-1} x_j \cdot q^j,\,\, x_j \in \{0,1,\dots,q-1\},\]
where $d = \lceil  \log_{q}(U+1) \rceil \le 1 + (\log U)^{1/3}$.
 Let $r:= \lceil q/\ell'\rceil$, and let $\tilde x_j := \lfloor x_j /r\rfloor \in \{0,1,\dots,\ell'-1\}, x'_j := x_j - \tilde x_j \cdot r \in \{0,1,\dots,r-1\}$.  We insert $x$ into the group indexed by the tuple
\[ (\tilde x_0,\tilde x_1,\dots,\tilde x_{d-1}; \|x'\|_2^2),\]
where $\|x'\|_2^2 = \sum_{j=0}^{d-1} (x'_j)^2 < d r^2 $. The total number of groups is $b \le (\ell')^d \cdot (d r^2)\le \exp(O(\log U)^{2/3})$.

Suppose there exist three integers $x,y,z$ from the same group that have a nontrivial $\ell$-relation. Without loss of generality we can assume the relation is $ax+by - (a+b)z = 0$ where $a,b\in \Z\cap [1, \ell]$.\footnote{All the three coefficients must be non-zero for the relation to be nontrivial. Then two of them have the same sign, which can be assumed to be positive.} Then from $\tilde x = \tilde y = \tilde z$ we obtain $ax'+by'=(a+b)z'$. Then, using triangle inequality and $\|x'\|_2=\|y'\|_2=\|z'\|_2$, we have
\[ (a+b)\|z'\|_2 = \|ax'+by'\|_2 \le \|ax'\|_2+\|by'\|_2 = (a+b)\|z'\|_2,\]
where the equality holds only if $x'$ and $y'$ are colinear. As $\|x'\|=\|y'\|$, we have $x'=y'$, and hence $x=y=z$, which makes the relation $ax+by - (a+b)z = 0$ trivial, a contradiction.
\end{proof}

\section{All-Edges Sparse Triangle on Quasirandom Graphs}
\label{sec:quasirandom}
Recall the All-Edges Sparse Triangle problem.
\begin{definition}[All-Edges Sparse Triangle]
Given an $n$-node $m$-edge undirected graph $G = (V, E)$, determine for every edge $e \in  E$ whether $e$ is in a triangle.
\end{definition}

In this section, we carefully analyze known fine-grained reductions from $3$SUM to the All-Edges Sparse Triangle problem, and show that $3$SUM instances on Sidon sets are reduced to All-Edges Sparse Triangle instances with certain quasirandomness property.
This chain of reduction goes through restricted versions of $3$SUM Convolution and Exact Triangle \cite{Patrascu10,williams2013finding,williams2020monochromatic}.

We first reduce $3$SUM on Sidon sets to $3$SUM Convolution on Sidon sets. 
Here, we follow (a slight modification of) a reduction by Chan and He \cite{ChanH20} for its simplicity. 
For technical reasons, we use a slight variant of $3$SUM Convolution:
\begin{definition}[$3$SUM Convolution']
Given three arrays $A,B,C$ indexed by $\{-n, \ldots, n\}$ whose values are either integers or $\perp$, determine whether there exist $i, j, k \in \{-n, \ldots, n\}$ such that $i+j+k = 0$, $A_i, B_j, C_k \ne \perp$ and $A_i+B_j+C_k = 0$. 
\end{definition}

\begin{lemma}
\label{lem:3sumsidon23sumconv}
If $3$SUM on Sidon sets requires $n^{2-o(1)}$ time, then $3$SUM Convolution' on arrays $A, B, C$ where all integer entries are distinct and form a Sidon set
requires $n^{2-o(1)}$ time. 
\end{lemma}
\begin{proof}
Suppose we are given a $3$SUM instance on a Sidon set $S \subseteq [-U,U]$ where $U = n^{O(1)}$.
Let $p$ be a random prime from $[n, 2n]$. We map each $a \in S$ to bucket $a \bmod{p}$. Let $t$ be a constant to be fixed later. If a bucket has more than $n^t$ elements, we compute whether each number in the bucket is in a $3$SUM solution in a brute-force way, i.e., $\tO(n)$ time per number. For each pair of $a, b \in S$, the probability that they are in the same bucket is $\tO(1/n)$, so the expected size of the bucket of $a$ is $\tO(1)$. Therefore, by Markov's inequality, we handle $a$ in this brute-force way with probability $\tO(1/n^t)$. Overall, the expected runtime of this step is $\tO(n^{2-t})$. 

For all remaining small buckets, we first fill each bucket with $\perp$ so that each bucket has $\Theta(n^t)$ elements, and then randomly permute all the elements inside each bucket. 
Then, we enumerate $i < j < k \in [n^t]$, and find $3$SUM solutions where the first number is the $i$-th number in a bucket, the second number is the $j$-th number in a bucket, and the third number is the $k$-th number in a bucket. Since we randomly permuted each bucket, if there is a $3$SUM solution consisting of numbers in the remaining small buckets, a solution will be found this way with probability at least $1-O(1/n^{t})$ (i.e., as long as the index of the three numbers in their buckets are distinct). Note that each triple $(i, j, k)$ corresponds to $O(1)$ instances of $3$SUM Convolution'. Also, all integer entries of each $3$SUM Convolution' instance is a subset  of the Sidon set $S$, so they are distinct and form a Sidon set. 

If $3$SUM Convolution' on such inputs can be solved in $O(n^{2-\eps})$ time for some $\eps > 0$, we can solve the $3$SUM instance on $S$ in $\tO(n^{2-t} + n^{3t} \cdot n^{2-\eps})$ time in expectation, which is truly subquadratic by setting $t$ appropriately. 
\end{proof}

Next, we reduce $3$SUM Convolution' to the Exact Triangle problem following the standard approach given by Vassilevska Williams and Williams~\cite{williams2013finding}.  

\begin{definition}[Exact Triangle]
Given a directed weighted graph $G=(V,E)$ with weight function $w\colon E\to \Z$, determine whether it contains a directed triangle $(i, j, k)$ with total edge weight  $w_{ij}+w_{jk}+w_{ki} = 0$. 
\end{definition}

We consider Exact Triangle on directed graphs with the following special property.
\begin{property} In a directed weighted graph $G=(V,E)$ with weight function $w\colon E\to \Z$,
\label{property:weighted-graphs}
\begin{itemize}
    \item\textbf{Antisymmetry:} For every $(i,j)\in E$, it holds that $(j,i)\in E$ and $w_{ji}=-w_{ij}$;
    \item\textbf{Few zero-weight $4$-cycles:} The number of directed labeled $4$-cycles in $G$ that have zero weight sum is at most $n^3$.
\end{itemize}
\end{property}

\begin{lemma}
\label{lem:3sumconv2zwt}
If $3$SUM Convolution' on length-$N$ arrays $A, B, C$ where all integer entries are distinct and form a Sidon set
requires $N^{2-o(1)}$ time, then Exact Triangle on $n$-vertex tripartite weighted graphs satisfying \cref{property:weighted-graphs} requires $n^{3-o(1)}$ time.
\end{lemma}
\begin{proof}
    Let $T = \lceil \sqrt{2N+1} \rceil$, and we create $T+1$ instances of Exact Triangle as follows. For each $i = \{0, \ldots, T\}$, we create a tripartite graph $G^i$ on vertex sets $X, Y, Z$, each indexed by $\{0, \ldots, T-1\}$, and add the following edges:
    \begin{itemize}
        \item For every $x \in X$ and $y \in Y$ where $xT+y - N \in [-N,  N]$ and $A_{xT+y - N} \ne \perp$, we add an edge $(x, y)$ with $w_{xy} = A_{xT+y - N}$;
        \item For every $y \in Y$ and $z \in Z$ where $iT+z-y-N \in [-N, N]$ and $B_{iT+z-y-N} \ne \perp$, we add an edge $(y, z)$ with $w_{yz} = B_{iT+z-y-N}$. \item For every $z \in Z$ and $x \in X$ where $2N-z-(x+i)T \in [-N, N]$ and $C_{2N-z-(x+i)T} \ne \perp$, we add an edge $(z, x)$ with $w_{zx} = C_{2N-z-(x+i)T}$. 
    \end{itemize}
    
    We call edges added above \emph{forward direction} edges (i.e., edges directing from $X$ to $Y$, $Y$ to $Z$, or $Z$ to $X$). Then, for every forward direction edge $(u, v)$ with weight $w_{uv}$, we also add a corresponding \emph{backward direction} edge $(v,u)$ with weight $w_{vu}=-w_{uv}$.
    
    It is not difficult to verify that the $3$SUM Convolution' instance has a solution if and only if at least one of the Exact Triangle instances has a solution. First, if there is a zero-weight triangle in some instance (without loss of generality assume its edges all have forward directions), then there exist $i, x, y, z$ such that $A_{xT+y-N} + B_{iT+z-y-N}+C_{2N-z-(x+i)T} = 0$, which is a $3$SUM Convolution' solution. For the other direction, suppose there exists a $3$SUM Convolution' solution $A_{p}+B_{q}+C_{-p-q}=0$. Then it is not difficult to verify that $(x, y, z)$ is a zero-weight triangle in $G^i$ for $x=\lfloor \frac{p+N}{T}\rfloor, y = (p+N) \bmod{T}, z = (q+y+N) \bmod{T}$ and $i = \lfloor \frac{q+y+N}{T}\rfloor$.
    
    Let $n = 3T$ be the number of vertices of each Exact Triangle instance. Clearly, if the initial $3$SUM Convolution'  instance requires $N^{2-o(1)}$ time, Exact Triangle requires $n^{3-o(1)}$ time.
    Since the graphs $G_i$ we constructed clearly have antisymmetric edge weights, it remains to show that they also satisfy the second requirement in  \cref{property:weighted-graphs}. 
    
    Fix any of the graph $G_i$ and fix any directed zero-weight $4$-cycle in it. 
    It is not difficult to verify that two of the edges must have forward direction and the other two edges must have backward direction. Therefore, a zero-weight $4$-cycle would imply a solution $d_1 + d_2 = d_3 + d_4$ where $d_1, d_2, d_3, d_4$ are the edge weights of the $4$-cycle  in the forward direction, which in turn are numbers in the initial $3$SUM Convolution'  instance. 
    
    If the zero-weight $4$-cycle is not completely inside $Y \cup Z$, one can verify that all edge weights of the $4$-cycle in the forward direction correspond to distinct numbers in the initial $3$SUM Convolution'  instance, so that $d_1 + d_2 = d_3 + d_4$ is impossible since the numbers form a Sidon set. For example, if the $4$-cycle is $(y,z_1,x,z_2)\in Y\times Z\times X\times Z$, then its four edge weights in the forward direction are $B_{iT+z_1-y-N} \neq B_{iT+z_2-y-N}$ and $C_{2N - z_1-(x+i)T} \neq  C_{2N - z_2-(x+i)T}$, which are distinct (recall that $B$ and $C$ contain disjoint integers). The other cases can be verified similarly.
    
   The only case where the zero-weight $4$-cycle may have repeated edge weights in the forward direction is when it is completely inside $Y \cup Z$, say it is $(y_1, z_1, y_2, z_2)$. If $z_1-y_1, z_2-y_1, z_1-y_2, z_2-y_2$ are distinct, then the four edge weights in the forward direction still correspond to distinct numbers in the initial $3$SUM Convolution'  instance, which is impossible. Otherwise, we must have $z_1+y_2 = y_1+z_2$ or $z_1+y_1 = z_2+y_2$. There are at most $8T^3 < n^3$ such $4$-cycles.
\end{proof}

Recall the definition of quasirandom graphs:
\Quasirandom*

\begin{lemma}
\label{lem:zwt2ae}
If Exact Triangle on $n$-vertex tripartite weighted graphs satisfying \cref{property:weighted-graphs} requires $n^{3-o(1)}$ time, then All-Edges Sparse Triangle on $n$-vertex quasirandom graphs requires $n^{2-o(1)}$ time. 
\end{lemma}
\begin{proof}
We will reduce Exact Triangle on an $n$-vertex tripartite weighted graph $G=(V , E)$ with $V= A \cup B \cup C$ satisfying \cref{property:weighted-graphs} to All-Edges Sparse Triangle on quasirandom graphs, by following Vassilevska Williams and Xu \cite{williams2020monochromatic}'s reduction from Exact Triangle to All-Edges Sparse Triangle. 

Their reduction works as follows in this setting of parameters. Let $p$ be some sufficiently large prime (here, we need it to be larger than the absolute weight of any triangle or $4$-cycle). Then we can regard the weights $w$ as in $\F_p$, and the set of zero-weight triangle and the set of zero-weight $4$-cycles do not change. Let $x, \{y_v\}_{v \in V} \sim \F_p$ be independent uniform random variables. For any edge $(i, j)$ with weight $w_{ij}$, we define its new weight to be $w'_{ij} = x \cdot w_{ij} - y_i + y_j$. Clearly, we still have $w'_{ij}=-w'_{ji}$ and the set of zero-weight $k$-cycles does not change for any $3 \le k \le 4$ as long as $x \ne 0$ (which happens with high probability).
Then we split $\F_p$ to up to $\sqrt{n}+1$ contiguous segments $L_1, L_2, \ldots, L_{\sqrt{n}+1}$, each of length $\le \lfloor p/\sqrt{n}\rfloor$. We create an instance of All-Edges $O(1)$-Triangle Listing (listing $O(1)$ triangles per edge) for every triple $(i, j, k)$ as long as $0 \in L_i + L_j + L_k$. It is easy to see that there are $O(n)$ instances in total. For each instance, we initially create an empty graph $H$, and add the following edges to it:
\begin{itemize}
    \item every edge $(a, b) \in E \cap (A \times B)$ where $w'_{ab} \in L_i$;
    \item every edge $(b, c) \in E \cap (B \times C)$ where $w'_{bc} \in L_j$;
    \item and every edge $(c, a) \in E \cap (C \times A)$ where $w'_{ca} \in L_k$.
\end{itemize}
Finally, we remove all vertices in $H$ whose degree is greater than $M\sqrt{n}$ for some sufficiently large constant $M$. 

Vassilevska Williams and Xu \cite{williams2020monochromatic} showed the followings about the reduction, and we omit their proofs for conciseness. 
\begin{claim}[Claim 3.5 and 3.6 in \cite{williams2020monochromatic}]
Suppose $G$ has a zero-weight triangle and fix any zero-weight triangle  $(a, b, c) \in A \times B \times C$ in $G$. Up to $0.01$ error probability, there exists an instance $H$ that contains it as a triangle. Also, up to $0.01$ error probability, listing $O(1)$ triangles per edge in $H$ finds at least one zero-weight triangle in $G$. 
\end{claim}

We then add enough isolated vertices to each $H$, so that the maximum degree of the graph becomes $\sqrt{n}$, instead of $M\sqrt{n}$. 
We then show that, the instance $H$ that contains the zero-weight triangle $(a, b, c)$, is quasirandom, up to $0.1$ error probability. 

We consider the expected number of labeled $4$-cycles $H$ contains. 

Let $\caC=(v_1, v_2, v_3, v_4)$ be any labeled $4$-cycle in $G$. For simplicity, let $v_{5} = v_1$. 
For each $i$, let $u_i = a$ if $v_i \in A$, $u_i = b$ if $v_i \in B$ and $u_i = c$ if $v_i \in C$. In order for $\caC$ to lie in $H$, it is necessary that for each $i \in [4]$, $$w'_{v_i, v_{i+1}}-w'_{u_i, u_{i+1}}=x \cdot (w_{v_i, v_{i+1}}-w_{u_i, u_{i+1}}) - y_{v_i}+y_{v_{i+1}} - y_{u_i}+y_{u_{i+1}}\in (- \lfloor p/\sqrt{n}\rfloor, \lfloor p/\sqrt{n}\rfloor).$$
Consider the following cases:
\begin{enumerate}
    \item $\caC$ shares exactly one vertex with $(a, b, c)$. Without loss of generality, assume $v_1 \in \{a, b, c\}$. Consider the list of random variables $\left(w'_{v_i, v_{i+1}}-w'_{u_i, u_{i+1}}\right)_{i=1}^{3}$. Each random variable is uniformly at random and independent to all previous variables, as $y_{v_{i+1}}$ is a fresh random variable that is added to the $i$-th random variable in the list. Thus, the probability that all of them are in $(- \lfloor p/\sqrt{n}\rfloor, \lfloor p/\sqrt{n}\rfloor)$ is $\le (2/\sqrt{n})^{3}$. The number of such labeled $4$-cycles $\caC$ is at most $\binom{4}{1}\cdot n^{3}$, so the expected number of them falling in $H$ is at most $2^{3}\cdot \binom{4}{1} \cdot n^{1.5} \le 32n^{1.5}$.
    \item $\caC$ shares exactly two or three vertices with $(a, b, c)$. First, assume $\caC$ shares two vertices with $(a, b, c)$. Without loss of generality, assume $v_1, v_t \in \{a, b, c\}$ for some $2 \le t \le 4$. Similar to the previous case, the random variables $\left(w'_{v_i, v_{i+1}}-w'_{u_i, u_{i+1}}\right)_{i \in [3] \setminus \{t-1\}}$ are independent and uniformly at random. Thus, the probability that $H$ contains $\caC$ is at most $(2/\sqrt{n})^{2}$, and the expected number of such labeled $4$-cycles in $H$ is $\left(\binom{4}{2} \cdot 2! \cdot n^2\right) (2/\sqrt{n})^2 \le 48 n$. Similarly, the expected number of labeled $4$-cycles in $H$ that share three vertices with $(a, b, c)$ is $\left(\binom{4}{3} \cdot 3! \cdot n\right)(2/\sqrt{n}) \le 48 n^{1/2}$. 
    \item $\caC$ does not contain any vertex in $(a, b, c)$, and $\caC$ is a zero-weight labeled $4$-cycle. Since $G$ is a weighted graph satisfying \cref{property:weighted-graphs},  the number of zero-weight labeled $4$-cycles is at most $n^3$. Using the same reason as the first case, each $4$-cycle is in $H$ with probability $(2/\sqrt{n})^{3}$, so the expected number of such labeled $4$-cycle in $H$ is $8n^{3/2}$.
    \item $\caC$ does not contain any vertex in $(a, b, c)$, and $\caC$ is not a zero-weight labeled $4$-cycle. In this case, we aim to show the random variables $\left(w'_{v_i, v_{i+1}}-w'_{u_i, u_{i+1}}\right)_{i=1}^{4}$ are independent. Equivalently, we could show $$\sum_{i=1}^4 \left(w'_{v_i, v_{i+1}}-w'_{u_i, u_{i+1}} \right),\left(w'_{v_i, v_{i+1}}-w'_{u_i, u_{i+1}}\right)_{i=1}^{3}$$
    are independent. The following two claims further simplifies the first term. 
    \begin{claim}
    $\sum_{i=1}^4 w'_{u_i, u_{i+1}} = 0$. 
    \end{claim}
    \begin{proof}
    We inductively show $\sum_{i=1}^t w'_{u_i, u_{i+1}} = w'_{u_1, u_{t+1}}$ for $0 \le t \le 4$, where $w'_{v, v}$ is defined to be $0$. The base case $t = 0$ is clearly true. Now suppose the equation is true for $t-1$, it suffices to show $w'_{u_1, u_t}+w'_{u_t, u_{t+1}} = w'_{u_1, u_{t+1}}$. Consider the following cases:
    \begin{itemize}
        \item $u_1 = u_t$. Then $w'_{u_1, u_t}+w'_{u_t, u_{t+1}} = 0 + w'_{u_1, u_{t+1}}=w'_{u_1, u_{t+1}}$.
        \item $u_1 \ne u_t$. If $u_{t+1} = u_1$, then $w'_{u_1, u_t}+w'_{u_t, u_{t+1}} = w'_{u_1, u_t}+w'_{u_t, u_1} = 0 = w'_{u_1, u_{t+1}}$. Otherwise, $\{u_1, u_t, u_{t+1}\} = \{a, b, c\}$. Thus, $w'_{u_1, u_t}+w'_{u_t, u_{t+1}} + w'_{u_{t+1}, u_1} = 0$ as $(a, b, c)$ is a zero-weight triangle. This implies $w'_{u_1, u_t}+w'_{u_t, u_{t+1}} = w'_{u_1, u_{t+1}}$.
    \end{itemize}
    Therefore, $\sum_{i=1}^4 w'_{u_i, u_{i+1}} = w'_{u_1, u_{5}} = w'_{u_1, u_1} = 0$. 
    \end{proof}
    \begin{claim}
    $\sum_{i=1}^4 w'_{v_i, v_{i+1}} = x \cdot \sum_{i=1}^4 w_{v_i, v_{i+1}}$. 
    \end{claim}
    \begin{proof}
    \begin{align*}
    \sum_{i=1}^4 w'_{v_i, v_{i+1}} &= \sum_{i=1}^4 \left(x \cdot w_{v_i, v_{i+1}} - y_{v_i} + y_{v_{i+1}}\right)\\
    &= -y_{v_1} + y_{v_{5}} + x \cdot \sum_{i=1}^4 w_{v_i, v_{i+1}}\\
    &= x \cdot \sum_{i=1}^4 w_{v_i, v_{i+1}}. \qedhere
    \end{align*}
    \end{proof}
    Let $W = \sum_{i=1}^4 w_{v_i, v_{i+1}}$. 
    After the simplification, we only need to show $\left(x \cdot W \right), \left(w'_{v_i, v_{i+1}}-w'_{u_i, u_{i+1}}\right)_{i=1}^{3}$ are independent. Since $\caC$ is not a zero-weight labeled $4$-cycle, $W \ne 0$, so $x \cdot W$ is uniformly at random. Each following variable $w'_{v_i, v_{i+1}}-w'_{u_i, u_{i+1}}$ for $i$ from $1$ to $3$ contains a fresh random variable $y_{v_{i+1}}$, so it is independent to all previous random variables. Thus, $\left(x \cdot W \right), \left(w'_{v_i, v_{i+1}}-w'_{u_i, u_{i+1}}\right)_{i=1}^{3}$ are independent, and so does $\left(w'_{v_i, v_{i+1}}-w'_{u_i, u_{i+1}}\right)_{i=1}^{4}$. Therefore, the probability that $C$ lies in $H$ is $\le (2/\sqrt{n})^4$. The number of labeled $4$-cycles in $G$ is bounded by $n^4$, so the expected number of such labeled $4$ cycles in $H$ is at most $16n^{2}$. 
\end{enumerate}
Overall, we have shown that the expected number of labeled $4$-cycles in $H$ is at most $(32+48+48+8+16)n^{2} \le 200 n^{2}$. By Markov's inequality, with error probability $0.1$, the graph has at most $2000n^{2}$ labeled $4$-cycles. 
By padding the graph with $O(n)$ isolated vertices, we obtain a graph with maximum degree at most $\sqrt{n}$ and at most $n^2$ $4$-cycles, as desired.

Suppose All-Edges $O(1)$-Triangle Listing can be solved in $T(n)$ time on $n$-vertex quasirandom graphs. Given an Exact Triangle instance on an $n$-vertex tripartite weighted graph satisfying \cref{property:weighted-graphs}, we run the above reduction to produce $O(n)$ All-Edges $O(1)$-Triangle Listing instances, and run the $T(n)$ time algorithm on each of the instances. It is possible that some instances are not on quasirandom graphs, so we need to stop the algorithm after $T(n)$ time even if it is still running. For each triangle the algorithm lists, we verify whether it is a zero-weight triangle in the original graph. By the above analysis, we have constant probability to find a zero-weight triangle if there is one. We can improve the success probability by repeating $O(\log n)$ times. Thus, if Exact Triangle  on  $n$-vertex tripartite  weighted graphs satisfying \cref{property:weighted-graphs} requires $n^{3-o(1)}$ time, All-Edges $O(1)$-Triangle Listing on $n$-vertex quasirandom graphs requires $n^{2-o(1)}$ time. 

Finally, it is known that All-Edges $O(1)$-Triangle Listing reduces to $\tO(1)$ instances of All-Edges Sparse Triangle, and each All-Edges Sparse Triangle instance is on a subgraph of the All-Edges $O(1)$-Triangle Listing instance \cite{williams2020monochromatic}. If the All-Edges $O(1)$-Triangle Listing is on a quasirandom graph, then so are the All-Edges Sparse Triangle instances. 
\end{proof}

Now we can immediately prove \cref{thm:ae-lb}:
\AllEdgesLowerBound*

\begin{proof}
Follows by combining \cref{thm:main}, \cref{lem:3sumsidon23sumconv}, \cref{lem:3sumconv2zwt}, and \cref{lem:zwt2ae}. 
\end{proof}

We  need the following standard lemma (e.g., \cite{sparsegraph}) before we prove \cref{cor:ae-lb}. 
\begin{lemma}
\label{lem:spectral}
Consider an undirected unweighted graph on $n$ vertices with maximum degree $d$ and let $C_k$ be the number of closed $k$-step walks. 
Then $C_{k} \le C_4 \cdot d^{k-4}$ for every $k \ge 4$. 
\end{lemma}
\begin{proof}
Let $A$ be the adjacency matrix of the graph. Then \[ C_{k} = \tr(A^k).\]

    Note that $A$ is a real-symmetric matrix, and let
  $\{\lambda_i\}$ be  the (real) eigenvalues of $A$.
 By Gershgorin disc theorem, all eigenvalues $\lambda_i$ satisfy \[|\lambda_i|\le \max_{i'} \sum_{j'}|A_{i',j'}| = \max_{i'}\deg(i')\le d.\] 
 Then,
    \[ \tr (A^{k}) = \sum_{i}  \lambda_i^{k} \le \max_i |\lambda_i|^{k-4} \sum_{i}  \lambda_i^{4} \le d^{k-4}\tr(A^{4}),\]
    and 
\[ C_{k} = \tr (A^{k})\le  d^{k-4}\tr (A^{4})  = C_4\cdot d^{k-4}. \qedhere\]
\end{proof}

Recall \cref{cor:ae-lb}:
\AllEdgesLowerBoundCor*
\begin{proof}
Let $G$ be an All-Edges Sparse Triangle instance on $n$-vertex quasirandom graphs. By \cref{thm:ae-lb}, solving All-Edges Sparse Triangle on $G$ requires $n^{2-o(1)}$ time under the $3$SUM hypothesis.  

By \cref{def:quasirandom}, the number of closed $4$-walks in $G$ is at most $n^2 + 2n^2 \le 3n^2$. 
Applying \cref{lem:spectral} with maximum degree $d\le \sqrt{n}$, we get that the number of closed $k$-walks in $G$ is at most $3 n^{k/2}$ for every $k \ge 4$, and so does the number of $k$-cycles. 
We can reduce the constant $3$ to $1$ by padding enough isolated vertices.

We then reduce the number of triangles to $n^{1.5}$, by adapting common techniques for witness listing~\cite{alon1992witnesses, seidel1995all}. 

By random color-coding \cite{alon1995color}, we can assume $G$ is tripartite on vertices $A \cup B \cup C$, and we are only required to report whether each edge between $A$ and $B$ are in a triangle. For $i$ from $\log n$ to $0$, we create a graph $G^i$ by randomly keeping each vertex in $C$ with probability $\frac{1}{2^i}$. For each $G^i$, we run an All-Edges Sparse Triangle algorithm to find if each $E(G^i) \cap (A \times B)$ is in a triangle. If an edge $(a, b)$ is found to be in a triangle, we delete it from $G$ (so it will not exist in $G^{i'}$ for any $i' \le i$ either). 

This algorithm is correct because in the final stage $i = 0$, $G^0 = G$, so will run an All-Edges Sparse Triangle algorithm on $G$, without only those edges between $A$ and $B$ that have been found in a triangle. 

Furthermore, each $G^i$ is a subgraph of $G$, so it contains at most $n^{k/2}$ $k$-cycles for $k \ge 4$. If an edge $(a, b) \in A \times B$ is in at least $D 2^i \log n$ triangles for some sufficiently large constant $D$, it is in a triangle in $G^{i'}$ for some $i' > i$ with high probability, so this edge is already deleted before we sample $G^i$. Therefore, we can assume all edges  $(a, b) \in A \times B$ before we sample $G^i$ are in at most $D 2^i \log n$ triangles. Thus, with high probability, each edge in $(A \times B) \cap E(G^i)$ is in at most $O(\log n)$ triangles, so $G^i$ contains $O(n^{1.5} \log n)$ triangles in total. 

By padding each $G^i$ with $O(n \log n)$ vertices, we can assume the number of triangles in each $G^i$ is at most $n^{1.5}$. 
\end{proof}
\section{Applications to Fine-Grained Complexity of Graph Problems}
\label{sec:graph}

In this section, we show our lower bounds for $4$-Cycle Enumeration, Approximate Distance Oracles, Approximate Dynamic Shortest Paths and Approximate All-Nodes Shortest Cycles, as applications of \cref{cor:ae-lb}. All the reductions start by following a random sampling step in \cite{AbboudBKZ22}, which we outline below.

\begin{lemma}
\label{lem:application-lemma}
Fix any constant $\sigma \in (0, 0.5)$, and any integer $k \ge 3$. Under the $3$SUM hypothesis, it requires $n^{2-o(1)}$ time to solve $n^{3\sigma}$ instances of All-Edges Sparse Triangle on tripartite graphs with $O(n^{1-\sigma})$ vertices and maximum degree $O(n^{0.5-\sigma})$, such that the total number of cycles of length at most $k$ over all instances is $O(n^{k/2-(k-3)\sigma})$ .
\end{lemma}
\begin{proof}

Fix an All-Edges Sparse Triangle instance on an $n$-node quasirandom graph $G = (V, E)$. By \cref{cor:ae-lb}, it requires $n^{2-o(1)}$ time under the $3$SUM hypothesis. 
By the standard color-coding technique~\cite{alon1995color}, we can assume $G$ is tripartite with three parts $A, B, C$. 

We partition the vertices of $G$ into $t=n^\sigma$ groups $A_1,\dots,A_t, B_1,\dots,B_t, C_1,\dots,C_t$, by  independently putting each vertex in $A$ into a uniformly random $A_i$, each vertex in $B$ into a uniformly random $B_i$, and each vertex in $C$ into a uniformly random  $C_i$. 
Then, it suffices to solve All-Edges Sparse Triangle on smaller instances induced by $A_i\cup B_j\cup C_\ell$ for all $(i,j,\ell)\in [t]\times [t]\times [t]$. Denote the instance by $G_{ij\ell}$.
By standard Chernoff bound, each $G_{ij\ell}$ has, with high probability, $\Theta(n/t)$ vertices and maximum degree $O(n^{0.5}/t)$. Also, each $G_{ij\ell}$ has at most $O(n^{k'/2}/t^{k'})$ $k'$-cycles in expectation for any $3 \le k' \le k$.
Then, the expected total number of $k'$-cycles across all instances is at most $O(n^{k'/2}/t^{k'-3})$. As $\sigma < 0.5$, $O(n^{k'/2}/t^{k'-3})$ is maximized when $k' = k$, so the expected number of cycles of length between $3$ and $k$ is $O(n^{k/2}/t^{k-3})$. Thus, with constant probability, the total number of cycles of length between $3$ and $k$ is $O(n^{k/2}/t^{k-3})$, and repeating the whole reduction $O(\log n)$ times boosts the success probability to $1/\poly(n)$. 
\end{proof}

In some applications, we are able to get a more refined bound by using an unbalanced version of \cref{lem:application-lemma}, as stated below. 
\begin{lemma}
\label{lem:application-lemma-unbalan}
Fix any constants $\sigma_a,\sigma_b \in (0, 0.5)$, and any integer $k \ge 3$. Under the $3$SUM hypothesis, it requires $n^{2-o(1)}$ time to solve $n^{2\sigma_a+\sigma_b}$ instances of All-Edges Sparse Triangle on tripartite graphs $G_i=(A_i\cup B_i\cup C_i,E_i)$ such that
\begin{itemize}
    \item $|A_i|=|C_i| = O(n^{1-\sigma_a}), |B_i|=O(n^{1-\sigma_b})$.
    \item Every vertex in $A_i\cup B_i\cup C_i$ has $O(n^{0.5-\sigma_a})$ neighbors in $A_i$ or $C_i$, and $O(n^{0.5-\sigma_b})$ neighbors in $B_i$.
    \item Summing over all instances $G_i$, the total number of cycles of length at most $k$ that use exactly one edge from $A_i\times C_i$ is $O(n^{k/2-(k-3)(\sigma_a+\sigma_b)/2})$.
    \item We only need to report whether each edge from $E_i \cap (A_i \times C_i)$ is in a triangle. 
\end{itemize}
\end{lemma}
When $\sigma_a,\sigma_b=\sigma$, \cref{lem:application-lemma-unbalan} is the roughly same as \cref{lem:application-lemma}. The proof of \cref{lem:application-lemma-unbalan} is almost the same as \cref{lem:application-lemma}, and we omit it for simplicity.

\subsection{\texorpdfstring{$4$}{4}-Cycle Enumeration}
Following \cite{AbboudBKZ22} we show tight $3$SUM hardness for the $4$-Cycle Enumeration problem, improving the bounds obtained by \cite{AbboudBKZ22}.

Recall \cref{thm:4cycle_lower}:

\FourCycleLower*
By a straightforward modification of the $O(\min\{n^2,m^{4/3}\})$-time $4$-cycle detection algorithm in \cite{AlonYZ97}, we obtain a $4$-cycle enumeration algorithm with the same pre-processing time $O(\min\{n^2,m^{4/3}\})$ and $O(1)$ delay. 
This algorithm is described in \cref{sec:4-cycle-enum-ub}. \cref{thm:4cycle_lower} shows that this running time is tight under $3$SUM hypothesis: the pre-processing time cannot be improved to $O(m^{4/3-\eps})$ or $O(n^{2-\eps})$, for any $\eps>0$.

\begin{proof}[Proof of \cref{thm:4cycle_lower}]
  The arguments follow \cite{AbboudBKZ22}, with only two differences: (1) we start from the stronger lower bound \cref{lem:application-lemma} (which in turn was implied by \cref{cor:ae-lb}), and (2) to avoid changing the graph density, we do not subdivide edges as \cite{AbboudBKZ22} did.

Suppose for the sake of contradiction that there is a $4$-Cycle Enumeration algorithm $\caA$ with $O(n^{2-\eps})$ pre-processing time and $n^{o(1)}$ delay on $n$-node graphs with $\lfloor 0.49 n^{1.5}\rfloor$ edges. 

We first apply \cref{lem:application-lemma} with $k=4$. 
For each small instance in \cref{lem:application-lemma}, which is a $3$-partite graph $G_3$ with vertex set $A \cup B\cup C$, create a $4$-partite graph $G_4$ with vertex set $A\cup B\cup C\cup C'$, in which $C'$ is a copy of $C$, $E(A,B),E(B,C),E(A,C')$ are copies of the edge sets $E(A,B),E(B,C),E(A,C)$ in graph $G_3$, and we add a perfect matching between $C,C'$ so that $c\in C, c'\in C'$ corresponding to the same vertex of $G_3$ are connected by an edge in $G_4$.   Note that a triangle in $G_3$ becomes a $4$-cycle in $G_4$, and observe that all the newly introduced $4$-cycles in $G_4$ must come from triangles in $G_3$.

We run the $4$-Cycle Enumeration algorithm $\caA$ on $G_4$, which has $n_0 = \Theta(n^{1-\sigma})$ vertices and $m_0 \le O(n^{1.5-2\sigma}) = O(n_0^{1.5}/n^{0.5\sigma})$ edges.  The pre-processing time is $O(n_0^{2-\eps})$.  (The edge density here is much smaller than the assumed density in the statement of \cref{thm:4cycle_lower}, but this can be easily fixed by padding a dense $4$-cycle free graph on $O(n_0)$ vertices, constructed in \cite{4cyclefree1,4cyclefree2})

The total pre-processing time across all $n^{3\sigma}$ instances is $O(n_0^{2-\eps}\cdot n^{3\sigma})$. The total time spent on outputting $4$-cycles is upper bounded by the total number of $4$-cycles and triangles across all instances, $O(n^{2-\sigma})$. Choosing $\sigma=\eps/4$, we get a subquadratic time algorithm for solving all the All-Edges Sparse Triangle instances produced by \cref{lem:application-lemma}, contradicting to the $3$SUM hypothesis.
\end{proof}

\subsection{Distance Oracles}

The following theorem follows by combining the approach in \cite{AbboudBKZ22} and \cref{lem:application-lemma}.

\begin{theorem}
\label{thm:approx-DO-lb-m}
Assuming the $3$SUM hypothesis, 
    for any constant integer $k \ge 3$ and $\eps, \delta>0$, there is no $O(n^{1+\frac{2}{k - 1} - \eps})$ time algorithm that can $(\lceil k / 2 \rceil - \delta)$-approximate the distances between $m$ given pairs of vertices in a given $n$-vertex $m$-edge undirected unweighted graph, where $m = \Theta(n^{1+\frac{1}{k-1}})$.  
\end{theorem}
\begin{proof}
Assume that such an algorithm $\caA$ exists for the sake of contradiction.

 We apply \cref{lem:application-lemma}. 
 For any instance $G=(A \cup B \cup C, E)$ produced by \cref{lem:application-lemma}, it suffices to test whether every edge in $(A \times C) \cap E$ is in a triangle. We first remove all edges between $A$ and $C$ from $G$ and call the new graph $G'$. Then we use $\caA$ to approximate the distances between $(a, c)$ on $G'$ if $(a, c)$ is an removed edge. The number of vertices $n_0$ in each instance is $O(n^{1-\sigma})$ and the number of edges $m_0$ in each instance is $O(n^{1.5-2\sigma})$. We will set $\sigma$ so that $m_0 = O(n_0^{1+\frac{1}{k-1}})$, and we can guarantee $m_0 = \Theta(n_0^{1+\frac{1}{k-1}})$ by padding a dense graph. Thus, this takes  $O((n^{1-\sigma})^{1+\frac{2}{k-1}-\eps})$ time per instance. 
 Finally, if the outputted distance between $a$ and $c$ is at most $2\lceil k/2\rceil - 1$, we check if $(a, c)$ is in a triangle in $G$ in time asymptotically bounded by the maximum degree of $G$, $O(n^{0.5-\sigma})$.
 
 \paragraph{Correctness. } If some edge $(a, c)$ is in a triangle $(a, b, c)$, and $(a, b, c) \in A \times B \times C$, running algorithm $\caA$ on graph $G'$ for query $(a, c)$ would return a distance at most $\lfloor 2 \cdot (\lceil k/2\rceil-\delta) \rfloor \le 2\lceil k/2\rceil - 1$, so we will check if $(a, c)$ is in a triangle in $G$. Therefore, the algorithm will find at least one triangle for each edge $(a, c)$ that is in a triangle. 
 
 \paragraph{Running time. } Note that every time $
 \caA$ outputs a distance at most $2\lceil k/2\rceil - 1$ for $(a, c)$ in $G'$, there must be a path between $a$ and $c$ of length at most $2\lceil k/2\rceil - 1$. Furthermore, since $G'$ is bipartite, every path between  $a$ and $c$ must have even length. Thus, the must be a path between $a$ and $c$ of length at most $2\lceil k/2\rceil-2 \le k-1$, 
 so the edge $(a, c)$ is in a cycle of length at most $k$ in $G$. Thus, the total number of checks is asymptotically bounded by the total number of cycles of length between $3$ and $k$, which is $O(n^{k/2-(k-3)\sigma})$. Therefore, the running time of the algorithm for handling all instances produced by \cref{lem:application-lemma} is 
 $$O\left(n^{3\sigma} \cdot (n^{1-\sigma})^{1+\frac{2}{k-1}-\eps} + n^{\frac{k}{2}-(k-3)\sigma} \cdot n^{0.5-\sigma}\right).$$
 Setting 
 $\sigma = \frac{k-3}{2(k-2)}+\frac{\eps}{4} < 0.5$ (as $\eps \le \frac{2}{k - 1}$) 
 gives a truly subquadratic  running time, which is impossible under the $3$SUM hypothesis by \cref{lem:application-lemma}. Also, we can verify $m_0 = O\left(n_0^{1+\frac{1}{k-1}}\right)$ as $1.5 -2\sigma \le (1-\sigma) (1 + \frac{1}{k-1})$. 
\end{proof}

This immediately implies \cref{thm:DO}, which we recall below:
\DOLower*

Note that the above lower bound even applies to distance oracles with $O(m^{\frac{1}{2k-1}-\eps})$ query time, similar to~\cite{AbboudBKZ22}.
Now, we use the unbalanced \cref{lem:application-lemma-unbalan} to get a better lower bound for offline distance oracles with subpolynomial query time.

\begin{theorem}
\label{thm:do-temp}
Assuming the $3$SUM hypothesis, 
    for any constant integer $k \ge 5$ and $\eps, \delta>0$, there is no
    $(\lceil k / 2 \rceil - \delta)$-approximate distance oracle with 
    $O(n^{1+\frac{2}{k - 1} - \eps})$ pre-processing time  and $n^{o(1)}$ query time for an $n$-vertex $O(n)$-edge undirected unweighted graph.
\end{theorem}
\begin{proof}
Assume that such an algorithm $\caA$ exists for the sake of contradiction. 
 
 We apply \cref{lem:application-lemma-unbalan} with some $0<\sigma_a<\sigma_b<0.5$ to be determined. 
 For any instance $G=(A \cup B \cup C, E)$ produced by \cref{lem:application-lemma-unbalan}, it suffices to test whether every edge in $(A \times C) \cap E$ is in a triangle. We first remove all edges between $A$ and $C$ from $G$ and call the new graph $G'$. Then we use $\caA$ to approximate the distances between $(a, c)$ on $G'$ if $(a, c)$ is an removed edge. The number of vertices $n_0$ in each instance is $O(2n^{1-\sigma_a} + n^{1-\sigma_b}) = O(n^{1-\sigma_a})$ and the number of edges $m_0$ in each instance is $O(n^{1.5-\sigma_a-\sigma_b})$, and for each instance we make $O(n^{1.5-2\sigma_a})$ approximate distance queries.
 We pad isolated vertices in each instance so that each instance has $\Theta(m_0)$ vertices, and run the pre-processing phase of $\caA$ on it in $O(m_0^{1+\frac{2}{k-1}-\eps}) = O(n^{(1.5-\sigma_a-\sigma_b)(1+\frac{2}{k-1}-\eps)})$ time per instance. Then we make $O(n^{1.5-2\sigma_a})$ queries to this distance oracle in $O(n^{1.5-2\sigma_a}\cdot n^{o(1)})$ time per instance.
 Finally, if the outputted distance between $a$ and $c$ is at most $2\lceil k/2\rceil - 1$, we check if $(a, c)$ is in a triangle in $G$ in time asymptotically bounded by the maximum number of neighbors in part $B$, which is $O(n^{0.5-\sigma_b})$.
 
 \paragraph{Correctness. } If some edge $(a, c)$ is in a triangle $(a, b, c)$, and $(a, b, c) \in A \times B \times C$, running algorithm $\caA$ on graph $G'$ for query $(a, c)$ would return a distance at most $\lfloor 2 \cdot (\lceil k/2\rceil-\delta) \rfloor \le 2\lceil k/2\rceil - 1$, so we will check if $(a, c)$ is in a triangle in $G$. Therefore, the algorithm will find at least one triangle for each edge $(a, c)$ that is in a triangle. 
 
 \paragraph{Running time. } Note that every time $
 \caA$ outputs a distance at most $2\lceil k/2\rceil - 1$ for $(a, c)$ in $G'$, there must be a path between $a$ and $c$ of length at most $2\lceil k/2\rceil - 1$. Furthermore, since $G'$ is bipartite, every path between  $a$ and $c$ must have even length. Thus, the must be a path between $a$ and $c$ of length at most $2\lceil k/2\rceil-2 \le k-1$, 
 so the edge $(a, c)$ is in a cycle of length at most $k$ in $G$ that uses exactly one edge from $A\times C$. Thus, the total number of checks is asymptotically bounded by the total number of cycles of length between $3$ and $k$ that use exactly one edge from $A\times C$, which is $O(n^{k/2-(k-3)(\sigma_a+\sigma_b)/2})$. Therefore, the running time of the algorithm for handling all instances produced by \cref{lem:application-lemma-unbalan} is 
 \begin{align*}
 &  O\left(n^{2\sigma_a+\sigma_b} \cdot \left ( (n^{1.5-\sigma_a-\sigma_b})^{1+\frac{2}{k-1}-\eps}  + n^{1.5-2\sigma_a} \cdot n^{o(1)}\right ) + n^{\frac{k}{2}-(k-3)(\sigma_a+\sigma_b)/2} \cdot n^{0.5-\sigma_b}\right)\\
  \le \ & O\left(n^{2\sigma_a+\sigma_b + (1.5-\sigma_a-\sigma_b)(1+\frac{2}{k-1}-\eps) } + n^{\sigma_b + 1.5 + o(1)} +  n^{\frac{k+1}{2}-(k-3)(\sigma_a+\sigma_b)/2-\sigma_b}\right).
 \end{align*}
 We set $\sigma_a = \frac{k-5+5\eps}{2(k-3)}  $ and $\sigma_b = 0.5 - \frac{\eps}{k-3} $ (assuming $\eps\in (0,0.1)$), and one can verify that $0<\sigma_a<\sigma_b<0.5$ and the above time complexity is truly subquadratic for all integers $k\ge 5$. This is impossible under the $3$SUM hypothesis by \cref{lem:application-lemma-unbalan}.
\end{proof}
\cref{thm:do-temp} immediately implies the following theorem.
\thmdo*
\subsection{Dynamic Shortest Paths}

\begin{theorem}
\label{thm:approx-sp-lb-m}
Assuming the $3$SUM hypothesis,  for any constant integer $k \ge 5$ and $\eps, \delta>0$, no algorithm can support insertion and deletion of edges and support querying $(\lceil k / 2 \rceil - \delta)$-approximate distance between two vertices in $O(n^{\frac{1}{k - 1} - \eps})$ time per update/query, after an $O(n^3)$ pre-processing, in $n$-vertex $m$-edge undirected unweighted graphs, where $m = \Theta(n^{1+\frac{1}{k-1}})$.  
\end{theorem}
\begin{proof}
Suppose such an algorithm $\caA$ exists. 
Then the proof is essentially the same as the proof of \cref{thm:approx-DO-lb-m}. The only difference is that, between two different instances of generated by \cref{lem:application-lemma}, we delete all edges of the old instance, and then add all edges of the new instance. 

The running time then becomes
 $$O\left((n^{1-\sigma})^3+ n^{3\sigma} \cdot n^{1.5-2\sigma} \cdot (n^{1-\sigma})^{\frac{1}{k-1}-\eps}+ n^{\frac{k}{2}-(k-3)\sigma} \cdot n^{0.5-\sigma}\right).$$
 Set 
 $\sigma = \frac{k-3}{2(k-2)}+\frac{\eps}{4} < 0.5$ (as $\eps \le \frac{2}{k - 1}$)
 gives a truly subquadratic  running time, which is impossible under the $3$SUM hypothesis by \cref{lem:application-lemma}. Note that in this theorem, we require $k \ge 5$ while in \cref{thm:approx-DO-lb-m} we do not, because we need $(n^{1-\sigma})^3$ to be subquadratic. 
\end{proof}

This immediately implies \cref{thm:DynamicSP}, which we recall below:
\DynamicSP*

\subsection{All-Nodes Shortest Cycles}

Finally, we show our lower bound for the All-Nodes Shortest Cycles problem. 

\ANSC*

\begin{proof}
Assume that such an algorithm $\caA$ exists for the sake of contradiction. Equivalently, say the running time of $\caA$ is $O(n^{1+\frac{2}{k-1}-\eps'})$ for some $\eps' > 0$.  
Similar as before, we apply \cref{lem:application-lemma}, but with parameter $k-1$.

 For any instance $G$ with vertex parts $A, B, C$ produced by \cref{lem:application-lemma}, we only need to report whether each edge between $A$ and $C$ are in a triangle, by symmetry. The number of vertices $n_0$ in each instance is $O(n^{1-\sigma})$ and the number of edges $m_0$ in each instance is $O(n^{1.5-2\sigma})$. We will set $\sigma$ so that $m_0 = O(n_0^{1+\frac{1}{k-1}})$, and we can guarantee $m_0 = \Theta(n_0^{1+\frac{1}{k-1}})$ by padding a dense graph. 
 
 We use $\caA$ to approximate the shortest cycle through every node. 
 If the outputted cycle length through $a \in A$ is at most $k-1$, 
 we check every pair of neighbors $b \in B, c \in C$ of $a$ to to find all triangles $(a, b, c)$ containing $a$ in $G$ and mark all edge $(a, c)$ in such a triangle. 
 This can be done in time asymptotically bounded by the square of the maximum degree of $G$, $O((n^{0.5-\sigma})^2)$ for each $a$. Finally, we report all marked edges as in triangles, and other edges as not in triangles. 
 
 \paragraph{Correctness. } If some edge $(a,c)$ is in a triangle $(a, b, c)$ for  $(a, b, c) \in A \times B \times C$, $\caA$ on graph $G$ on node $a$ would return a cycle of length at most $\lfloor 3 \cdot (k/3-\delta) \rfloor \le k-1$, so we will find all triangles containing $a$  in $G$. Therefore, the algorithm will report $(a, c)$ as in a triangle. Also, the algorithm clearly does not have false positives. 
 
 \paragraph{Running time. } Note that every time $
 \caA$ outputs a cycle of length at most $k-1$ through $a$ in $G$, there must be a cycle through $a$ of length at most $k-1$. Thus, the total number of checks is asymptotically bounded by the total number of cycles of length between $3$ and $k-1$. Therefore, the running time of the algorithm over all instances produced by \cref{lem:application-lemma} is 
 $$O\left(n^{3\sigma} \cdot (n^{1-\sigma})^{1+\frac{2}{k-1}-\eps'} + n^{\frac{k-1}{2}-(k-4)\sigma} \cdot \left (n^{0.5-\sigma}\right)^2\right).$$

 Setting $\sigma = \frac{k-3}{2(k-2)} + \frac{\eps}{4} < 0.5$ gives a truly subquadratic  running time, which is impossible under the $3$SUM hypothesis by \cref{lem:application-lemma}. Similar as before, we indeed have $m_0 = O(n_0^{1+\frac{1}{k-1}})$.   
\end{proof}
\section{\texorpdfstring{$4$}{4}-Cycle Enumeration Algorithms}
\label{sec:4-cycle-enum-ub}
In this section, we present algorithms for the  $4$-Cycle Enumeration problem, and prove \cref{thm:4cycle_upper}.
Recall that in the $4$-Cycle Enumeration problem, we need to first pre-process a given simple undirected graph, and then enumerate all the $4$-cycles in this graph with subpolynomial time delay for every $4$-cycle enumerated. 

We adapt the known algorithms for $4$-cycle detection  \cite{yuster1997finding,AlonYZ97}, so that they actually find all the $4$-cycles in the graph.
 We show that if the algorithm has been running for $t$ time, where $t\ge T(n,m)$, then it must have found at least $ct$ $4$-cycles so far, for some constant $0<c<1$. 
 Then, we can use a standard trick to convert it into an enumeration algorithm with pre-processing time $T(n,m)$ and worst-case $O(1)$ delay: use a buffer to store the found $4$-cycles that have not been outputted. After $T(n,m)$ time, we start to output the $4$-cycles from the buffer with a delay of $1/c$ time steps each, which ensures that the buffer does not  become empty until all $4$-cycles have been outputted.

We start with the simpler $O(n^2)$ time algorithm.

\begin{theorem}
\label{thm:4-cycle-enum-dense}
After $O(n^2)$ time pre-processing, we can enumerate $4$-cycles with $O(1)$ delay in an $n$-vertex undirected graph. The algorithm is deterministic.
\end{theorem}
\begin{proof}
Let  the input graph be $G=(V, E)$. 
    Consider the following algorithm. First, we initialize an $n\times n$ table, indexed by $V \times V$. Each entry of the table $L[u, w]$ is a list of vertices, initially empty. Then, for every $v \in V$ in increasing order, we enumerate all pairs of distinct neighbors $u, w$ of $v$ where $u < w$. For every $(v, u, w)$, we enumerate all $v' \in L[u, w]$, and output $(v, u, w, v')$ as a $4$-cycle (if it has not been outputted before, which can be verified by checking $v < \max\{u, w\}$). Then, we add $v$ to $L[u, w]$. 
    
    It is not difficult to verify that this algorithm eventually outputs all $4$-cycles of the graph and each $4$-cycle will be found $2$ times. 
    
    After enumerating the first $x$ triples $(v, u, w)$, the number of $4$-cycles the algorithm has outputted is at least 
    \begin{align*}
        \frac{1}{2}\sum_{u, w} \binom{|L[u, w]|}{2} &\ge  \frac{1}{4}\sum_{u, w} |L[u, w]|^2 - \frac{1}{4}\sum_{u, w} |L[u, w]|\\
        &\ge \frac{1}{4}\left( \left(\frac{x}{n^2}\right)^2 \cdot n^2 - x\right) \tag{Jensen's inequality}\\
        &= \frac{x}{4}\left(\frac{x}{n^2}-1\right).
    \end{align*}
    Therefore, after enumerating the first $5n^2$ triples $(v, u, w)$, the number of $4$-cycles the algorithm finds is at least the number of triples $(v, u, w)$ enumerated. 
    
    Therefore, after an $O(n^2)$ pre-processing, the algorithm outputs $4$-cycles in $O(1)$ amortized delay. This amortized delay can be easily turned into worst-case by standard tricks described earlier.
\end{proof}

Now we describe the algorithm with $O(m^{4/3})$ pre-processing time, which is better for sparse graphs.
\begin{theorem}
\label{thm:4-cycle-enum-sparse}
After $O(m^{4/3})$ time pre-processing, we can enumerate $4$-cycles with $O(1)$ delay in an $m$-edge undirected graph. The algorithm is deterministic.
\end{theorem}

\begin{proof}
We first describe an algorithm that requires a hash table.	In the end we remove this requirement.

We use the well-known supersaturation property of $4$-cycles: in an $n$-vertex $m$-edge graph, the number of $4$-cycles is at least $\max\{0, c_0 (m/n)^4 - m^{4/3}\}$, for some constant $c_0>0$ (e.g., see \cite[Lemma 2.4]{bringmann2019truly} for a proof of the bipartite graph case, which can be adapted to general graphs).

Given a graph $G$ with $n$ vertices and $m$ edges, we repeatedly peel off the vertex with minimum degree, and obtain an ordering of the vertices, $v_1,\dots,v_n$. 
Let $G_i$ denote the induced subgraph $G[\{v_i,\dots,v_n\}]$. Then, by definition, $v_i$ has the minimum degree in $G_i$, and we denote this degree by $d_i$. Observe that $d_i - 1\le d_{i+1}$. 

 The number of edges in $G_i$ is $m_i = d_i+d_{i+1}+\dots+d_n$, and satisfies $m_i\ge d_i (n-i+1)/2 $.
 Thus, the number of $4$-cycles in $G_i$ is at least \[c_0\big (m_i/(n-i+1)\big )^4 - m_i^{4/3} \ge  cd_i^4 - m_i^{4/3},\] where $c=c_0/16$.

Define \[d_i^+ := \max_{i\le i' \le n} d_{i'},\]
and 
\[ i^+ := \argmax_{i\le i' \le n} d_{i'}. \]
 Then, the number of $4$-cycles in $G_i$ is at least the number of $4$-cycles in $G_{i^+}$, which is at least 
 \begin{equation}
    \label{eqn:4cycnumbound}
    c(d_i^+)^4 - m_i^{4/3}.
 \end{equation}

Starting from the empty graph $G_n$, we proceed in $(n-1)$ rounds:  in the $(n-i)$-th round, we add a new vertex $v_i$, and by the end of this round we will have reported all $4$-cycles in $G_i$. When adding $v_i$, we do the following: 
\begin{itemize}
   \item  Enumerate all $2$-paths $(v_i,  v_j, v_k)$  such that $i<j<k$, and add them to the bucket indexed by $(i,k)$.  There are $\sum_{j\in N(i),j>i} d_j \le d_i\cdot d_{i+1}^+$ such $2$-paths.
   \item Enumerate all $2$-paths  $(v_j,v_i,v_k)$ such that $i<j<k$, and add them to the bucket indexed by $(j,k)$. There are $\binom{d_i}{2}$ such $2$-paths.
\end{itemize}
Here, a bucket indexed by $(x,y)$ ($x<y$) is  a linked list containing several $2$-paths with $x,y$ as endpoints. The head pointer of this list is stored in the $(x,y)$ entry of a hash table.

    Every time we insert a $2$-path $(x,z,y)$ into bucket $(x,y)$, we find new $4$-cycles by combining it with other $2$-paths $(x,z',y)$ in this bucket.  As noticed by \cite{AlonYZ97}, this actually allows us to find all the $4$-cycles: any $4$-cycle with vertex set $\{u_1,u_2,u_3,u_4\}$ (sorted so that $u_1,u_2,u_3,u_4$ appear as a subsequence in $v_1,\dots,v_n$ from left to right) has three possible cases:
	\begin{itemize}
		\item $(u_1,u_2,u_4,u_3,u_1)$: found by combining $2$-paths $(u_1,u_2,u_4)$ and $(u_1,u_3,u_4)$.  
		\item $(u_1,u_2,u_3,u_4,u_1)$: 
        found by combining $2$-paths $(u_2,u_1,u_4)$ and $(u_2,u_3,u_4)$.  
		\item $(u_1,u_3,u_2,u_4,u_1)$:
        found by combining $2$-paths $(u_3,u_1,u_4)$ and $(u_3,u_2,u_4)$.  
	\end{itemize}

The total number of 2-paths found (which bounds the time spent so far) in the first $(n-i)$ rounds is at most 
\begin{align}
 & \sum_{j=i}^n \left( d_j\cdot d_{j+1}^+ +\binom{d_j}{2} \right) \nonumber \\
  \le \ &  \sum_{j=i}^n\left(  d_j\cdot d_{j+1}^+ + d_j \cdot d_{j+1}^+/2\right) \tag{since $d_j-1\le d_{j+1} \le d_{j+1}^+$}\nonumber \\
\le \ &  \frac{3}{2}(d_i+\dots +d_n)\cdot d_{i+1}^+ \nonumber  \\
= \ &  \frac{3}{2}m_i\cdot d_{i+1}^+.  
    \label{eqn:enum-bound}
\end{align}

By the end of the first $(n-i-1)$ rounds, we must have found all the $4$-cycles in $G_{i+1}$. Recall that $G_{i+1}$ contains at least $c(d_{i+1}^+)^4 - m_i^{4/3}$ $4$-cycles (\cref{eqn:4cycnumbound}).
Hence, at any point during the execution of the algorithm (say, during the $(n-i)$-th round), the number of $4$-cycles found so far is at least
\begin{align*}
	  c (d_{i+1}^+)^4 - m_i^{4/3}   & \ge    3m_i^{4/3} + c(d_{i+1}^+)^4-4m^{4/3}\\
	  & \ge  c^{1/4} m_i\cdot d_{i+1}^+-4m^{4/3} \tag{AM-GM}\\
	  & \ge \Omega(\text{\cref{eqn:enum-bound}})-4m^{4/3}.
\end{align*}
Therefore, after $O(m^{4/3})$ time, the number of $4$-cycles found is at least a constant fraction of the time spent. This implies the desired enumeration algorithm by previous discussions.
\\

To avoid using a hash table (which requires randomization, or has worse lookup/insertion time deterministically), we can use the following off-line strategy:  instead of doing $(n-1)$ rounds, we do $k=O(\log n)$ rounds only. Pick $1=i_1<i_2<\dots < i_k = n'$, where $n' = \max \{1\le n'\le n-1 : d_{n'+1}^+ \neq 0 \}$, such that
    \begin{equation}
        \label{eqn:lrbound}
        2\le \frac{ m_{i_j}\cdot d_{i_j+1}^+  } { m_{i_{j+1}}\cdot d_{i_{j+1}+1}^+} < 12
    \end{equation}
    for all $1\le j\le k-2$.
To see why this is possible, note that for all $1\le i <n'$,
\begin{align*}
    1 \le \frac{m_{i}\cdot d_{i+1}^+ }{ m_{i+1}\cdot d_{i+2}^+} & =\big  (1 + \frac{d_i}{m_{i+1}}\big )\cdot \frac{d_{i+1}^+}{d_{i+2}^+}\\
    & \le  \big  (1 + \frac{d_{i+1}+1}{m_{i+1}}\big )\cdot \frac{d_{i+2}^+ + 1}{d_{i+2}^+}\\
    & \le 6,
\end{align*}
so we can iteratively pick $i_{j+1}$ to be the maximum $i \in (i_j ,n']$ such that $\frac{ m_{i_j}\cdot d_{i_j+1}^+  } { m_{i}\cdot d_{i+1}^+}<12$, and it must satisfy \cref{eqn:lrbound}, unless the final step $i_{j+1}=n'$ is reached.

Then, in the $(k-j+1)$-th round $(1\le j\le k)$, we find all $4$-cycles in $G_{i_{j}}$ from scratch, as follows: run the previous algorithm up to $v_{i_j}$, but instead of inserting each $2$-path into the bucket in real time, we wait until the end, and then use linear-time radix sort to group them into buckets.
Then we use these buckets to find all the $4$-cycles in $G_{i_{j}}$ as before (excluding those already appearing in $G_{i_{j+1}}$).
  Similar to previous analysis, at any point during the $(k-j+1)$-th round, we must have found at least $c^{1/4} m_{i_{j+1}}\cdot d_{i_{j+1}+1}^+-4m^{4/3}$ $4$-cycles, and the time spent so far on finding $2$-paths (including previous rounds) is at most
\begin{align*}
     \sum_{j'=j}^k \frac{3}{2} m_{i_{j'}} \cdot d_{i_{j'}+1}^+ & \le  O(m_{i_{j}} \cdot d_{i_{j}+1}^+) \tag{by left part of \cref{eqn:lrbound} and a geometric sum}\\
      & \le O(m_{i_{j+1}} \cdot d_{i_{j+1}+1}^+). \tag{by right part of \cref{eqn:lrbound}}
\end{align*}
The extra time for radix sort is $O(n)$ time per round, which is dominated by the $m^{4/3}$ term. 
Hence, similar to previous  arguments, this implies a deterministic $4$-Cycle Enumeration algorithm with $O(m^{4/3})$ pre-processing time and $O(1)$ delay.
\end{proof}

\cref{thm:4-cycle-enum-dense} and \cref{thm:4-cycle-enum-sparse} together imply \cref{thm:4cycle_upper}.
In particular, for $m=\Theta(n^{1.5})$, the pre-processing times of the two algorithms match. In fact, $m=\Theta(n^{1.5})$ is the hardest sparsity for $4$-Cycle Enumeration on $n$-vertex graphs (besides $m = \Theta(n^2)$).
For $m \ll n^{1.5}$, the pre-processing time is $O(m^{4/3}) \ll n^2$.
For $m = n^{1.5+\delta}$,  we have the following theorem.
\begin{theorem}
   On an $n$-vertex $m$-edge simple undirected graph with $m \ge n^{1.5+\delta}$ (for some constant $\delta >0$), there is a randomized algorithm that enumerates $4$-cycles in this graph with $O(1)$ delay, after $O(m+n^{2-2\delta})$ time pre-processing.
\end{theorem}
\begin{proof}[Proof Sketch]
The supersaturation property states that the input graph $G$ has at least $\Omega( (m/n)^4 ) = \Omega(n^{2+4\delta})$ $4$-cycles. If we subsample a set of $n' = n\cdot n^{-\delta}$ vertices, the induced subgraph $G'$ still contains $\Omega(n^2)$ $4$-cycles in expectation. In fact, by a standard variance bound and Chebyshev's inequality, $G'$ has $m'=\Omega(n^{1.5-\delta}) \gg (n')^{1.5}$ edges with at least $0.99$ probability, so it contains at least $\Omega((m'/n')^4) \ge \Omega(n^2)$ $4$-cycles. 
This buys us enough time to pre-process the whole graph $G$. 

More specifically, we use the algorithm in \cref{thm:4-cycle-enum-dense} to pre-process the subgraph $G'$ in $O(n'^2)=O(n^{2-2\delta})$ time.
Then, while enumerating $4$-cycles in the $G'$, we can start running the pre-processing step of \cref{thm:4-cycle-enum-dense} on the whole graph $G$, and adding the enumerated $4$-cycles to a buffer (instead of outputting them immediately). 
 Moreover, for each $4$-cycle enumerated from $G$, we check whether it is in $G'$; if so, we throw it away. Once we have enumerated all the $4$-cycles in $G'$ (say there are $T\ge \Omega(n^2)$ of them), the algorithm on $G$ should have enumerated $(T+1)$ $4$-cycles. Then we switch to outputting $4$-cycles from the buffer, and enumerate an additional one from $G$ after each output. 
This way, the buffer is always nonempty until we enumerate all $4$-cycles in the graph, because we throw away at most $T$ $4$-cycles.  Thus, the algorithm runs in $O(n^{2-2\delta})$ pre-processing time and enumerates $4$-cycles with $O(1)$ delay. 
\end{proof}
\section{Consequences of the Strong \texorpdfstring{$3$}{3}SUM Hypothesis}
\label{sec:strong-3sum}

Finally, we show our conditional lower bound of Triangle Detection under the Strong $3$SUM hypothesis. Recall the theorem:
\TriangleDet*
\begin{proof}
Given a $3$SUM instance $A$ with input range $U = n^2$, we construct a Triangle Detection instance as follows. 

Let $p, q, r$ be three distinct primes in $[2n^{2/3} \log^2 n, 4n^{2/3} \log^2 n]$, sampled uniformly at random. Let $A'$ be a copy of $A$ but we take the modulo of every number by $pqr$. 
Let $U$ (resp. $V, W$) be a vertex set identified by all numbers in $[0, pqr)$ that are congruent to $0$ mod $p$ (resp. $q, r$). We add an edge between $u \in U$ and $v \in V$ if $(v - u) \bmod{pqr} \in A'$. Note that these edges can be added efficiently: for each $a \in A'$, if $v - u \equiv a \pmod{pqr}$, then $v - u \equiv a \pmod{q}$. Also, by construction, $u \equiv 0 \pmod{p}$ and $v \equiv 0 \pmod{q}$, which imply $u \equiv -a \pmod{q}$. Therefore, given $a$, the values of $u \bmod{p}$ and $u \bmod{q}$ are fixed. We can enumerate all possible values of $u \bmod{r}$, each (together with $u \bmod{p}$ and $u \bmod{q}$) determines a unique value of $u$ by the Chinese remainder theorem. For each such $u$, we add an edge between it and $(a + u) \bmod{pqr} \in V$. It thus takes $\tO(n \cdot r) = \tO(n^{5/3})$ time to add all the edges. Similarly, we  add
an edge between $v \in V$ and $w \in W$ if $(w - v) \bmod{pqr} \in A'$, and an edge between $w \in W$ and $u \in U$ if $(u-w)  \bmod{pqr} \in A'$. Finally, we remove all vertices in the graph whose degree is larger than $C n^{1/3}$ for some sufficiently large constant $C$.  

If this graph has a triangle $(u, v, w)$, then $(v-u) \bmod{pqr}, (w-v) \bmod{pqr}, (u-w) \bmod{pqr} \in A'$ forms a $3$SUM solution in $\mathbb{Z}_{pqr}$, so $A$ must have a $3$SUM solution as well since $pqr > 3U$. For the other direction, suppose $A$ has a $3$SUM solution $x+y+z \equiv 0 \pmod{pqr}$. Consider vertices $u \in U, v \in V, w \in W$ such that 
\[
\begin{array}{ccc}
	u \equiv 0 \pmod{p}, & u \equiv -x \pmod{q}, & u \equiv z \pmod{r};\\
	v \equiv x \pmod{p}, & v \equiv 0 \pmod{q}, & v \equiv -y \pmod{r};\\
	w \equiv -z \pmod{p}, & w \equiv y \pmod{q}, & w \equiv 0 \pmod{r}.
\end{array}
\]
Note that $v-u \equiv x \pmod{pqr}, w-v \equiv y \pmod{pqr}, u - w \equiv z \pmod{pqr}$, so $(u, v, w)$ forms a triangle before we remove high-degree vertices in the graph. Thus, it suffices to bound the probability that this triangle is removed. By previous discussions, the number of neighbors of $u$ in $V$ is the number of $a \in A$ where $u \equiv -a \pmod{q}$, which implies $a - x\equiv 0 \pmod{q}$. Note that $q$ is a uniformly random prime from $[2n^{2/3} \log^2 n, 4n^{2/3} \log^2 n]$, which contains $\Theta(n^{2/3} \log n)$ primes by the prime number theorem. Also, $a-x$ has $\Theta(\log n)$ distinct prime factors, so the probability $a - x\equiv 0 \pmod{q}$ is $O(n^{-2/3})$. Therefore, if we use $\deg(s, T)$ to denote the number of neighbors of $s$ in set $T$, then 
$\Ex[\deg(u, V)] \le O(n^{1/3})$. We can similarly bound $\Ex[\deg(u, W)]$, so $\Ex[\deg(u)] \le O(n^{1/3}) \le \frac{C}{10} \cdot n^{1/3}$ for some sufficiently large constant $C$. By Markov's inequality, $\Pr[\deg(u) > C n^{1/3}] \le \frac{1}{10}$, so we remove vertex $u$ with probability at most $\frac{1}{10}$. Similarly, we remove vertex $v$ and $w$ with probability at most $\frac{1}{10}$. By union bound, the triangle $(u, v, w)$ will remain in the final graph with probability at least $\frac{7}{10}$. 

We can repeat this reduction $O(\log n)$ times to boost the success probability. 
The number of vertices $N$ in the graph is $\tilde{\Theta}(n^{4/3})$, and the maximum degree is $O(n^{1/3}) = O(N^{1/4})$, so the triangle detection instance requires $n^{2-o(1)} = N^{1.5-o(1)}$ time under the Strong $3$SUM hypothesis. 
\end{proof}

\section{Open Questions}
\label{sec:open}
We conclude with several related open questions.
\begin{enumerate}
    \item To show the $3$SUM-hardness of Sidon Set Verification,  our reduction crucially relies on the efficient self-reduction of $3$SUM, which is not known to exist for $k$SUM with $k\ge 4$.
    Does there exist a more powerful reduction that can show $4$SUM-hardness for Sidon Set Verification? 

\item On the upper bound side, can we solve Sidon Set Verification faster than $4$SUM? 
Observe that Sidon Set Verification on $n$ integers can be easily reduced to the Element Distinctness problem on $\binom{n}{2}$ integers, but the same strategy only reduces $4$SUM to the \emph{bichromatic} version of Element Distinctness (or, List Disjointness problem). In the low-space setting, Element Distinctness appears much easier than its bichromatic version \cite{BCM13,BansalGN018,ChenJWW22,xinlyu}, and known results imply
$\polylog(n)$-space algorithms solving Sidon Set Verification in $\tilde O(n^3)$ time, and  $4$SUM  in $\tilde O(n^{3.5})$ time (assuming distinct input integers). 
 Can we obtain a speedup for Sidon Set Verification  in the standard setting as well?
    \item Our $\Omega(m^{4/3-\eps})$ lower bound for $4$-Cycle Enumeration also applies to $2k$-Cycle Enumeration with larger $k$, by a reduction described in \cite{AbboudBKZ22}.
    However, it is natural to conjecture that the correct bound should increase with $k$, analogous to the situation with $2k$-Cycle \emph{Detection} algorithms \cite{AlonYZ97,DahlgaardKS17}.
    Can we show better lower bounds for  $2k$-Cycle Enumeration when $k>2$?
    \item  Our techniques improve previous lower bounds \cite{AbboudBKZ22} on the  pre-processing time of $k$-approximate Distance Oracles, but the dependency on $k$ still does not match the best known upper bounds. 
    Specifically for $3$-approximation, the Thorup-Zwick distance oracle \cite{thorup2005approximate} has $O(m\sqrt{n})$ preprocessing time and $O(1)$ query time.
    Can we prove a tight lower bound?
    \item 
Can we show conditional hardness for Sidon Set Verification on $n$ input integers from the input range $[n^{2+\delta}]$, for some $0< \delta < 1$? 

Note that one of the steps in our reduction for input range $[n^{3+\delta}]$ resembles the random construction of $\Omega(N^{1/3})$-size Sidon sets from $[N]$ via the probabilistic method.\footnote{The construction is as follows: uniformly independently sample $m$ integers from $[N]$, which produce $s = \Theta(m^4/N)$ Sidon 4-tuples in expectation. Remove one integer from each Sidon 4-tuple, and the remaining integers form a Sidon set of expected size $m-s \ge m/2$, by setting $m = c\cdot N^{1/3}$ for a small enough constant $c$.} In the proof of \cref{thm:main-smallinput} (more spefically, the second bullet point), when we reduce the input range of $3$SUM instances by hashing modulo a random prime $p$, we introduce additional Sidon 4-tuples in the instances, which will be removed later. By choosing $p$ to be slightly super-cubic in the size of the instance, the number of removed integers is small, and the remaining integers in the instance indeed form a large Sidon set.

If we want to improve the input range, the reduction would have to produce Sidon sets of size $\gg N^{1/3}$. 
Such density is usually achieved using algebraic constructions rather than simple probabilistic method, so such a reduction would have to use drastically different strategies.
\end{enumerate}
 \section*{Acknowledgement}
We would like to thank Virginia Vassilevska Williams for many helpful discussions and suggestions.  We also thank Ryan Williams for helpful discussions.
	\bibliographystyle{alphaurl} 
	\bibliography{main}
\end{document}